\documentclass[11pt]{article}

\usepackage[pagebackref,colorlinks]{hyperref}
\usepackage{amsmath} % AMS Math Package
\usepackage{amsthm} % Theorem Formatting
\usepackage{thmtools}
\usepackage{amssymb}	% Math symbols such as \mathbb
\usepackage{graphicx} % Allows for eps images
\usepackage{multicol} % Allows for multiple columns
\usepackage{multirow}
\usepackage{color}
\usepackage{bm}
\usepackage[dvips,letterpaper,margin=1in,bottom=1in]{geometry}
\usepackage[capitalize,noabbrev]{cleveref}

\usepackage[utf8]{inputenc}
\usepackage[english]{babel}

\usepackage{mathtools}

\usepackage{adjustbox}

\newtheorem{theorem}{Theorem}[section]

\newtheorem{lemma}[theorem]{Lemma}
\newtheorem{corollary}[theorem]{Corollary}

\newtheorem{definition}[theorem]{Definition}

\newtheorem{remark}[theorem]{Remark}

\newcommand{\braket}[2]{\left< #1 \vphantom{#2} \middle| #2 \vphantom{#1} \right>} % for Dirac brackets

\DeclarePairedDelimiter\rbra{\lparen}{\rparen}
\DeclarePairedDelimiter\sbra{\lbrack}{\rbrack}
\DeclarePairedDelimiter\cbra{\{}{\}}
\DeclarePairedDelimiter\abs{\lvert}{\rvert}
\DeclarePairedDelimiter\Abs{\lVert}{\rVert}
\DeclarePairedDelimiter\ceil{\lceil}{\rceil}

\DeclarePairedDelimiter\ket{\lvert}{\rangle}
\DeclarePairedDelimiter\bra{\langle}{\rvert}

\newcommand{\tr} {\operatorname{tr}}
\newcommand{\poly} {\operatorname{poly}}
\newcommand{\diag} {\operatorname{diag}}

\newcommand{\SV} {\text{SV}}

\newcommand{\bb}[1]{#1^{\diamond}}

\usepackage{algorithm}
\usepackage{algpseudocode}
\usepackage{tabularx}
\usepackage{booktabs}
\usepackage{threeparttable}

\newcommand{\footremember}[2]{%
    \footnote{#2}
    \newcounter{#1}
    \setcounter{#1}{\value{footnote}}%
}

\usepackage{tikz}
\usetikzlibrary{quantikz2}

\title{Quantum Lower Bounds by Sample-to-Query Lifting}
\author{
    Qisheng Wang \footremember{1}{Qisheng Wang is with the School of Informatics, University of Edinburgh, EH8 9AB Edinburgh, United Kingdom (e-mail: \href{mailto:QishengWang1994@gmail.com}{\nolinkurl{QishengWang1994@gmail.com}}). Part of the work was done when the author was with the Graduate School of Mathematics, Nagoya University, Nagoya 464-8602, Japan.}
    \and Zhicheng Zhang \footremember{2}{Zhicheng Zhang is with the Centre for Quantum Software and Information, University of Technology Sydney, Ultimo, NSW 2007, Australia (e-mail: \href{mailto:iszczhang@gmail.com}{\nolinkurl{iszczhang@gmail.com}}).}
}
\date{}

\begin{document}

\maketitle

\begin{abstract}
    The polynomial method by \hyperlink{cite.BBC+01}{Beals, Buhrman, Cleve, Mosca, and de Wolf (FOCS 1998, \textit{J.\ ACM} 2001)}, the adversary method by \hyperlink{cite.Amb02}{Ambainis (STOC 2000, \textit{J.\ Comput.\ Syst.\ Sci.}\ 2002)}, and the compressed oracle method by \hyperlink{cite.Zha19}{Zhandry (CRYPTO 2019)} have been shown to be powerful in proving quantum query lower bounds for a wide variety of problems.
In this paper, we propose a new method for proving quantum query lower bounds by a quantum sample-to-query lifting theorem, which is from an information theory perspective. 
Using this method, we obtain the following new results:
\begin{enumerate}
    \item A quadratic relation between quantum sample and query complexities regarding quantum property testing, which is optimal and saturated by quantum state discrimination. 
    Here, the sample complexity is measured given sample access to the quantum state to be tested, while the query complexity is measured given query access to an oracle that block-encodes the quantum state. 
    \item A matching lower bound $\widetilde \Omega(\beta)$ for quantum Gibbs sampling at inverse temperature $\beta$,\footnote{$\widetilde\Omega\rbra{\cdot}$ suppresses logarithmic factors.} showing that the quantum Gibbs sampler by \hyperlink{cite.GSLW19}{Gily\'en, Su, Low, and Wiebe (STOC 2019)} is optimal. 
    \item A new lower bound $\widetilde \Omega\rbra{1/\sqrt{\Delta}}$ for the entanglement entropy problem with gap $\Delta$, which was recently studied by \hyperlink{cite.SY23}{She and Yuen (ITCS 2023)}. 
    \item A series of quantum query lower bounds for matrix spectrum testing, based on the sample lower bounds for quantum state spectrum testing by \hyperlink{cite.OW21}{O'Donnell and Wright (STOC 2015, \textit{Comm.\ Math.\ Phys.}\ 2021)}. 
\end{enumerate}
In addition, we also provide unified proofs for some known lower bounds that have been proven previously via different techniques, 
including those for phase/amplitude estimation and Hamiltonian simulation. 
\end{abstract}


\newpage
\tableofcontents
\newpage

\section{Introduction}
Quantum query algorithms typically assume access to a quantum unitary oracle $U$, and the corresponding computational complexity, known as the \textit{quantum query complexity}, is measured by the number of queries to $U$ used in the algorithms. 
Well-known examples include Grover's search algorithm \cite{Gro96} that finds a solution among $N$ unstructured items with quantum query complexity $O\rbra{\sqrt{N}}$, and quantum phase estimation \cite{Kit95} which is a key step of Shor's factoring algorithm \cite{Sho94}. 

To understand the limitations of quantum computing, it was shown in \cite{BBBV97,BBHT98,Zal99} that any quantum search algorithm over $N$ items requires $\Omega\rbra{\sqrt{N}}$ queries, which implies the optimality of Grover's algorithm.
Motivated by the quantum lower bound for search problem, the \textit{polynomial method} \cite{BBC+01,BKT20,SY23}, the \textit{adversary method} \cite{Amb02,SS06}, and the \textit{compressed oracle method} \cite{Zha19} were successively proposed for proving quantum lower bounds for a wide range of problems.
It was shown in \cite{ABDK+21} via the polynomial method that $\mathsf{D}\rbra{f} = O\rbra{\mathsf{Q}\rbra{f}^4}$ for any total Boolean function $f$, where $\mathsf{D}$ and $\mathsf{Q}$ are the (classical) deterministic query complexity and the quantum query complexity (with two-sided error) for computing $f$, respectively.
This polynomial relation is the best possible due to the quantum speedup $\mathsf{D}\rbra{f} = \widetilde\Omega\rbra{\mathsf{Q}\rbra{f}^4}$ for certain total Boolean function $f$ \cite{ABB+17}. 

Other than for Boolean functions, quantum query algorithms are also widely studied for testing properties of unitary oracles \cite{Wan11}, which is an important part of quantum property testing \cite{MdW16}. 
Query lower bounds for several unitary property testing problems were recently established via the polynomial method \cite{AKKT20,CNY23,SY23} and also via the adversary method \cite{MdW23}. 
An interesting question is whether there is a relation related to quantum property testing like the aforementioned polynomial relation for Boolean functions.

Unitary oracles that encode inputs in a matrix block, known as block-encodings \cite{LC19}, have recently attracted great interest
due to that many quantum algorithms can be described in this framework, via the quantum singular value transformation (QSVT)~\cite{GSLW19}. Examples include Hamiltonian simulation \cite{LC19} and solving systems of linear equations \cite{CKS17,CAS+22}.
Among these examples, quantum Gibbs sampling was recently used as a key subroutine in quantum semidefinite programming \cite{BS17,vAGGdW20,BKL+19,vAG19}.
The current best quantum Gibbs sampler is due to \cite{GSLW19} which prepares the Gibbs state of the block-encoded Hamiltonian at inverse temperature $\beta$ with query complexity $O\rbra{\beta}$.
However, there is no known query lower bound for quantum Gibbs sampling.\footnote{To the best of our knowledge, we are only aware of a quantum circuit depth lower bound of $\Omega\rbra{\log \log \rbra{N}}$ \cite{Has11,Yos11,Eld21} where $N$ is the dimension of the Hamiltonian, which is incomparable with the quantum query complexity considered in this paper.}

In this paper, we propose a new method for proving quantum query lower bounds
and use it to prove a matching query lower bound $\widetilde \Omega\rbra{\beta}$ for quantum Gibbs sampling, as well as a new query lower bound for the entanglement entropy problem recently studied in \cite{SY23}.
It also provides a unified proof for several other results that have been previously proven, including the quantum query lower bounds for phase estimation, amplitude estimation, and Hamiltonian simulation. 

Very different from the previous methods, our lower bounds are obtained by reducing from information-theoretic bounds for quantum state testing, e.g., the Helstrom-Holevo bound \cite{Hel67,Hol73}. 
To enable such reduction, we propose a quantum sample-to-query lifting theorem which translates the sample lower bound for any type of quantum state testing into the query lower bound for a corresponding type of unitary property testing.
Moreover, our lifting theorem also implies an optimal polynomial relation between quantum sample and query complexities for quantum property testing, which is analogous to the aforementioned polynomial relation for Boolean functions \cite{BBC+01,ABDK+21} 
(the analogy will be explained in \cref{remark:qsq-vs-dq}).
Here, the quantum query complexity of property testing of quantum states considered in this paper is defined to be the quantum query complexity of property testing of quantum unitary oracles with the unitary oracle to be tested being a (scaled) block-encoding of the quantum state to be tested (see \cref{sec:intro-lifting} for the detailed definition).

\subsection{Quantum sample-to-query lifting theorem} \label{sec:intro-lifting}

Quantum state testing is to test whether a quantum state satisfies certain property, given its independent and identical samples. 
Formally, a quantum state testing problem $\mathcal{P}$ is a promise problem represented by $\mathcal{P} = \rbra{\mathcal{P}^{\textup{yes}}, \mathcal{P}^{\textup{no}}}$, where $\mathcal{P}^{\textup{yes}}$ and $\mathcal{P}^{\textup{no}}$ are disjoint sets of quantum states. 
The goal of $\mathcal{P}$ is to determine whether a given quantum state $\rho \in \mathcal{P}^{\textup{yes}}$ or $\rho \in \mathcal{P}^{\textup{no}}$, given the promise that it is in either case.
The quantum sample complexity of $\mathcal{P}$, denoted by $\mathsf{S}\rbra{\mathcal{P}}$, is the minimum number of independent samples of $\rho$ for solving $\mathcal{P}$ with success probability at least $2/3$.

In this paper, we also consider a variant of quantum state testing, given access to a unitary oracle that block-encodes the tested quantum state $\rho$. 
In other words, this variant version is actually a unitary property testing problem,
with quantum oracles being block-encodings of quantum states. 
Here, a unitary operator $U$ is said to be a block-encoding of $A$ if the matrix $A$ is in the upper left corner in the matrix representation of $U$ (see \cref{def:block-encoding}). 
The (diamond) quantum query complexity of $\mathcal{P}$, denoted by $\mathsf{Q}_{\diamond}\rbra{\mathcal{P}}$,\footnote{There are other definitions of the quantum query complexity for quantum state testing when encoding it in the scenario of unitary property testing. 
A common way is to assume a quantum oracle that prepares a purification of the tested quantum state, e.g., \cite{Wat02,BASTS10,GL20,GHS21,SH21,WZC+23,RASW23,WGL+22,GP22,Liu23}.
To avoid possible ambiguity, we use $\mathsf{Q}_\diamond\rbra{\mathcal{P}}$ instead of $\mathsf{Q}\rbra{\mathcal{P}}$ to denote the quantum query complexity of a quantum state testing problem $\mathcal{P}$ studied in this paper. \label{fnt:oracle-def}} is then defined as the minimum number of queries to a block-encoding of $\frac{1}{2}\rho$ for solving $\mathcal{P}$ with success probability at least $2/3$.\footnote{The constant factor in the block-encoding is due to technical reasons, which is not a problem in most applications. Generally, let $\mathsf{Q}_\diamond^{\alpha}\rbra{\mathcal{P}}$ be the $\alpha$-quantum query complexity that is defined to be the minimum number of queries to a block-encoding of $\rho/\alpha$.
The quantum query complexity $\mathsf{Q}_\diamond\rbra{\mathcal{P}}$ is actually the case of $\alpha = 2$.
Moreover, it holds that $\mathsf{Q}_\diamond\rbra{\mathcal{P}} \coloneqq \mathsf{Q}_\diamond^2\rbra{\mathcal{P}} = \widetilde \Theta\rbra{\mathsf{Q}_\diamond^{\alpha}\rbra{\mathcal{P}}}$ for any constant $\alpha > 1$ (see \cref{sec:alpha-query} for further discussion regarding the choice of $\alpha$).}

We prove a quantum sample-to-query lifting theorem that relates the above two complexities $\mathsf{S}\rbra{\mathcal{P}}$ and $\mathsf{Q}_\diamond\rbra{\mathcal{P}}$, stated as follows.

\begin{theorem} [Quantum sample-to-query lifting, \cref{thm:main}]
\label{thm:lifting}
    Let $\mathcal{P}$ be a promise problem for quantum state testing. Then,
    \[
    \mathsf{Q}_\diamond\rbra{\mathcal{P}} = \widetilde \Omega\rbra*{\sqrt{\mathsf{S}\rbra{\mathcal{P}}}}.
    \]
\end{theorem}

\paragraph{What is lifting?}
Typically, a lifting theorem in computational complexity (cf.\ \cite{GPW18,CKLM18,GPW20,CFK+19,ABDK21}) is a quantitative relation of the form $M_2 \rbra{L\rbra{f}} \geq q\rbra{M_1 \rbra{f}}$, where 
\begin{enumerate}
    \item $f \colon \mathcal{X} \rightharpoonup \cbra{0, 1}$ is a property of the input $x \in \mathcal{X}$. 
    \item $L$ maps the property $f$ to another property $f' = L\rbra{f} \colon \mathcal{X}' \rightharpoonup \cbra{0, 1}$. 
    \item $M_1$ and $M_2$ are two complexity measures of computational problems $f$ and $f' = L\rbra{f}$, respectively.
    \item $q$ is a scaling function. 
\end{enumerate}
The relations among $f, L, M_1, M_2, q$ are illustrated in \cref{fig:lifting}.
Note that in a lifting theorem, $M_1$ is usually a complexity measure in a weaker computational model than $M_2$. 
For example, it is easy to simulate query algorithms by communication protocols for XOR functions but the reverse is challenging (see \cite{GRT22,GSTW23}).

\begin{figure}[!htp]
    \centering

\tikzset{every picture/.style={line width=0.75pt}} %set default line width to 0.75pt        

\begin{tikzpicture}[x=0.75pt,y=0.75pt,yscale=-1,xscale=1]
%uncomment if require: \path (0,480); %set diagram left start at 0, and has height of 480

%Straight Lines [id:da09048089334822729] 
\draw    (245,210) -- (318,210) ;
\draw [shift={(320,210)}, rotate = 180] [color={rgb, 255:red, 0; green, 0; blue, 0 }  ][line width=0.75]    (10.93,-3.29) .. controls (6.95,-1.4) and (3.31,-0.3) .. (0,0) .. controls (3.31,0.3) and (6.95,1.4) .. (10.93,3.29)   ;
%Straight Lines [id:da8346513627988235] 
\draw    (330,120) -- (330,198) ;
\draw [shift={(330,200)}, rotate = 270] [color={rgb, 255:red, 0; green, 0; blue, 0 }  ][line width=0.75]    (10.93,-3.29) .. controls (6.95,-1.4) and (3.31,-0.3) .. (0,0) .. controls (3.31,0.3) and (6.95,1.4) .. (10.93,3.29)   ;
%Straight Lines [id:da41463764802267566] 
\draw    (245,200) -- (318.59,121.41) ;
\draw [shift={(320,120)}, rotate = 135] [color={rgb, 255:red, 0; green, 0; blue, 0 }  ][line width=0.75]    (10.93,-3.29) .. controls (6.95,-1.4) and (3.31,-0.3) .. (0,0) .. controls (3.31,0.3) and (6.95,1.4) .. (10.93,3.29)   ;

% Text Node
\draw (229,202) node [anchor=north west][inner sep=0.75pt]   [align=left] {$F$};
% Text Node
\draw (324,202) node [anchor=north west][inner sep=0.75pt]   [align=left] {$\mathbb{N}$};
% Text Node
\draw (324,101) node [anchor=north west][inner sep=0.75pt]   [align=left] {$F'$};
% Text Node
\draw (264,191) node [anchor=north west][inner sep=0.75pt]   [align=left] {$q \circ M_1$};
% Text Node
\draw (269,141) node [anchor=north west][inner sep=0.75pt]   [align=left] {$L$};
% Text Node
\draw (335,151) node [anchor=north west][inner sep=0.75pt]   [align=left] {$M_2$};

\end{tikzpicture}
\caption{Commutative diagram for the form of lifting theorems, where $F$ is the set of all possible $f$'s and $F' = L\rbra{F} \coloneqq \{L\rbra{f}:f \in F\}$.}
\label{fig:lifting}
\end{figure}

For (classical) deterministic query-to-communication lifting (cf.\ \cite{GPW18}), $f$ is a Boolean function, $M_1 \coloneqq \mathsf{D}$ is the deterministic query complexity of $f$, $M_2 \coloneqq \mathsf{D}^{\mathsf{cc}}$ is the corresponding communication complexity, and $L$ is induced by a ``gadget'' $g$ such that $\rbra{L\rbra{f}}\rbra{x_1, x_2, \dots, x_n, y_1, y_2, \dots, y_n} \coloneqq f\rbra{g\rbra{x_1, y_1}, g\rbra{x_2, y_2}, \dots, g\rbra{x_n, y_n}}$. 
The query-to-communication lifting theorem in \cite{GPW18} states that there exists a mapping $L$ (induced by a ``gadget'' $g$) such that $\mathsf{D}^{\mathsf{cc}}\rbra{L\rbra{f}} = \Omega\rbra{\mathsf{D}\rbra{f}\log\rbra{n}}$, where $q\rbra{x} = \Omega\rbra{x \log\rbra{n}}$.

For our quantum sample-to-query lifting (\cref{thm:lifting}), $f \coloneqq \mathcal{P}$ is a property of quantum states, $M_1 \coloneqq \mathsf{S}$ is the quantum sample complexity of $\mathcal{P}$, and $M_2 \coloneqq \mathsf{Q}$ is the quantum query complexity of $\mathcal{P}^{\diamond}$, where $\mathcal{P}^{\diamond}$ is the unitary property testing problem corresponding to $\mathcal{P}$ with each of its yes/no instances being a unitary operator that block-encodes the quantum state $\rho$ to be tested (see \cref{def:p-diamond}). 
By taking $L$ to be the mapping that maps every quantum state testing problem $\mathcal{P}$ to the corresponding unitary property testing problem $\mathcal{P}^{\diamond}$, our quantum sample-to-query lifting theorem (\cref{thm:lifting}) states that $\mathsf{Q}_\diamond\rbra{\mathcal{P}} \coloneqq \mathsf{Q}\rbra{\mathcal{P}^{\diamond}} = \widetilde\Omega\rbra{\sqrt{\mathsf{S}\rbra{\mathcal{P}}}}$, where $q\rbra{x} = \widetilde\Omega\rbra{\sqrt{x}}$. 
Note that $\mathsf{Q}_{\diamond}\rbra{\mathcal{P}}$ is a complexity measure of the quantum state testing problem $\mathcal{P}$ while $\mathsf{Q}\rbra{\mathcal{P}^{\diamond}}$ is a complexity measure of the unitary property testing problem $\mathcal{P}^{\diamond}$.

\paragraph{What does it mean?}

Our lifting theorem reveals a novel connection between quantum state testing and unitary property testing.
It provides a new method for proving quantum query lower bounds: 
by reducing from quantum state testing, we can systematically derive quantum query lower bounds 
using corresponding sample lower bounds from quantum information theory. 
This method is in sharp contrast to the polynomial method \cite{BBC+01}, the adversary method \cite{Amb02}, and the compressed oracle method \cite{Zha19}, where they characterize the progress of the computation in quantum query algorithms by degrees of polynomials or values of potential functions.
Applications of \cref{thm:lifting} in proving lower bounds will be discussed in \cref{sec:app}.

\begin{proof} [Proof sketch of \cref{thm:lifting}]
    The basic idea comes from quantum principal component analysis \cite{LMR14}, which allows one to implement the unitary transformation $e^{-i\rho t}$ to precision $\delta$ using $\Theta\rbra{t^2/\delta}$ samples of $\rho$ \cite{KLL+17}. 
    Using quantum singular value transformation \cite{GSLW19}, one can implement a block-encoding of $\frac{\pi}{4}\rho$ to precision $\delta$ using $\widetilde O\rbra{1/\delta}$ samples of $\rho$ \cite{GP22}. 

    For convenience, we use $\mathcal{P}^\diamond$ to denote the unitary property testing problem corresponding to $\mathcal{P}$.
    We reduce $\mathcal{P}$ to $\mathcal{P}^\diamond$ as follows.
    Let $\mathcal{A}$ be a quantum tester for $\mathcal{P}^\diamond$ with query complexity $Q$, and let $\rho$ be a quantum state to be tested. 
    We first implement a block-encoding $U_\rho$ of $\frac{1}{2}\rho$ to precision $\Theta\rbra{1/Q}$ using $\widetilde O\rbra{Q}$ samples of $\rho$. 
    Then, we use the tester $\mathcal{A}$ to test the approximated $U_\rho$ (regardless of whether it satisfies the promise of $\mathcal{P}^\diamond$). Finally, we decide whether $\rho$ is a \textit{yes} or \textit{no} instance of $\mathcal{P}$ following the decision by $\mathcal{A}$. 
    With a constant number of repetitions of the above procedure, we can decide whether $\rho \in \mathcal{P}^{\textup{yes}}$ or $\rho \in \mathcal{P}^{\textup{no}}$ with success probability at least $2/3$.

    Now that $\mathcal{A}$ uses $Q$ queries to $U_\rho$ and each $U_\rho$ uses $\widetilde O\rbra{Q}$ samples of $\rho$, this approach uses $\widetilde O\rbra{Q^2}$ samples of $\rho$ in total.
    The conclusion follows by noting that $\widetilde O\rbra{Q^2} \geq \mathsf{S}\rbra{\mathcal{P}}$.
\end{proof}

\begin{remark} \label{remark:choi}
Our lifting theorem implies a Karp reduction from quantum state testing to unitary property testing. 
In contrast, only a reduction in the reverse direction is previously known via the Choi-Jamio{\l}kowski isomorphism \cite{Cho75,Jam72}.
For example, the equality of unitary operators can be reduced to the equality of their Choi states (cf. \cite[Section 5.1.3]{MdW16}). 
For an $N$-dimensional unitary operator $U$, the Choi state of $U$ is a pure state $\ket{U} = \frac{1}{\sqrt{N}} \sum_{j=0}^{N-1} \sum_{j=0}^{N-1} U_{jk} \ket{j} \ket{k}$, which can be prepared by applying $U \otimes I$ on the maximally entangled state $\frac{1}{\sqrt{N}} \sum_{j=0}^{N-1} \ket{j} \ket{j}$.
Then, two unitary operators $U$ and $V$ are equal (up to a global phase, say $U = e^{i\theta}V$) if and only if their Choi states $\ket{U}$ and $\ket{V}$ are indistinguishable (which can be simply tested through the SWAP test \cite{BCWdW01}).
According to the analysis in \cite{MdW16}, whether two unitary operators are equal or $\varepsilon$-far (in the distance $d\rbra{U, V} = \sqrt{1 - \abs{\tr\rbra{U^\dag V}}^2/N^2}$ induced by the normalized Hilbert-Schmidt inner product) can be determined with query complexity $O\rbra{1/\varepsilon^2}$.
\end{remark}

\paragraph{Quantum state discrimination.}
Discriminating quantum states is a basic problem in quantum information theory and quantum property testing (cf. \cite{Che00,BC09}). 
As an illustrative example, we show a new query lower bound for quantum state discrimination by applying \cref{thm:lifting}. 
Specifically, let $\rho$ and $\sigma$ be two quantum states. We use $\textsc{Dis}_{\rho, \sigma}$ to denote the promise problem for quantum state discrimination, with only one \textit{yes} instance $\rho$ and one \textit{no} instance $\sigma$. 
That is, the task is to determine whether the tested quantum state is $\rho$ or $\sigma$.
By the Helstrom-Holevo bound \cite{Hel67,Hol73}, we know that $\mathsf{S}\rbra{\textsc{Dis}_{\rho, \sigma}} = \Omega\rbra{1/\gamma}$ (see \cref{lemma:dis-sample}), where $\gamma = 1-\operatorname{F}\rbra{\rho, \sigma}$ is the infidelity. 
By simply invoking \cref{thm:lifting}, we have:
\begin{corollary} [Query lower bound for quantum state discrimination]
\label{corollary:qsd-intro} 
    Let $\rho$ and $\sigma$ be two quantum states. Then, $\mathsf{Q}_\diamond\rbra{\textsc{Dis}_{\rho, \sigma}} = \widetilde \Omega\rbra{1/\sqrt{\gamma}}$, where $\gamma = 1-\operatorname{F}\rbra{\rho, \sigma}$ is the infidelity. 
\end{corollary}

The new lower bound given in \cref{corollary:qsd-intro} turns out to be useful for proving a number of quantum lower bounds (see \cref{sec:app}).

\paragraph{Optimal separation between quantum sample and query complexities.}
\cref{thm:lifting} can be alternatively expressed as $\mathsf{S}\rbra{\mathcal{P}} = \widetilde O\rbra{\mathsf{Q}_\diamond\rbra{\mathcal{P}}^2}$, which is a quadratic relation between $\mathsf{S}\rbra{\mathcal{P}}$ and $\mathsf{Q}_\diamond\rbra{\mathcal{P}}$ for any quantum state testing problem $\mathcal{P}$.
We show that this quadratic relation is the best possible up to logarithmic factors by giving an optimal separation based on quantum state discrimination. 

\begin{theorem} [Lifting is tight, \cref{thm:tightness}]
\label{thm:lifting-is-tight}    
$\mathsf{S}\rbra{\textsc{Dis}_{\rho_+, \rho_-}} = \Omega\rbra{\mathsf{Q}_\diamond\rbra{\textsc{Dis}_{\rho_+, \rho_-}}^2}$, 
where $\rho_{\pm} = I/2 + 8 \varepsilon Z$ is a one-qubit quantum state and $Z$ is the Pauli-Z operator.
\end{theorem}

\begin{proof} [Proof sketch]
    Note that the infidelity between $\rho_+$ and $\rho_-$ is $\gamma = O\rbra{\varepsilon^2}$. 
    Then by \cref{lemma:dis-sample}, we have $\mathsf{S}\rbra{\textsc{Dis}_{\rho_+, \rho_-}} = \Omega\rbra{1/\gamma} = \Omega\rbra{1/\varepsilon^2}$.
    To complete the proof, we only have to show that $\mathsf{Q}_\diamond\rbra{\textsc{Dis}_{\rho_+, \rho_-}} = O\rbra{1/\varepsilon}$.
    This is obtained by using quantum amplitude estimation \cite{BHMT02} and by noting that $\bra{0} U_{\rho_\pm} \ket{0} = \frac{1}{4} \pm \Theta\rbra{\varepsilon}$ for every block-encoding $U_{\rho_\pm}$ of $\frac{1}{2}\rho_\pm$.
\end{proof}

The proof of \cref{thm:lifting-is-tight} actually provides a natural reduction from quantum state discrimination to quantum amplitude estimation, which can also be used to show the optimality of quantum amplitude estimation later in \cref{sec:app}.

\begin{remark} \label{remark:qsq-vs-dq}
    The optimal $\mathsf{S}$-$\mathsf{Q}_\diamond$ relation for quantum state testing problems revealed by \cref{thm:lifting} and \cref{thm:lifting-is-tight} is analogous to the optimal $\mathsf{D}$-$\mathsf{Q}$ relation for Boolean functions due to \cite{ABB+17,ABDK+21}, as shown in \cref{fig:cmp}.
    The similarity between the $\mathsf{D}$-$\mathsf{Q}$ relation for Boolean functions and the $\mathsf{S}$-$\mathsf{Q}_\diamond$ relation for quantum state testing problems is explained as follows.
    \begin{itemize}
        \item $\mathsf{D}$-$\mathsf{Q}$ relation for Boolean functions.
        For a Boolean function $f$, the deterministic query complexity $\mathsf{D}\rbra{f}$ and the quantum query complexity $\mathsf{Q}\rbra{f}$ are different complexity measures of computing the function value $f\rbra{x_1, x_2, \dots, x_n} \in \cbra{0, 1}$ given the input variables $x_1, x_2, \dots, x_n$. 
        \item $\mathsf{S}$-$\mathsf{Q}_\diamond$ relation for quantum state testing problems. 
        For a quantum state testing problem $\mathcal{P}$, the quantum sample complexity $\mathsf{S}\rbra{\mathcal{P}}$ and the quantum query complexity $\mathsf{Q}_\diamond\rbra{\mathcal{P}}$ are different complexity measures of computing the function value $\mathcal{P}\rbra{\rho} \in \cbra{0, 1}$ given the input quantum state $\rho$. 
    \end{itemize}
    Given the above, it can be seen that \cref{fig:cmp} shows a structural similarity in the correspondence $f \leftrightarrow \mathcal{P}$, $\mathsf{D} \leftrightarrow \mathsf{S}$, and $\mathsf{Q} \leftrightarrow \mathsf{Q}_\diamond$. 

    It is worth mentioning that the relationships between randomized query complexity and quantum query complexity have also been investigated in the literature \cite{BBC+01,Aar08,Mid05,Amb16,ABB+17}.
    For example, in \cite{ABB+17}, they showed that $\mathsf{R}_0\rbra{f} = \widetilde\Omega\rbra{\mathsf{Q}_{\textup{E}}\rbra{f}^2}$ for certain Boolean function $f$, while $\mathsf{R}_0\rbra{f} = O\rbra{\mathsf{Q}_{\textup{E}}\rbra{f}^3}$ in general \cite{Mid04}, where $\mathsf{Q}_{\textup{E}}$ stands for the exact quantum query complexity, and $\mathsf{R}_0$ stands for the zero-error randomized query complexity.
\end{remark}

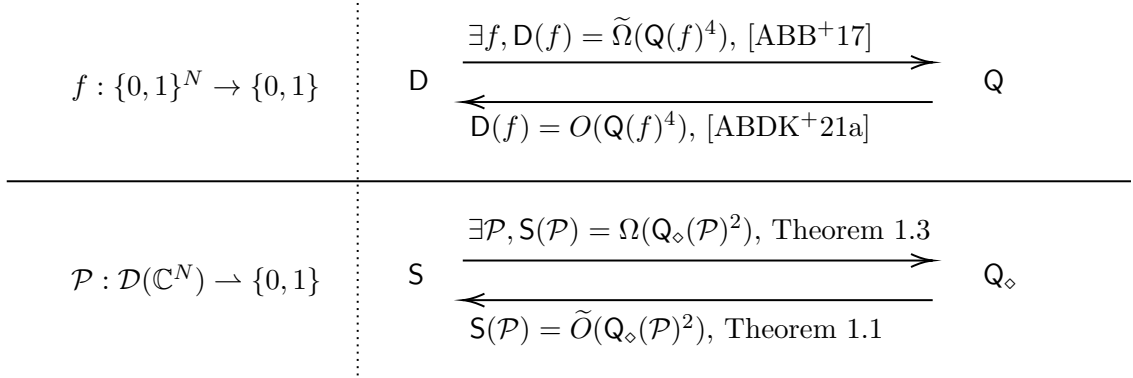
\begin{figure}[!htp]
    \centering

\tikzset{every picture/.style={line width=0.75pt}} %set default line width to 0.75pt        

\adjustbox{max width=\textwidth}{
\begin{tikzpicture}[x=0.75pt,y=0.75pt,yscale=-1,xscale=1]
%uncomment if require: \path (0,480); %set diagram left start at 0, and has height of 480

%Straight Lines [id:da08583005640690722] 
\draw    (63,140) -- (630,140) ;
%Straight Lines [id:da8000927327436469] 
\draw    (530,100) -- (295,100) ;
\draw [shift={(293,100)}, rotate = 360] [color={rgb, 255:red, 0; green, 0; blue, 0 }  ][line width=0.75]    (10.93,-3.29) .. controls (6.95,-1.4) and (3.31,-0.3) .. (0,0) .. controls (3.31,0.3) and (6.95,1.4) .. (10.93,3.29)   ;
%Straight Lines [id:da8438232380698756] 
\draw    (293,80) -- (528,80) ;
\draw [shift={(530,80)}, rotate = 180] [color={rgb, 255:red, 0; green, 0; blue, 0 }  ][line width=0.75]    (10.93,-3.29) .. controls (6.95,-1.4) and (3.31,-0.3) .. (0,0) .. controls (3.31,0.3) and (6.95,1.4) .. (10.93,3.29)   ;
%Straight Lines [id:da10080114737303258] 
\draw    (530,200) -- (295,200) ;
\draw [shift={(293,200)}, rotate = 360] [color={rgb, 255:red, 0; green, 0; blue, 0 }  ][line width=0.75]    (10.93,-3.29) .. controls (6.95,-1.4) and (3.31,-0.3) .. (0,0) .. controls (3.31,0.3) and (6.95,1.4) .. (10.93,3.29)   ;
%Straight Lines [id:da7615519757595726] 
\draw    (293,180) -- (528,180) ;
\draw [shift={(530,180)}, rotate = 180] [color={rgb, 255:red, 0; green, 0; blue, 0 }  ][line width=0.75]    (10.93,-3.29) .. controls (6.95,-1.4) and (3.31,-0.3) .. (0,0) .. controls (3.31,0.3) and (6.95,1.4) .. (10.93,3.29)   ;
%Straight Lines [id:da7162670558133288] 
\draw[dotted]    (240,50) -- (240,240) ;

% Text Node
\draw (94,180.4) node [anchor=north west][inner sep=0.75pt]    {$\mathcal{P} :\mathcal{D}(\mathbb{C}^{N})\rightharpoonup \{0,1\}$};
% Text Node
\draw (94,82.4) node [anchor=north west][inner sep=0.75pt]    {$f:\{0,1\}^{N}\rightarrow \{0,1\}$};
% Text Node
\draw (264,82.4) node [anchor=north west][inner sep=0.75pt]    {$\mathsf{D}$};
% Text Node
\draw (554,82.4) node [anchor=north west][inner sep=0.75pt]    {$\mathsf{Q}$};
% Text Node
\draw (334,55.4) node [anchor=north west][inner sep=0.75pt]    {$\exists f, \mathsf{D}(f) =\widetilde\Omega ( \mathsf{Q}( f)^{4})$, \cite{ABB+17}};
% Text Node
\draw (345,103.4) node [anchor=north west][inner sep=0.75pt]    {$\mathsf{D}( f) =O( \mathsf{Q}( f)^{4})$, \cite{ABDK+21}};
% Text Node
\draw (264,180.4) node [anchor=north west][inner sep=0.75pt]    {$\mathsf{S}$};
% Text Node
\draw (554,180.4) node [anchor=north west][inner sep=0.75pt]    {$\mathsf{Q}_\diamond$};
% Text Node
\draw (299,155.4) node [anchor=north west][inner sep=0.75pt]    {$\exists \mathcal{P}, \mathsf{S}(\mathcal{P}) =\Omega ( \mathsf{Q}_\diamond(\mathcal{P})^{2})$, \cref{thm:lifting-is-tight}};
% Text Node
\draw (315,203.4) node [anchor=north west][inner sep=0.75pt]    {$\mathsf{S}(\mathcal{P}) =\widetilde O( \mathsf{Q}_\diamond(\mathcal{P})^{2})$, \cref{thm:lifting}};

\end{tikzpicture}
}

\caption{$\mathsf{S}$-$\mathsf{Q}_\diamond$ relation for quantum state testing vs.\ $\mathsf{D}$-$\mathsf{Q}$ relation for Boolean functions.
Here, a quantum state testing problem $\mathcal{P} = \rbra{\mathcal{P}^{\textup{yes}}, \mathcal{P}^{\textup{no}}}$
is expressed as a partial function $\mathcal{P} \colon \mathcal{D} \rbra{\mathbb{C}^{N}}\rightharpoonup\{0,1\}$,
where $0$ and $1$ stand for \textit{yes} and \textit{no} instances, respectively.
}
\label{fig:cmp}
\end{figure}

\subsection{Applications: quantum query lower bounds}
\label{sec:app}

Now we show how our lifting theorem leads to unified proofs for quantum query lower bounds for several problems of general interest.
The proof for each bound is essentially a reduction from a corresponding quantum state testing problem, e.g., quantum state discrimination (\cref{corollary:qsd-intro}).
Such reduction also reveals an intrinsic relationship between them.
In \cref{fig:diagram}, a diagram on how we obtain the lower bounds and the relationships between them is presented. 

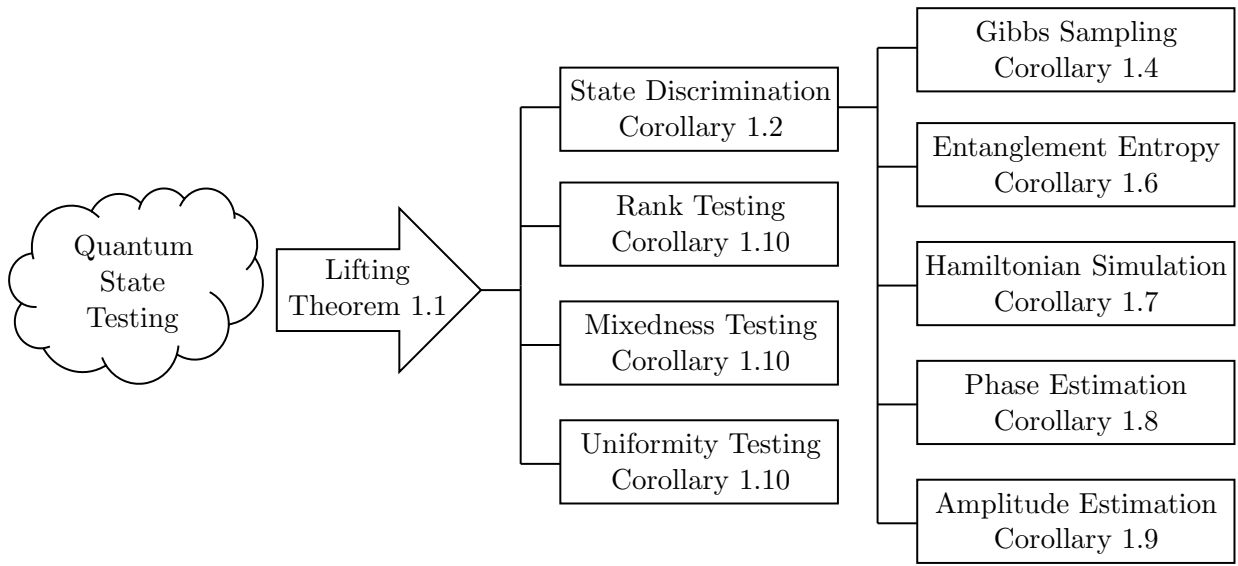
\begin{figure}[!htp]
    \centering

\tikzset{every picture/.style={line width=0.75pt}} %set default line width to 0.75pt       

\adjustbox{max width=\textwidth}{
\begin{tikzpicture}[x=0.75pt,y=0.75pt,yscale=-1,xscale=1]
%uncomment if require: \path (0,480); %set diagram left start at 0, and has height of 480

%Right Arrow [id:dp4542711620363211] 
\draw   (137,150.65) -- (198.8,150.65) -- (198.8,130) -- (240,171.3) -- (198.8,212.6) -- (198.8,191.95) -- (137,191.95) -- cycle ;
%Shape: Cloud [id:dp7687387083403769] 
\draw   (13.65,152.46) .. controls (12.62,144.51) and (16,136.64) .. (22.36,132.19) .. controls (28.73,127.74) and (36.95,127.49) .. (43.55,131.55) .. controls (45.89,126.92) and (50.16,123.73) .. (55.09,122.94) .. controls (60.01,122.14) and (65.01,123.84) .. (68.56,127.51) .. controls (70.55,123.32) and (74.46,120.5) .. (78.9,120.06) .. controls (83.34,119.62) and (87.68,121.62) .. (90.39,125.34) .. controls (93.98,120.9) and (99.7,119.03) .. (105.08,120.54) .. controls (110.45,122.05) and (114.51,126.67) .. (115.49,132.4) .. controls (119.9,133.66) and (123.57,136.87) .. (125.56,141.19) .. controls (127.54,145.52) and (127.65,150.54) .. (125.85,154.95) .. controls (130.19,160.88) and (131.21,168.78) .. (128.52,175.7) .. controls (125.83,182.62) and (119.84,187.53) .. (112.79,188.59) .. controls (112.74,195.08) and (109.35,201.05) .. (103.91,204.17) .. controls (98.48,207.3) and (91.86,207.11) .. (86.61,203.67) .. controls (84.37,211.45) and (78.07,217.17) .. (70.42,218.37) .. controls (62.78,219.57) and (55.17,216.02) .. (50.87,209.26) .. controls (45.61,212.59) and (39.3,213.55) .. (33.35,211.92) .. controls (27.41,210.29) and (22.34,206.21) .. (19.28,200.6) .. controls (13.91,201.26) and (8.71,198.34) .. (6.26,193.28) .. controls (3.82,188.22) and (4.66,182.1) .. (8.36,177.97) .. controls (3.56,175) and (1.11,169.12) .. (2.29,163.39) .. controls (3.47,157.66) and (8.01,153.37) .. (13.54,152.77) ; \draw   (8.36,177.96) .. controls (10.63,179.37) and (13.24,180) .. (15.86,179.78)(19.29,200.6) .. controls (20.41,200.46) and (21.51,200.17) .. (22.56,199.73)(50.87,209.26) .. controls (50.08,208.01) and (49.42,206.68) .. (48.9,205.29)(86.61,203.67) .. controls (87.01,202.25) and (87.28,200.79) .. (87.4,199.31)(112.79,188.59) .. controls (112.84,181.67) and (109.1,175.34) .. (103.17,172.31)(125.85,154.95) .. controls (124.89,157.31) and (123.42,159.39) .. (121.57,161.05)(115.49,132.4) .. controls (115.66,133.35) and (115.73,134.32) .. (115.72,135.28)(90.39,125.34) .. controls (89.49,126.45) and (88.75,127.69) .. (88.19,129.02)(68.56,127.51) .. controls (68.08,128.52) and (67.72,129.58) .. (67.49,130.68)(43.55,131.55) .. controls (44.94,132.41) and (46.24,133.44) .. (47.4,134.62)(13.65,152.46) .. controls (13.79,153.56) and (14.01,154.64) .. (14.32,155.7) ;
%Straight Lines [id:da5563863285850266] 
\draw    (260,169) -- (260,259) ;
%Straight Lines [id:da07151479858212251] 
\draw    (260,79) -- (260,169) ;
%Straight Lines [id:da5709950540404007] 
\draw    (260,259) -- (280,259) ;
%Straight Lines [id:da6324972523962966] 
\draw    (260,79) -- (280,79) ;
%Straight Lines [id:da9439578377945677] 
\draw    (260,139) -- (280,139) ;
%Straight Lines [id:da31814434169187833] 
\draw    (260,199) -- (280,199) ;
%Straight Lines [id:da7867636516596606] 
\draw    (420,79) -- (440,79) ;
%Straight Lines [id:da6016950721001328] 
\draw    (440,49) -- (440,289) ;
%Straight Lines [id:da39193948956227365] 
\draw    (440,49) -- (460,49) ;
%Straight Lines [id:da33356728242172196] 
\draw    (440,109) -- (460,109) ;
%Straight Lines [id:da33740083621328565] 
\draw    (440,169) -- (460,169) ;
%Straight Lines [id:da10578078507635946] 
\draw    (440,229) -- (460,229) ;
%Straight Lines [id:da7444719908790369] 
\draw    (440,289) -- (460,289) ;
%Shape: Rectangle [id:dp05727774897119242] 
\draw   (280,59) -- (420,59) -- (420,101) -- (280,101) -- cycle ;
%Shape: Rectangle [id:dp48813110585858155] 
\draw   (280,117) -- (420,117) -- (420,159) -- (280,159) -- cycle ;
%Shape: Rectangle [id:dp33977625133735634] 
\draw   (280,177) -- (420,177) -- (420,219) -- (280,219) -- cycle ;
%Shape: Rectangle [id:dp31167697303045716] 
\draw   (280,237) -- (420,237) -- (420,279) -- (280,279) -- cycle ;
%Shape: Rectangle [id:dp1392161217521699] 
\draw   (460,87) -- (620,87) -- (620,129) -- (460,129) -- cycle ;
%Shape: Rectangle [id:dp0031032975005680363] 
\draw   (460,147) -- (620,147) -- (620,189) -- (460,189) -- cycle ;
%Shape: Rectangle [id:dp04538071832549506] 
\draw   (460,207) -- (620,207) -- (620,249) -- (460,249) -- cycle ;
%Shape: Rectangle [id:dp046633148792090795] 
\draw   (460,267) -- (620,267) -- (620,309) -- (460,309) -- cycle ;
%Shape: Rectangle [id:dp30380497077731383] 
\draw   (460,29) -- (620,29) -- (620,71) -- (460,71) -- cycle ;
%Straight Lines [id:da39948807531317243] 
\draw    (240,171.3) -- (260,171.3) ;

% Text Node
\draw (33,134) node [anchor=north west][inner sep=0.75pt]   [align=left] {\begin{minipage}[lt]{44.68pt}\setlength\topsep{0pt}
\begin{center}
Quantum\\State\\Testing
\end{center}

\end{minipage}};
% Text Node
\draw (141,154) node [anchor=north west][inner sep=0.75pt]   [align=left] {\begin{minipage}[lt]{60.55pt}\setlength\topsep{0pt}
\begin{center}
Lifting\\ \cref{thm:lifting}
\end{center}

\end{minipage}};
% Text Node
\draw (276,64) node [anchor=north west][inner sep=0.75pt]   [align=left] {\begin{minipage}[lt]{110pt}\setlength\topsep{0pt}
\begin{center}
State Discrimination\\ \cref{corollary:qsd-intro}
\end{center}

\end{minipage}};
% Text Node
\draw (276,122) node [anchor=north west][inner sep=0.75pt]   [align=left] {\begin{minipage}[lt]{110pt}\setlength\topsep{0pt}
\begin{center}
Rank Testing\\ \cref{corollary:spectrum-testing}
\end{center}

\end{minipage}};
% Text Node
\draw (276,182) node [anchor=north west][inner sep=0.75pt]   [align=left] {\begin{minipage}[lt]{110pt}\setlength\topsep{0pt}
\begin{center}
Mixedness Testing\\ \cref{corollary:spectrum-testing}
\end{center}

\end{minipage}};
% Text Node
\draw (276,242) node [anchor=north west][inner sep=0.75pt]   [align=left] {\begin{minipage}[lt]{110pt}\setlength\topsep{0pt}
\begin{center}
Uniformity Testing\\ \cref{corollary:spectrum-testing}
\end{center}

\end{minipage}};
% Text Node
\draw (462,34) node [anchor=north west][inner sep=0.75pt]   [align=left] {\begin{minipage}[lt]{115pt}\setlength\topsep{0pt}
\begin{center}
Gibbs Sampling\\ \cref{thm:gibbs-query-lower-bound}
\end{center}

\end{minipage}};
% Text Node
\draw (462,92) node [anchor=north west][inner sep=0.75pt]   [align=left] {\begin{minipage}[lt]{115pt}\setlength\topsep{0pt}
\begin{center}
Entanglement Entropy\\ \cref{corollary:entro-intro}
\end{center}

\end{minipage}};
% Text Node
\draw (462,152) node [anchor=north west][inner sep=0.75pt]   [align=left] {\begin{minipage}[lt]{115pt}\setlength\topsep{0pt}
\begin{center}
Hamiltonian Simulation\\ \cref{corollary:hamsim-intro}
\end{center}

\end{minipage}};
% Text Node
\draw (462,212) node [anchor=north west][inner sep=0.75pt]   [align=left] {\begin{minipage}[lt]{115pt}\setlength\topsep{0pt}
\begin{center}
Phase Estimation\\ \cref{corollary:phe-intro}
\end{center}

\end{minipage}};
% Text Node
\draw (462,272) node [anchor=north west][inner sep=0.75pt]   [align=left] {\begin{minipage}[lt]{115pt}\setlength\topsep{0pt}
\begin{center}
Amplitude Estimation\\ \cref{corollary:ampl-intro}
\end{center}

\end{minipage}};

\end{tikzpicture}
}

\caption{Diagram of reductions and relationships amongst our results.}
\label{fig:diagram}
\end{figure}

\vspace{5pt}
\textbf{Quantum Gibbs sampling} is a task to prepare the Gibbs (thermal) state $e^{-\beta H}/\mathcal{Z}$ of a given Hamiltonian $H$ at inverse temperature $\beta \geq 0$, where $\mathcal{Z} = \tr\rbra{e^{-\beta H}}$ is the normalization factor. 
This task is of great interest in physics, with several Gibbs samplers proposed in \cite{PW09,TOV+11,YAG12,KB16,CS17,BK19}. 
Recently, it was found that Gibbs sampling is useful in solving semidefinite programs \cite{BS17,vAGGdW20,BKL+19,vAG19}. 
The current best quantum Gibbs sampler by \cite{GSLW19} uses $O\rbra{\beta}$ queries to the block-encoding of $H$.

We give a matching quantum query lower bound for quantum Gibbs sampling as follows.

\begin{corollary} [Optimality of quantum Gibbs sampling, \cref{thm:gibbs}]
\label{thm:gibbs-query-lower-bound}
    Every quantum query algorithm for quantum Gibbs sampling for Hamiltonian $H$ at inverse temperature $\beta$, given an oracle that is a block-encoding of $H$, has query complexity $\widetilde \Omega\rbra{\beta}$.
\end{corollary}

An $\Omega\rbra{\sqrt{N}}$ lower bound for quantum Gibbs sampling for $N$-dimensional Hamiltonians can be derived by a simple reduction from the search problem (see \cref{sec:gibbs-dimension}). 
Therefore, combining it with \cref{thm:gibbs-query-lower-bound}, we obtain a query lower bound $\widetilde\Omega\rbra{\beta+\sqrt{N}}$ for quantum Gibbs sampling, which is almost tight (up to logarithmic factors) with respect to all parameters, compared to the upper bound $\widetilde O\rbra{\beta \sqrt{N}}$ given in \cite[Theorem 35]{GSLW19}.

\begin{proof} [Proof sketch of \cref{thm:gibbs-query-lower-bound}]
    We reduce $\textsc{Dis}_{\rho_+, \rho_-}$ to quantum Gibbs sampling, where $\rho_\pm = \frac{1}{2}I \mp \frac{2}{\beta}Z$ is a one-qubit quantum state with infidelity $\gamma = O\rbra{1/\beta^2}$. 
    Suppose one can prepare the Gibbs state of any Hamiltonian at inverse temperature $\beta$ with query complexity $Q$. 
    The Gibbs state of $\frac{1}{2}\rho_\pm$ (note that the matrix representation of $\frac{1}{2}\rho_\pm$ is a valid Hamiltonian!) at inverse temperature $\beta$ is $\propto e^{-\frac{1}{2}\beta \rho_\pm} \propto e^{\pm 1} \ket{0}\bra{0} + e^{\mp 1} \ket{1} \bra{1}$, which can be easily distinguished by measuring it in the computational basis.
    By \cref{corollary:qsd-intro}, we conclude that $Q \geq \mathsf{Q}_\diamond\rbra{\textsc{Dis}_{\rho_+, \rho_-}} = \widetilde \Omega\rbra{1/\sqrt{\gamma}} = \widetilde\Omega\rbra{\beta}$.
\end{proof}

It was shown in \cite{CS17} that if the given oracle is a block-encoding of $\sqrt{H}$, then one can prepare the Gibbs state of $H$ at inverse temperature $\beta$ with query complexity $O\rbra{\sqrt{\beta}}$ in some special cases, e.g., when $H$ is a linear combination of Pauli operators. 
We show that such quadratic improvement on $\beta$, however, cannot extend to the general case.

\begin{corollary} [Impossibility of $O\rbra{\sqrt{\beta}}$-time quantum Gibbs sampling, \cref{thm:sqrt-beta-gibbs}]
    Every quantum query algorithm for quantum Gibbs sampling for Hamiltonian $H$ at inverse temperature $\beta$, given an oracle that is a block-encoding of $\sqrt{H}$, has query complexity $\widetilde \Omega\rbra{\beta}$.
\end{corollary}
\begin{proof} [Proof sketch]
    We reduce the same problem to quantum Gibbs sampling given access to $\sqrt{H}$ as the proof of \cref{thm:gibbs-query-lower-bound}. 
    The difference is that the Hamiltonian $H$ becomes either $\frac{1}{4}\rho_+^2$ or $\frac{1}{4}\rho_-^2$. 
    The proof is immediately done by noting that the Gibbs state of $\frac{1}{4}\rho_{\pm}^2$ at inverse temperature $\beta$ is $e^{-\frac{1}{4}\beta\rho_\pm^2} \propto e^{\pm \frac{1}{2}} \ket{0}\bra{0} + e^{\mp \frac{1}{2}} \ket{1} \bra{1}$, which can also be easily distinguished by measuring it in the computational basis.
\end{proof}

\vspace{5pt}
\textbf{The entanglement entropy problem} is a problem of unitary property testing recently studied in \cite{SY23}, which is to decide whether a bipartite state $\ket{\psi} \in \mathbb{C}^N \otimes \mathbb{C}^N$ has low ($\leq a$) or high ($\geq b$) ($2$-R\'enyi) entanglement entropy, given the reflection operator $U = I - 2\ket{\psi}\bra{\psi}$ as the oracle. 
An $\Omega\rbra{e^{a/4}}$ query lower bound for this problem was given in \cite[Theorem 1.14]{SY23}. 

We give a new query lower bound for the entanglement entropy problem as follows. 

\begin{corollary} [New lower bound for the entanglement entropy problem, \cref{thm:entangle-entro}]
\label{corollary:entro-intro}
    Every quantum tester for the entanglement entropy problem requires $\widetilde\Omega\rbra{1/\sqrt{\Delta}}$ queries, where $\Delta = b - a$.
\end{corollary}
\begin{proof}[Proof sketch]
    We reduce $\textsc{Dis}_{\rho, \sigma}$ to the entanglement entropy problem, where $\rho = I/2$ and $\sigma = I/2-\sqrt{\Delta}Z$ are one-qubit quantum states with infidelity $\gamma = O\rbra{\Delta}$. 
    Suppose there is a quantum tester for the entanglement entropy problem with query complexity $Q$. 
    For $\varrho \in \cbra{\rho, \sigma}$, we can (approximately) implement $R_{\varrho} = I - 2\ket{\varrho}\bra{\varrho}$ to precision $\Theta\rbra{1/Q}$ by QSVT techniques \cite{GSLW19}, using $\widetilde O\rbra{Q}$ queries, where $\ket{\varrho}$ is a purification of $\varrho$.
    Note that $\ket{\rho}$ has high ($\geq \ln\rbra{2}$) entropy and $\ket{\sigma}$ has low ($\leq \ln\rbra{2} - \Delta$) entropy. 
    Then, we can distinguish between $\rho$ and $\sigma$ by distinguishing between $R_{\rho}$ and $R_{\sigma}$ via the entanglement entropy problem. 
    By \cref{corollary:qsd-intro}, we conclude that $\widetilde O\rbra{Q} \geq \mathsf{Q}_\diamond\rbra{\textsc{Dis}_{\rho, \sigma}} = \widetilde \Omega\rbra{1/\sqrt{\gamma}} = \widetilde\Omega\rbra{1/\sqrt{\Delta}}$.
\end{proof}

We note that our lower bound is stronger than the one given in \cite{SY23} when $b - a \leq O\rbra{e^{-a/2}}$.

\vspace{5pt}
\textbf{Hamiltonian simulation} is a task to implement the unitary operator $e^{-iHt}$ for a given Hamiltonian $H$ and simulation time $t$.
It was first proposed by \cite{Fey82} in the hope that quantum systems can be simulated efficiently on quantum computers. Later, it was shown in \cite{Llo96} that this is possible for local Hamiltonians. 
Since then, more efficient quantum algorithms for Hamiltonian simulation have been proposed \cite{ATS03,BACS07,CW12,BCC+14,BCC+15,BCK15,LC17,LC19}. The state-of-the-art approaches \cite{LC17,LC19,GSLW19} have the optimal query complexity $\Theta\rbra{t}$ with respect to the simulation time $t$.

The query lower bound $\Omega\rbra{t}$ for Hamiltonian simulation given in \cite{BACS07,CK10,BCC+14,BCK15} is obtained by a reduction from the parity problem, whose lower bound can be derived by the polynomial method \cite{BBC+01} (see also \cite{FGGS98}).

We provide a new proof for the quantum query lower bound for Hamiltonian simulation. 

\begin{corollary} [Optimality of Hamiltonian simulation, \cref{thm:ham-sim}]
\label{corollary:hamsim-intro}
    Every quantum query algorithm for Hamiltonian simulation for time $t$, given an oracle that is a block-encoding of the Hamiltonian $H$, has query complexity $\widetilde \Omega\rbra{t}$. 
\end{corollary}
\begin{proof} [Proof sketch]
    We reduce $\textsc{Dis}_{\rho,\sigma}$ to Hamiltonian simulation, where $\rho = \frac{1}{2}I$ and $\sigma = \frac{1}{2}I+\frac{\pi}{t} Z$ are one-qubit quantum states with infidelity $\gamma = O\rbra{1/t^2}$.
    If we can simulate Hamiltonians $\frac{1}{2}\rho$ and $\frac{1}{2}\sigma$ (note that the matrix representation of $\frac{1}{2}\rho$ and $\frac{1}{2}\sigma$ are valid Hamiltonians!) for time $t$ with query complexity $Q$, then we can distinguish them by simulating them on the state $\ket{+}$ because $e^{-i \frac{\rho}{2} t} \ket{+} \propto \ket{+}$ and $e^{-i \frac{\sigma}{2} t} \ket{+} \propto \ket{-}$. 
    By \cref{corollary:qsd-intro}, we conclude that $Q \geq \mathsf{Q}_\diamond\rbra{\textsc{Dis}_{\rho, \sigma}} = \widetilde \Omega\rbra{1/\sqrt{\gamma}} = \widetilde\Omega\rbra{t}$.
\end{proof}

The reduction above was also used in \cite{KLL+17} to prove the optimality of quantum principal component analysis \cite{LMR14}. 
By comparison, we combine their idea with our lifting theorem to obtain a quantum query lower bound for Hamiltonian simulation.

\vspace{5pt}
\textbf{Quantum phase estimation} is a task to estimate the phase $\lambda$ of an eigenvector $\ket{\psi}$ of a given oracle $U$, promised that $U \ket{\psi} = e^{i \lambda} \ket{\psi}$. 
Initially introduced by \cite{Kit95}, it is used as a subroutine in various quantum algorithms, e.g., Shor's factoring \cite{Sho94}, quantum counting \cite{BHT98b}, and solving systems of linear equations \cite{HHL09}. 

A matching lower bound $\Omega\rbra{\log\rbra{1/\delta}}$ was shown in \cite{Bes05} regarding the precision $\delta$ in the case (e.g., for factoring \cite{Sho94}) when the quantum oracle $U$ and its (controlled) powers $U^t$ are given. 
For the general case when only the quantum oracle $U$ is given, a matching lower bound $\Omega\rbra{1/\delta}$ was shown in \cite{LT20} by reducing from quantum counting \cite{NW99}.
Very recently, the lower bound regarding the success probability was considered in \cite{MdW23}, showing that $\Omega\rbra{\log\rbra{1/\varepsilon}/\delta}$ queries are required to estimate the phase with probability at least $1-\varepsilon$. This is done by reducing from a unitary property testing problem, with its lower bound obtained by the polynomial method \cite{BBC+01}. 

We provide a new proof for the quantum query lower bound for phase estimation. 

\begin{corollary} [Optimality of quantum phase estimation, \cref{thm:ph-est}]
\label{corollary:phe-intro}
    Every quantum query algorithm for phase estimation requires $\widetilde\Omega\rbra{1/\delta}$ queries to the oracle $U$.
\end{corollary}

\begin{proof} [Proof sketch]
    We reduce $\textsc{Dis}_{\rho_+, \rho_-}$ to quantum phase estimation, where $\rho_\pm = I/2 \mp \delta Z$ are two one-qubit quantum states with infidelity $\gamma = O\rbra{\delta^2}$. 
    Suppose quantum phase estimation can be done with query complexity $Q$. 
    We first implement the unitary operator $e^{-i\frac{\rho_{\pm}}{2}}$ by Hamiltonian simulation (e.g., \cite[Corollary 32]{GSLW19}) to precision $\Theta\rbra{1/Q}$ for a unit simulation time, using $\widetilde O\rbra{1}$ queries.
    Then, we perform phase estimation on $e^{-i \frac{\rho_{\pm}}{2}}$ and its eigenvector $\ket{0}$ (note that the phase is $\frac{1}{4} \mp \Theta\rbra{\delta}$) to precision $\delta$, using $Q$ queries to the approximated $e^{-i\frac{\rho_\pm}{2}}$. 
    In summary, using $\widetilde O\rbra{Q}$ queries, we can distinguish the two cases by checking whether the estimate of the phase is greater or less than $\frac{1}{4}$. 
    By \cref{corollary:qsd-intro}, we conclude that $\widetilde O\rbra{Q} \geq \mathsf{Q}_\diamond\rbra{\textsc{Dis}_{\rho_+, \rho_-}} = \widetilde \Omega\rbra{1/\sqrt{\gamma}} = \widetilde\Omega\rbra{1/\delta}$.
\end{proof}

Both the reduction and construction are new,
compared to known proofs for phase estimation.
The key observation is that combining Hamiltonian simulation and phase estimation can solve quantum state discrimination. 

\vspace{5pt}
\textbf{Quantum amplitude estimation} is to estimate the amplitude $p$ of the quantum state prepared by a quantum oracle $U \ket{0} = \sqrt{p} \ket{0}\ket{\phi_0} + \sqrt{1-p} \ket{1} \ket{\phi_1}$. 
It is well-known that estimating $p$ to within additive error $\varepsilon$ can be done using $O\rbra{1/\varepsilon}$ queries \cite{BHMT02},\footnote{It was recently noted in \cite{Wan24} that estimating $\sqrt{p}$ to within additive error $\varepsilon$ can also be done using $O\rbra{1/\varepsilon}$ queries.} which is optimal by reducing from the lower bound for quantum counting \cite{BBC+01,NW99}.

We provide a new proof for the quantum query lower bound for amplitude estimation. 

\begin{corollary} [Optimality of quantum amplitude estimation, \cref{thm:ampl-est}]
\label{corollary:ampl-intro}
    Every quantum algorithm for amplitude estimation with precision $\varepsilon$ has query complexity $\widetilde \Omega\rbra{1/\varepsilon}$.
\end{corollary}

\begin{proof} [Proof sketch]
    We reduce $\textsc{Dis}_{\rho_+, \rho_-}$ to quantum amplitude estimation, where $\rho_\pm = I/2 \mp 8 \varepsilon Z$ 
    are two one-qubit quantum states with infidelity $\gamma = O\rbra{\varepsilon^2}$. 
    Suppose quantum amplitude estimation can be done with query complexity $Q$. 
    Using the same idea as in the proof of \cref{thm:lifting-is-tight}, we can distinguish the two cases by estimating the amplitude $\bra{0}U_{\rho_{\pm}}\ket{0} = \frac{1}{4} \mp \Theta\rbra{\varepsilon}$ for every block-encoding $U_{\rho_{\pm}}$ of $\frac{1}{2}\rho_{\pm}$ to precision $\varepsilon$ using $Q$ queries. 
    By \cref{corollary:qsd-intro}, we conclude that $Q \geq \mathsf{Q}_\diamond\rbra{\textsc{Dis}_{\rho_+, \rho_-}} = \widetilde \Omega\rbra{1/\sqrt{\gamma}} = \widetilde\Omega\rbra{1/\varepsilon}$.
\end{proof}

Compared to known proofs for amplitude estimation, our proof is relatively simple 
by observing that block-encodings of quantum states store their amplitudes.

\textbf{Matrix spectrum testing.}
We also obtain new quantum query lower bounds for the spectrum testing of matrices block-encoded in unitary operators, which are directly reduced from the spectrum testing of quantum states \cite{OW21}. 

\begin{corollary} [Matrix spectrum testing, Corollaries \ref{corollary:rank-testing}, \ref{corollary:mixed-testing}, and \ref{corollary:uniformity-distinguishing}, informal]
\label{corollary:spectrum-testing}
    Let $U$ be a block-encoding of an $N$-dimensional Hermitian matrix $A$.
    \begin{itemize}
        \item \emph{Rank Testing}: Testing whether $A$ has rank $r$ requires $\widetilde \Omega\rbra{\sqrt{r}}$ queries. 
        \item \emph{Mixedness Testing}: Testing whether $A$ has the uniform spectrum $\rbra{\frac{1}{N}, \dots, \frac{1}{N}}$ requires $\widetilde \Omega\rbra{\sqrt{N}}$ queries.
        \item \emph{Uniformity Testing}: Testing whether $A$ has a spectrum uniform on $r$ or $r+\Delta$ eigenvalues requires $\Omega^*\rbra{r/\sqrt{\Delta}}$ queries.\footnote{$\Omega^*\rbra{\cdot}$ suppresses quasi-logarithmic factors such as $r^{o\rbra{1}}$ and $\rbra{1/\varepsilon}^{o\rbra{1}}$.}
    \end{itemize}
\end{corollary}

\begin{proof} [Proof sketch]
    These results are reduced from the sample lower bounds $\Omega\rbra{r}$, $\Omega\rbra{N}$, and $\Omega^*\rbra{r^2/\Delta}$ for the corresponding quantum state testing problems studied in \cite{OW21}, respectively.
\end{proof}

\subsection{Related work}

The quantum polynomial method \cite{BBC+01} extends the polynomial method for classical circuit complexity (cf. \cite{Bei93}), which was employed to show numerous matching quantum lower bounds, e.g., parity checking (see also \cite{FGGS98}), collision finding and element distinctness \cite{AS04} (see \cite{BHT98a,Amb07} for their upper bounds), and property testing of probability distributions \cite{BKT20} (e.g., $k$-junta testing \cite{ABRdW16}, uniformity testing \cite{BHH11}, and entropy estimation \cite{LW19}).
Recently, the polynomial method was generalized for unitary property testing by \cite{SY23}.

The adversary method \cite{Amb02} is another powerful method for proving quantum lower bounds, which was also used to show a wide range of quantum lower bounds, e.g., AND of ORs, surjectivity \cite{BM12,She18}, and some basic graph problems \cite{DHHM06}. 
It is worth noting that, in \cite{Amb06}, the adversary method is proven to be more powerful than the polynomial method on certain problems.
Moreover, the quantum query complexity for Boolean functions was shown in \cite{Rei11} to be fully characterized by the general adversary bound only up to a constant factor.

The compressed oracle method \cite{Zha19} is a method for proving average-case quantum lower bounds, 
which originated from quantum cryptology and have been recently extended to handling parallel queries \cite{CFHL21} and random permutation oracles \cite{Unr23,MMW25}.
This method was used to show quantum lower bounds for several problems, e.g., collision finding \cite{HM23,HLS24}, inverting functions \cite{CGLQ20} and permutations \cite{Ros21}, search with noisy oracle \cite{Ros23}, and amplitude estimation without inverses \cite{TW25}.

As a lifting theorem, our sample-to-query lifting is analogous to the query-to-communication lifting that converts lower bounds in query complexity to lower bounds in communication complexity. 
Classical query-to-communication lifting has been extensively studied for deterministic and randomized communication complexities, e.g., \cite{GPW18,CKLM18,GPW20,CFK+19}.
For the quantum case, related results include that quantum communication complexity can be lifted from approximate degrees \cite{She11} and the adversary bound \cite{ABDK21}.

\subsection{Discussion}

In this paper, we propose a quantum sample-to-query lifting theorem that reveals a relation between the sample complexity of quantum state testing and the query complexity of unitary property testing. 
It provides a novel method for proving query lower bounds on quantum query algorithms, which is quite different from known methods such as the polynomial method \cite{BBC+01}, the adversary method \cite{Amb02}, and the compressed oracle method \cite{Zha19}.
In particular, we obtain new lower bounds for quantum Gibbs sampling (\cref{thm:gibbs-query-lower-bound}) and the entanglement entropy problem (\cref{corollary:entro-intro}). 
We believe that our lifting theorem brings new insights regarding quantum sample and query complexities and will  serve as a new tool for proving quantum query lower bounds.

We conclude by mentioning several open questions related to our work. 

\begin{itemize}
    \item Can we improve the logarithmic factors in our lifting theorem (\cref{thm:lifting})? 
    \item Our quantum query lower bound for Gibbs sampling (\cref{thm:gibbs-query-lower-bound}) holds when the oracle is a block-encoding of the Hamiltonian. Can we prove a stronger lower bound when given query access to the classical data of the (e.g., sparse) Hamiltonian? 
    \item We provide a quantum query lower bound for the entanglement entropy problem (\cref{corollary:entro-intro}) different from the one given in \cite{SY23} by the guessing lemma \cite{Aar12,AKKT20}, while they also considered the $\mathsf{QMA}$ lower bound for this problem. Can we extend our method to obtain a different $\mathsf{QMA}$ lower bound?
    \item Is there a classical analog of the quantum sample-to-query lifting theorem? 
    If it exists, we believe it could be a new tool for proving classical lower bounds. 
    However, the difficulty is that it is not clear how to interpret ``block-encode'' in classical computing, namely, how to ``block-encode'' classical information in classical logic gates.
    \item In contrast to the sample-to-query lifting considered in this paper, can we do the reverse? Namely, can we reduce a unitary property testing problem (with query access) to a quantum state testing problem (with sample access)? 
    We are only aware of a simple example via the Choi-Jamio{\l}kowski isomorphism \cite{Cho75,Jam72} as discussed in \cref{remark:choi}. 
    \item How does our ``lifting'' method for proving quantum query lower bounds compare to the polynomial method \cite{BBC+01}, the adversary method \cite{Amb02}, and the compressed oracle method \cite{Zha19}?
    The relationship between the polynomial and adversary methods has been extensively investigated in the literature.
    Prior work mainly focuses on the quantum query complexity of Boolean functions, where the negative-weighted version of the adversary method \cite{LHS07} is known to be tight \cite{RS12,Rei09,Rei11} while the polynomial method is known to be non-tight \cite{Amb06,AB23}. 
    Direct reductions from the polynomial method to the adversary method were shown in \cite{MR13,Bel24}. 
    On the other hand, the compressed oracle method turns out to be more useful for the average-case complexity, e.g., \cite{HM23,HLS24,CGLQ20,Ros21,Ros23,TW25}.
    In comparison, our lifting method is able to provide simple proofs of quantum query lower bounds for certain problems previously given by the polynomial and adversary methods (e.g., \cref{corollary:hamsim-intro,corollary:phe-intro,corollary:ampl-intro}), though with a loss of logarithmic factors.
    Meanwhile, we are not aware of a general approach using our lifting method to prove quantum query lower bounds for Boolean functions. 
\end{itemize}

\subsection{Recent Developments}

In parallel with this work, Chen, Kastoryano, Brand{\~a}o, and Gily{\'e}n \cite{CKBG23} showed a matching lower bound $\Omega\rbra{\beta}$ on the \textit{black-box Hamiltonian simulation time complexity} of quantum Gibbs sampling, which, compared to \cref{thm:gibbs-query-lower-bound}, also implies an $\widetilde \Omega\rbra{\beta}$ query lower bound if given the block-encoding access to the Hamiltonian. 

After the work of this paper, Weggemans \cite{Weg24} provided a query lower bound for unitary channel discrimination in the diamond norm.
Based on this, the author removed the logarithmic factors in the lower bounds given in this paper that are reduced from quantum state discrimination (\cref{corollary:qsd-intro}), e.g., quantum Gibbs sampling (\cref{thm:gibbs-query-lower-bound}) and the entanglement entropy problem (\cref{corollary:entro-intro}). 
In addition, the author showed that there exists a quantum oracle relative to which $\mathsf{QMA}\rbra{2} \not\supseteq \mathsf{SBQP}$ and $\mathsf{QMA}/\mathsf{qpoly} \not\supseteq \mathsf{SBQP}$.

Using the quantum sample-to-query lifting theorem presented in this paper, Chen, Wang, and Zhang \cite{CWZ24} further showed a quantum query lower bound $\widetilde\Omega\rbra{\sqrt{N/\Delta}}$ for the entanglement entropy problem (\cref{corollary:entro-intro}), simultaneously improving both the lower bound in $N$ due to \cite{SY23} and the lower bound in $\Delta$ due to \cite{Weg24} and this work. 
Their improvement is achieved by using the quantum sample lower bound for estimating the entanglement R\'enyi entropy of bipartite pure quantum states other than quantum state discrimination (\cref{corollary:qsd-intro}). 

In addition to quantum lower bounds, the quantum sample-to-query lifting theorem has inspired new quantum algorithms. 
Wang and Zhang \cite{WZ24b} extended the quantum sample-to-query lifting theorem to an algorithmic tool called \textit{samplizer}, which can be used to convert a quantum query algorithm with query complexity $Q$ to a quantum sample algorithm with sample complexity $\widetilde O\rbra{Q^2/\delta}$, where the two quantum algorithms are $\delta$-close in the diamond norm distance. 
Using the samplizer, they further developed time-efficient quantum estimators for von Neumann entropy and R\'enyi entropy. 
The samplizer later turned out to be useful in several other applications:
\begin{itemize}
    \item Liu, Wang, Wilde, and Zhang \cite{LWWZ24} developed a quantum estimator for the fidelity between well-conditioned quantum states. 
    \item Liu and Wang \cite{LW25} developed a quantum estimator for the Tsallis entropy of quantum states, exponentially improving the prior approaches known or implied in \cite{AISW20,WGL+22,WZL24,WZ24b}. 
    \item Liu and Wang \cite{LW25b} developed a quantum estimator for the $\ell_\alpha$ distance between quantum states, exponentially improving the prior approach in \cite{WGL+22}.
    \item Niwa, Rossi, Taranto, and Murao \cite{NRTM25} developed a singular value transformation scheme for quantum channels, which has applications to testing if a quantum channel is entanglement breaking. 
\end{itemize}
It is worth noting that later Wang and Zhang \cite{WZ25} investigated the samplizer especially for pure states, polylogarithmically improving the general samplizer in \cite{WZ24b}.
Using the samplizer for pure states, they further discovered a sample-optimal (up to a constant factor) quantum estimator for pure-state trace distance and square root fidelity with sample complexity $\Theta\rbra{1/\varepsilon^2}$; 
in comparison, query-optimal quantum algorithms for this task were given in \cite{Wan24,FW25} with query complexity $\Theta\rbra{1/\varepsilon}$. 

\subsection{Organization of This Paper}

In \cref{sec:preliminaries}, we include necessary preliminaries for quantum computing. 
In \cref{sec:lifting}, we prove the quantum sample-to-query lifting theorem. 
In \cref{sec:query-lower-bound}, we give formal proofs for a series of quantum query lower bounds that are based on the quantum sample-to-query lifting theorem.

\section{Preliminaries}
\label{sec:preliminaries}

In this section, we introduce the notations for quantum computation and provide some useful tools that will be used in our analysis.

\subsection{Basic notations}

Consider a finite-dimensional Hilbert space $\mathcal{H}$. 
A quantum state in $\mathcal{H}$ is denoted by $\ket{\psi}$. 
The inner product of two quantum states $\ket\psi$ and $\ket\varphi$ is denoted by $\braket{\psi}{\varphi}$. 
The norm of a quantum state is defined by $\Abs{\ket{\psi}} = \sqrt{\braket{\psi}{\psi}}$. 
We denote an $a$-qubit quantum state $\ket{\psi}$ as $\ket{\psi}_a$ with the subscript $a$ indicating the label (and also the number of qubits) of the quantum system; and we write $\bra{\psi}_a$ for the conjugate of $\ket{\psi}_a$.
For example, we write $\ket{0}_a = \ket{0}^{\otimes a}$ and $\ket{0}_a \ket{0}_b = \ket{0}^{\otimes a} \otimes \ket{0}^{\otimes b}$.

A mixed quantum state in $\mathcal{H}$ can be described by a density operator $\rho$ with $\tr\rbra{\rho} = 1$ and $\rho \sqsupseteq 0$, where $\sqsupseteq$ is the L\"owner order, i.e., $A \sqsupseteq B$ if $A-B$ is positive semidefinite. 
Let $\mathcal{D}\rbra{\mathcal{H}}$ denote the set of all density operators in $\mathcal{H}$. 
The trace distance between two mixed quantum states $\rho$ and $\sigma$ is defined by
\[
\frac 1 2 \Abs{\rho - \sigma}_1,
\]
where $\Abs{A}_1 = \tr\rbra{\sqrt{A^\dag A}}$ and $A^\dag$ is the Hermitian conjugate of $A$. 
The fidelity between $\rho$ and $\sigma$ is defined by
\[
\operatorname{F}\rbra*{\rho, \sigma} = \tr\rbra*{\sqrt{\sqrt{\sigma}\rho\sqrt{\sigma}}}.
\]
The following is a relationship between fidelity and trace distance. 
\begin{theorem} [{\cite[Theorem 1]{FvdG99}}]
\label{thm:fidelity-vs-td}
For any two quantum states $\rho$ and $\sigma$, 
\[
1 - \operatorname{F}\rbra{\rho, \sigma} \leq \frac 1 2 \Abs*{\rho - \sigma}_1 \leq \sqrt{1 - \rbra{\operatorname{F}\rbra{\rho, \sigma}}^2}.
\]
\end{theorem}
A quantum gate or quantum oracle is a unitary operator $U$ such that $U^\dag U = UU^\dag = I$, where $I$ is the identity operator. 
The operator norm of an operator $A$ is defined by
\[
\Abs{A} = \sup_{\Abs{\ket{\psi}} = 1} \Abs{A \ket{\psi}}.
\]

For two quantum channels $\mathcal{E}$ and $\mathcal{F}$ that act on $\mathcal{D}\rbra{\mathcal{H}}$, the trace norm distance between 
 $\mathcal{E}$ and $\mathcal{F}$ is defined by
\[
\Abs{\mathcal{E} - \mathcal{F}}_{\tr} = \sup_{\varrho \in \mathcal{D}\rbra{\mathcal{H}}} \Abs{ \mathcal{E}\rbra{\varrho} - \mathcal{F}\rbra{\varrho} }_1.
\]
The diamond norm distance between 
$\mathcal{E}$ and $\mathcal{F}$ is defined by
\[
    \Abs{\mathcal{E} - \mathcal{F}}_{\diamond} = \sup_{\varrho \in \mathcal{D}\rbra{\mathcal{H} \otimes \mathcal{H}'}} \Abs{ \rbra*{\mathcal{E} \otimes \mathcal{I}_{\mathcal{H}'}}\rbra{\varrho} - \rbra*{\mathcal{F} \otimes \mathcal{I}_{\mathcal{H}'}}\rbra{\varrho} }_1,
\]
where the supremum is chosen over all finite-dimensional Hilbert space $\mathcal{H}'$. We note that the following inequalities between the trace norm distance and the diamond norm distance always hold:
\begin{equation}
\label{eq:tr-vs-diamond}
    \frac{1}{\dim\rbra{\mathcal{H}}} \Abs{\mathcal{\mathcal{E} - \mathcal{F}}}_{\diamond} \leq \Abs{\mathcal{\mathcal{E} - \mathcal{F}}}_{\tr} \leq \Abs{\mathcal{\mathcal{E} - \mathcal{F}}}_{\diamond}.
\end{equation}

\subsection{Quantum state discrimination}

We need an information-theoretic lower bound for the success probability of quantum hypothesis testing between two quantum states in terms of the trace distance between them, which is originated from \cite{Hel67,Hol73}.

\begin{theorem} [Quantum state discrimination, cf. {\cite[Section 9.1.4]{Wil13}}]
\label{thm:HH-measurement}
    Suppose $\rho_0$ and $\rho_1$ are two quantum states. 
    Let $\varrho$ be a quantum state such that $\varrho = \rho_0$ or $\varrho = \rho_1$ with equal probability. 
    For any POVM $\Lambda = \cbra{\Lambda_0, \Lambda_1}$,
    the success probability of distinguishing the two cases is 
    \[
    p = \frac 1 2 \tr\rbra*{\Lambda_0\rho_0} + \frac 1 2 \tr\rbra*{\Lambda_1\rho_1} \leq \frac{1}{2}\rbra*{1+\frac{1}{2}\Abs{\rho_0-\rho_1}_1}. 
    \]
\end{theorem}

Now we are ready to derive the sample lower bound for $\textsc{Dis}_{\rho, \sigma}$. 
Recall that $\textsc{Dis}_{\rho, \sigma}$ is a promise problem for quantum state discrimination for which $\rho$ is the only \textit{yes} instance and $\sigma$ is the only \textit{no} instance.

\begin{lemma}
\label{lemma:dis-sample}
    $\mathsf{S}\rbra{\textsc{Dis}_{\rho, \sigma}} = \Omega\rbra{1/\gamma}$, where $\gamma = 1 - \operatorname{F}\rbra{\rho, \sigma}$ is the infidelity.
\end{lemma}
\begin{proof}
    Suppose that there is a quantum tester $\mathcal{T}$ for $\textsc{Dis}_{\rho, \sigma}$ with sample complexity $S = \mathsf{S}\rbra{\textsc{Dis}_{\rho, \sigma}}$.
    Then, it means that $\mathcal{T}$ can be used to distinguish between $\rho^{\otimes S}$ and $\sigma^{\otimes S}$ with success probability $p_{\textup{succ}} \geq 2/3$. 
    On the other hand, by \cref{thm:HH-measurement}, we have
    \[
    p_{\textup{succ}} \leq \frac{1}{2}\rbra*{1+\frac{1}{2}\Abs*{\rho^{\otimes S}-\sigma^{\otimes S}}_1}.
    \]
    Moreover, by \cref{thm:fidelity-vs-td}, we have
    \begin{align*}
    p_{\textup{succ}}
    & \leq \frac 1 2 \rbra*{1 + \sqrt{1 - \rbra*{\operatorname{F}\rbra{\rho^{\otimes S}, \sigma^{\otimes S}}}^2}} \\
    & = \frac 1 2 \rbra*{1 + \sqrt{1 - \rbra*{F\rbra{\rho, \sigma}}^{2S}}} \\
    & = \frac 1 2 \rbra*{1+\sqrt{1-\rbra{1-\gamma}^{2S}}},
    \end{align*}
    which gives $\sqrt{1 - \rbra{1-\gamma}^{2S}} \geq 1/6$.
    This further implies that
    \[
    \frac{35}{36} \geq \rbra{1-\gamma}^{2S} \geq 1 - 2S\gamma,
    \]
    and thus $S \geq 1/72\gamma = \Omega\rbra{1/\gamma}$.
\end{proof}

\subsection{Quantum query algorithms}

A quantum query algorithm $\mathcal{A}$ with quantum query complexity $Q$ is described by a quantum circuit $\mathcal{A}^U$ that uses $Q$ queries to the quantum unitary oracle $U$:
\[
\mathcal{A}^U = G_Q U_Q \dots G_2 U_2 G_1 U_1 G_0,
\]
where each $U_j$ is either (controlled-)$U$ or (controlled-)$U^\dag$ and each $G_j$ is a unitary operator (consisting of one- and two-qubit quantum gates) that do not depend on $U$. 
Throughout this paper, for simplicity, a query to $U$ means a query to $U$, $U^\dag$, controlled-$U$, or controlled-$U^\dag$.

Quantum amplitude estimation is a basic quantum query algorithm, which is usually used as a subroutine in other algorithms. 

\begin{theorem} [Quantum amplitude estimation, {\cite[Theorem 12]{BHMT02}}]
\label{thm:ampest}
    Suppose that $U$ is a unitary operator such that
    \[
    U \ket{0} = \sqrt{p} \ket{0} \ket{\phi_0} + \sqrt{1-p} \ket{1} \ket{\phi_1},
    \]
    where $\ket{\phi_0}$ and $\ket{\phi_1}$ are normalized pure quantum states and $p \in \sbra{0, 1}$. 
    There is a quantum query algorithm that outputs $\tilde p \in \sbra{0, 1}$ such that 
    \[
    \abs*{\tilde p - p} \leq \frac{2\pi\sqrt{p\rbra{1-p}}}{M} + \frac{\pi^2}{M^2}
    \]
    with probability at least $8/\pi^2$, using $M$ queries to $U$. 
    In particular, we can estimate $p$ to within additive error $\varepsilon$ by letting $M = O\rbra{1/\varepsilon}$. 
\end{theorem}

\subsection{Block-encoding}

Block-encoding is a useful tool to describe information encoded in a matrix block in a quantum unitary operator, which was initially developed for Hamiltonian simulation \cite{LC19,CGJ19} and was later systematically studied in \cite{GSLW19}. 

\begin{definition} [Block-encoding]
\label{def:block-encoding}
    Suppose that $A$ is an $n$-qubit operator, $\alpha, \varepsilon \geq 0$ and $a \in \mathbb{N}$. 
    An $\rbra{n+a}$-qubit unitary operator $B$ is an $\rbra{\alpha, a, \varepsilon}$-block-encoding of $A$, if 
    \[
    \Abs*{\alpha \bra{0}^{\otimes a} B \ket{0}^{\otimes a} - A} \leq \varepsilon.
    \]
    Especially when $\alpha = 1$ and $\varepsilon = 0$, we may simply call $B$ a block-encoding of $A$ (if the parameter $a$ is clear or unimportant in context).
\end{definition}

Intuitively, if $B$ is an $(\alpha,a,\varepsilon)$-block-encoding of $A$, then the upper left block of the matrix $B$ is close to $A/\alpha$, i.e.,
\[
B \approx \begin{bmatrix}
    A/\alpha & * \\
    * & *
\end{bmatrix}.
\]

The following lemma shows how to construct a block-encoding of the product of two block-encoded operators. 

\begin{lemma} [Product of block-encoded operators, {\cite[Lemma 30]{GSLW19}}]
\label{lemma:product-block-encoding}
    Suppose $\rbra{n+a}$-qubit unitary operator $U$ is an $\rbra{\alpha, a, \delta}$-block-encoding of $n$-qubit operator $A$, and $\rbra{n+b}$-qubit unitary operator $V$ is a $\rbra{\beta, b, \varepsilon}$-block-encoding of $n$-qubit operator $B$. Then, $\rbra{U \otimes I_b} \rbra{V \otimes I_a}$ is an $\rbra{\alpha\beta, a+b, \alpha\varepsilon+\beta\delta}$-block-encoding of $AB$. 
\end{lemma}

With \cref{lemma:product-block-encoding}, we can scale down the block-encoded operator. 

\begin{lemma} [Down-scaling of block-encoded operators]
\label{corollary:scaling-block-encoding}
    Suppose $\rbra{n+a}$-qubit unitary operator $U$ is a $\rbra{1, a, \varepsilon}$-block-encoding of $n$-qubit operator $A$. Then, for every $\alpha > 1$, there is an $\rbra{n+a+1}$-qubit unitary operator $U_\alpha$ that is a $\rbra{1, a+1, \varepsilon/\alpha}$-block-encoding of $A/\alpha$, using $1$ query  to $U$ and $O(1)$ one- and two-qubit quantum gates.
\end{lemma}
\begin{proof}
    This is directly obtained from the construction of \cref{lemma:product-block-encoding} by letting 
    \[
    V = \begin{bmatrix}
        1/\alpha & -\sqrt{1 - 1/\alpha^2} \\
        \sqrt{1 - 1/\alpha^2} & 1/\alpha
    \end{bmatrix} \otimes I_n,
    \]
    which is a $\rbra{1, 1, 0}$-block-encoding of $I_n/\alpha$.
    Here, $V$ is a one-qubit rotation and can be implemented by $O(1)$ one- and two-qubit quantum gates.
    We can verify that $W = \rbra{U_{n+a} \otimes I_{1}}\rbra{V_{n+1} \otimes I_a}$ is a $\rbra{1, a+1, \varepsilon/\alpha}$-block-encoding of $A/\alpha$, where the subscripts $n, a, 1$ denote the quantum system that the unitary operator acts on.
    The detailed calculation is given as follows.
    \begin{align*}
        \Abs*{ \frac{A}{\alpha} - \bra{0}_{a+1} W \ket{0}_{a+1} }
        & = \Abs*{ \frac{A}{\alpha} - \bra{0}_{a+1} \rbra{U_{n+a} \otimes I_{1}}\rbra{V_{n+1} \otimes I_a} \ket{0}_{a+1} } \\
        & = \Abs*{ \frac{A}{\alpha} - \bra{0}_{a} U_{n+a} \ket{0}_a \cdot \bra{0}_{1} V_{n+1} \ket{0}_1 } \\
        & = \Abs*{ \frac{A}{\alpha} - \bra{0}_{a} U_{n+a} \ket{0}_a \cdot \frac{I_n}{\alpha} } \\
        & = \frac{1}{\alpha} \Abs[\big]{ A - \bra{0}_{a} U_{n+a} \ket{0}_a } \\
        & \leq \frac{\varepsilon}{\alpha}.
    \end{align*}
\end{proof}

Robust oblivious amplitude amplification \cite{BCK15} can be seen as a method for scaling up a block-encoded operator when it is unitary. 
Here, we use the version from \cite{GSLW19}.

\begin{lemma} [Robust oblivious amplitude amplification, {\cite[Theorem 15]{GSLW19}}]
    \label{lmm:roaa}
    Suppose that $n$ is an odd positive integer, $\varepsilon \in \rbra{0, 1}$, and $U$ is a unitary operator that is a $\rbra{1, a, \varepsilon}$-block-encoding of $\sin\rbra{\frac{\pi}{2n}} \cdot W$, where $W$ is a unitary operator. 
    Then, there is a quantum circuit $\widetilde U$ that is a $\rbra{1, a+1, 2n\varepsilon}$-block-encoding of $W$, using $O\rbra{n}$ queries to $U$. 
\end{lemma}

We also provide a method for scaling up block-encoded operators in the general case, which can be seen as a generalization of {\cite[Theorem 2]{LC17a}}, {\cite[Theorem 15]{GSLW19}}, and {\cite[Lemma 12]{LS24}}. 

\begin{lemma}[Up-scaling of block-encoded operators] \label{corollary:up-scaling-block-encoding}
    Let $U$ be a unitary operator that is an $\rbra{\alpha, a, \varepsilon}$-block-encoding of $n$-qubit operator $A$ with $\alpha > 1$, $\varepsilon \in [0, 1]$ and $\Abs{A} \leq 1$. 
    Then, for every $1 < \beta < \alpha$ and $0 < \varepsilon' < \min\cbra{\frac{\beta}{\alpha}, \frac{1}{2}}$, there is a quantum circuit $\widetilde U$ that is a $\rbra{\beta, a+1, \Theta\rbra{\sqrt{\alpha\varepsilon} \log\rbra{\frac{1}{\varepsilon'}}} + \varepsilon'}$-block-encoding of $A$, using $O\rbra{\frac{\alpha}{\beta-1} \log\rbra{\frac{1}{\varepsilon'}}}$ queries to $U$. 
    If we further have $\varepsilon = 0$ and 
    $\varepsilon' \leq 2\sqrt{2\beta^2+2\beta}-2\beta-2$,
    then there is a unitary operator $U'$ that is a $\rbra{\beta, a+3, 0}$-block-encoding of $A$ satisfying
    \[
    \Abs*{\widetilde U \otimes I_2 - U'} \leq \frac{2\sqrt{2}\varepsilon'}{\sqrt{4\beta^2 - \rbra*{2+\varepsilon'}^2}},
    \]
    where $I_2$ is the identity operator acting on $2$ qubits.
\end{lemma}
\begin{proof}
    See \cref{sec:up-scaling}.
\end{proof}

We also need the linear combination of unitaries (LCU) lemma \cite{CW12,BCC+15}. 

\begin{definition} [State-preparation-pair]
    \label{def:state-preparation-pair}
    Let $y \in \mathbb{C}^m$ and $\Abs{y}_1 \leq \beta$. 
    The pair of unitary operators $\rbra{P_L, P_R}$ is a $\rbra{\beta, b, \varepsilon}$-state-preparation-pair for $y$, if $P_L \ket{0}^{\otimes b} = \sum_{j = 0}^{2^b-1} c_j \ket{j}$ and $P_R \ket{0}^{\otimes b} = \sum_{j = 0}^{2^b-1} d_j \ket{j}$ such that $\sum_{j=0}^{m-1} \abs{\beta c_j^* d_j - y_j} \leq \varepsilon$ and $c_j^* d_j = 0$ for all $m \leq j < 2^b$.
\end{definition}

\begin{lemma} [Linear combination of block-encoded operators, {\cite[Lemma 29]{GSLW19}}]
\label{lemma:lcu}
    Let $A = \sum_{j=0}^{m-1} y_jA_j$ be an $s$-qubit operator and $\varepsilon_1, \varepsilon_2 > 0$.
    Suppose that $\rbra{P_L, P_R}$ is a $\rbra{\beta, b, \varepsilon_1}$-state-preparation-pair for $y$, and each $U_j$ is an $\rbra{\alpha, a, \varepsilon_2}$-block-encoding of $A_j$ for $0\leq j\leq m-1$. 
    Then, we can implement a quantum circuit that is an $\rbra{\alpha\beta, a+b, \alpha\varepsilon_1 + \alpha\beta\varepsilon_2}$-block-encoding of $A$, using $1$ query to each of $P_L$, $P_R$ and $U_j$ for $0\leq j\leq m-1$.
\end{lemma}

We provide a lemma that allows one to approximate the algorithm $\mathcal{A}^U$ if an approximate block-encoding of $U$ is given. 

\begin{lemma} [Block-encoding quantum oracles]
\label{lemma:block-encoding-as-query}
    Suppose that $\mathcal{A}^U$ is a quantum query algorithm using $Q$ queries to quantum oracle $U$.
    Let $V$ be a $\rbra{1, a, \varepsilon}$-block-encoding of $U$.
    Then, we can implement a quantum circuit $\mathcal{A}^{U \to V}$ that is a $\rbra{1, Qa, Q\varepsilon}$-block-encoding of $\mathcal{A}^U$, using $O\rbra{Q}$ queries to $V$. 
\end{lemma}

\cref{lemma:block-encoding-as-query} explains how to implement the quantum query algorithm $\mathcal{A}^U$ when it is only given a block-encoding of $U$ (possibly with errors), which will be used in some of our applications such as \cref{thm:entangle-entro} and \cref{thm:ph-est}. 
Here, the notation $\mathcal{A}^{U \to V}$ means a quantum query algorithm with query access to $V$ that approximately simulates the quantum query algorithm $\mathcal{A}^U$ with query access to $U$, where $V$ is a block-encoding of $U$ with errors.

\begin{proof} [Proof of \cref{lemma:block-encoding-as-query}]
    Note that $V^\dag$ is a $\rbra{1, a, \varepsilon}$-block-encoding of $U^\dag$, and similarly, controlled-$V$ is a $\rbra{1, a, \varepsilon}$-block-encoding of controlled-$U$. 
    Suppose that $\mathcal{A}^U = G_Q U_Q \dots G_2 U_2 G_1 U_1 G_0$, where $U_j$ is either (controlled-)$U$ or (controlled-)$U^\dag$ and $G_j$ consists of one- and two-qubit quantum gates that do not depend on $U$. 
    Then, we can implement $V_j$ that is a $\rbra{1, a, \varepsilon}$-block-encoding of $U_j$, using $1$ query to $V$. 
    Also note that $G_j$ is a $\rbra{1, 0, 0}$-block-encoding of $G_j$ (itself). 
    By applying \cref{lemma:product-block-encoding} multiple times, we can implement a quantum circuit $\mathcal{A}^{U \to V}$ such that it is a $\rbra{1, Qa, Q\varepsilon}$-block-encoding of $\mathcal{A}^U$, using $Q$ queries to $V$. 
\end{proof}

\subsection{Quantum eigenvalue transformation}

The technique of polynomial eigenvalue transformation of unitary operators will be used as a key tool in our analysis.

\begin{theorem} [Polynomial eigenvalue transformation, {\cite[Theorem 31]{GSLW19}}]
\label{thm:qsvt}
    Suppose unitary operator $U$ is an $\rbra{\alpha, a, \varepsilon}$-block-encoding of Hermitian operator $A$. If $\delta > 0$ and $p\rbra{x} \in \mathbb{R}\sbra{x}$ is a polynomial of degree $d$ such that
    \[
    \forall x \in \sbra{-1, 1}, \quad \abs*{p\rbra{x}} \leq \frac 1 2,
    \]
    then there is a unitary operator $\widetilde U$ which is a $\rbra{1, a+2, 4d\sqrt{\varepsilon/\alpha} + \delta}$-block-encoding of $p\rbra{A/\alpha}$, using $O\rbra{d}$ queries to $U$, and $O\rbra{\rbra{a+1}d}$ one- and two-qubit quantum gates. 
    Moreover, the description of $\widetilde U$ can be computed classically in $\poly\rbra{d, \log\rbra{1/\delta}}$ time. 
\end{theorem}

Based on \cref{thm:qsvt}, a quantum query algorithm for Hamiltonian simulation is developed in \cite{GSLW19}. 

\begin{theorem} [Hamiltonian simulation, {\cite[Corollary 32]{GSLW19}}]
\label{thm:hamiltonian-simulation}
    Suppose $U$ is a unitary operator that is a $\rbra{1, a, 0}$-block-encoding of Hamiltonian $H$.
    For every $\varepsilon \in \rbra{0, 1/2}$ and $t \in \mathbb{R}$, it is necessary and sufficient to use 
    \[
    \Theta\rbra*{\abs*{t} + \frac{\log\rbra{1/\varepsilon}}{\log\rbra{e + \log\rbra{1/\varepsilon}/\abs*{t}}}}
    \]
    queries to $U$ to implement a unitary operator $V$ that is a $\rbra{1, a+2, \varepsilon}$-block-encoding of $e^{-iHt}$. 
\end{theorem}

To apply polynomial eigenvalue transformation (\cref{thm:qsvt}), we need some useful polynomial approximations. 
The following is a polynomial approximation of positive power function.

\begin{lemma} [Polynomial approximation of positive power functions, {\cite[Corollary 3.4.14]{Gil19}}]
\label{lemma:positive-power}
    Suppose $\delta, \varepsilon \in \rbra{0, 1/2}$ and $c \in \rbra{0, 1}$. Then, there is an efficiently computable even polynomial $p\rbra{x} \in \mathbb{R}\sbra{x}$ of degree $O\rbra*{\frac{1}{\delta}\log\rbra*{\frac{1}{\varepsilon}}}$ such that
    \begin{align*}
        & \forall x \in \sbra*{-1, 1}, \quad \abs*{p\rbra*{x}} \leq 1, \\
        & \forall x \in \sbra*{\delta, 1}, \quad \abs*{p\rbra*{x} - \frac 1 2 x^c} \leq \varepsilon.
    \end{align*}
\end{lemma}

\subsection{Block-encoding of density operators}

The following lemma introduces a method to construct a quantum circuit that is a block-encoding of a quantum state, given a quantum oracle that prepares its purification.

\begin{lemma} [Block-encoding of density operators, {\cite[Lemma 25]{GSLW19}}]
\label{lemma:block-encoding-of-density}
    Suppose that an $\rbra{n+a}$-qubit unitary operator $\mathcal{O}_\rho$ prepares a purification of an $n$-qubit quantum state $\rho$. Then, there is a $\rbra{2n+a}$-qubit unitary operator $U_\rho$, which is a $\rbra{1, n+a, 0}$-block-encoding of $\rho$, using $O\rbra{1}$ queries to $\mathcal{O}_\rho$ and $O\rbra{a}$ one- and two-qubit quantum gates. 
\end{lemma}

For our purpose, we give a generalized version of \cref{lemma:block-encoding-of-density} as follows. 

\begin{lemma}
    \label{lmm:blk den}
    Let $\rho$ be an $n$-qubit density operator and $V$ be an $(n+a)$-qubit operator such that  $V\ket{0}_a\ket{0}_n=\ket{\rho}$,
    where
    $\ket{\rho}$ is a purification of $\rho$,
    i.e.,
    $\tr_a(\ket{\rho}\bra{\rho})=\rho$.
    Suppose that $U$ is a $(\beta, b, \varepsilon)$-block-encoding of $V$, then we can construct a quantum circuit $\widetilde U$ that is a $(\beta^2,n+a+2b,2\varepsilon)$-block-encoding of $\rho$,
    using $O(1)$ queries to $U$
    and $O(n)$ one- and two-qubit quantum gates.
\end{lemma}

\begin{proof}
    The proof of \cref{lemma:block-encoding-of-density} in \cite{GSLW19} has shown that
    \begin{equation*}
        W=\rbra*{V^\dagger\otimes I_n}
        \rbra*{I_a\otimes \text{SWAP}_n}
        \rbra*{V\otimes I_n}
    \end{equation*}
    is a $(1,n+a,0)$-block-encoding of $\rho$,
    as long as $V\ket{0}_{n+a}=\ket{\rho}$.
    Let $b_1,b_2$ be two copies of a system of $b$ qubits,
    with $b_1=b_2=b$.
    We now show that
    \begin{equation*}
        \widetilde U =
        \rbra*{U^\dagger\otimes I_n\otimes I_{b_2}}
        \rbra*{I_a\otimes \text{SWAP}_n\otimes I_{b_1}\otimes I_{b_2}}
        \rbra*{U\otimes I_n\otimes I_{b_1}}
    \end{equation*}
    is a $(\beta^2,b_1+b_2,2\varepsilon)$-block-encoding of $W$,
    and is therefore a $(\beta^2,n+a+2b,2\varepsilon)$-block-encoding of $\rho$.
    To see this, note that
    \begin{align*}
        & \beta^2 \bra{0}_{b_1+b_2} \widetilde U \ket{0}_{b_1+b_2} \\
        &=\beta^2\bra{0}_{b_1+b_2}
        \rbra*{U^\dagger\otimes I_n\otimes I_{b_2}}
        \rbra*{I_a\otimes \text{SWAP}_n\otimes I_{b_1}\otimes I_{b_2}}
        \rbra*{
        U\otimes I_n\otimes I_{b_1}}
        \ket{0}_{b_1+b_2}\\
        &=
        \rbra*{\beta\bra{0}_{b_1}U^\dagger\ket{0}_{b_1}\otimes I_n}
        \rbra*{I_a\otimes \text{SWAP}_n}
        \rbra*{\beta\bra{0}_{b_2}U\ket{0}_{b_2}\otimes I_s}\\
        &=
        \rbra*{\widetilde{V}^\dag\otimes I_n}
        \rbra*{I_a\otimes \text{SWAP}_n}
        \rbra*{\widetilde{V}\otimes I_n},
    \end{align*}
    where $\Abs{\widetilde{V}-V}\leq \varepsilon$.
    The final result immediately follows by the linearity of error propagation in the product of operators.
\end{proof}

The following lemma provides a way to construct a block-encoding of a purification $\ket{\rho}\bra{\rho}$ of $\rho$, given a block-encoding of $\rho$. 

\begin{lemma}
\label{lemma:block-encoding-of-purification}
    Let $\rho$ be an $n$-qubit density operator with $\rho\geq I/\kappa$ for some $\kappa \geq 1$, and let $N = 2^n$.
    Given a unitary $U_\rho$ that is a $(\beta,a,0)$-block-encoding of $\rho$ with $\beta\geq 1$, for any $\varepsilon\in \rbra{0,1/(4\beta N)}$, we can construct a quantum circuit $\widetilde U$ that is a $(2,2n+2a+5,\varepsilon)$-block-encoding of $\ket{\rho}\bra{\rho}$,
    using $O(\beta^2\kappa N \log(\beta N/\varepsilon)\log(1/\varepsilon))$ queries to $U_\rho$ and $O(a\beta^2\kappa N \log(\beta N/\varepsilon)\log(1/\varepsilon))$ gates,
    where $\ket{\rho}$ is a purification of $\rho$.
\end{lemma}

\begin{proof}
    Let $\delta = \min\cbra{1/(\beta\kappa), 1/4}$ and 
    $\varepsilon'$ be determined later.
    By \cref{lemma:positive-power}, there is a polynomial $p\rbra{x}$ of degree $d = O\rbra*{\frac{1}{\delta}\log\rbra*{\frac{1}{\varepsilon'}}}$ such that
    \begin{align*}
        & \forall x \in \sbra*{-1, 1}, \quad \abs*{p\rbra*{x}} \leq 1, \\
        & \forall x \in \sbra*{\delta, 1}, \quad \abs*{p\rbra*{x} - \frac 1 2 x^{1/2}} \leq \varepsilon'.
    \end{align*}
    Let $q\rbra{x} = p\rbra{x}/2$. 
    By \cref{thm:qsvt}, we can implement a unitary operator $U_{\rho^{1/2}}$ that is a $\rbra{1, a+2, \varepsilon'/2}$-block-encoding of $q\rbra{\rho/\beta}$, using $O\rbra{d}$ queries to $U_{\rho}$ and $O\rbra{\rbra{a+1}d}$ one- and two-qubit quantum states, and the description of $U_{\rho^{1/2}}$ can be computed in classical time $\poly\rbra{d, \log\rbra{1/\varepsilon'}}$.
    Note that $U_{\rho^{1/2}}$ is a $\rbra{1, a+2, \varepsilon'}$-block-encoding of $\rho^{1/2}/(4\beta^{1/2})$. 

    Let $G = \text{CNOT}^{\otimes n} \rbra*{H^{\otimes n}\otimes I_n}$ be a unitary operator that prepares the maximal entangled state by 
    \[
    G\ket{0}_{2n}=\frac{1}{\sqrt{N}}\sum_{j=0}^{N-1} \ket{j}_n\ket{j}_n.
    \]
    Let $U = (U_{\rho^{1/2}}\otimes I_n) (G\otimes I_{a+2})$.
    It can be shown that $U$ is a $\rbra{4\sqrt{\beta N}, a+2,
    4\sqrt{\beta N}\varepsilon'}$-block-encoding of $V =\sqrt{N} \rbra{\rho^{1/2} \otimes I_n} G$, where $V\ket{0}_{2n} = \ket{\rho}$ is a purification of $\rho$. 
    By \cref{lmm:blk den} (using $\ket{\rho}$ is a purification of $\ket{\rho} \bra{\rho}$),
    we can construct a quantum circuit $U_{\ket{\rho}}$ that is a $(16\beta N,2n+2a+4,
    8\sqrt{\beta N}\varepsilon')$-block-encoding of $\ket{\rho}\bra{\rho}$,
    using $O(1)$ queries to $U$ and additional $O(n)$ one- and two-qubit quantum gates.
    Finally, by \cref{corollary:up-scaling-block-encoding} (with $\alpha \coloneqq 16\beta N$, $\beta \coloneqq 2$, $\varepsilon \coloneqq 8\sqrt{\beta N}\varepsilon'$ and $\varepsilon'\coloneqq \varepsilon/2$),
    we can construct a quantum circuit $\widetilde U$ that is a $(2,2n+2a+5, \Theta\rbra{\beta^{3/4} N^{3/4} 
    \sqrt{\varepsilon'} \log\rbra{\frac{1}{\varepsilon}}} + \frac{\varepsilon}{2})$-block-encoding of $\ket{\rho}\bra{\rho}$, using $O\rbra{\beta N\log(1/\varepsilon)}$ queries to $U_{\ket{\rho}}$.
    By taking $\varepsilon' = \Theta(\beta^{-3/2}N^{-3/2}\varepsilon^4)$ to be sufficiently small, $\widetilde U$ is a $\rbra{2, 2n+2a+5, \varepsilon}$-block-encoding of $\ket{\rho}\bra{\rho}$.
    The total number of queries to $U_\rho$ in the construction of $\widetilde U$ is $O\rbra{\beta Nd\log(1/\varepsilon)} = O\rbra{\beta^2\kappa N \log\rbra{\beta N/\varepsilon}\log(1/\varepsilon)}$.
    The final result immediately follows.
\end{proof}

\subsection{Density matrix exponentiation}

Given only samples of a quantum state $\rho$, we can approximately implement the unitary operator $e^{-i\rho t}$ \cite{LMR14,KLL+17}, also known as the sample-based Hamiltonian simulation. 
We use the version in \cite{GKP+24}.

\begin{theorem} [Density matrix exponentiation, {\cite[Corollary 3]{GKP+24}}]
\label{thm:sample-based-hamiltonian-simulation}
    For every $0 < \delta < 1$ and $t \geq 1$, and given access to independent samples of a quantum state $\rho$, we can implement a quantum channel $\mathcal{E}$ such that
    \[
    \Abs*{\mathcal{E} - e^{-i\rho t}}_{\diamond} \leq \delta,
    \]
    using $O\rbra{t^2/\delta}$ samples of $\rho$.
\end{theorem}

Based on the technique of density matrix exponentiation given in \cref{thm:sample-based-hamiltonian-simulation}, we can implement block-encodings of quantum states using their samples by quantum singular value transformation \cite{GSLW19}, as noted in \cite{GP22}.

\begin{lemma} [{\cite[Corollary 21]{GP22}}]
\label{lemma:sample-to-block-encoding}
    For every $\delta \in \rbra{0, 1}$ and given access to independent samples of a quantum state $\rho$, we can implement a quantum channel $\mathcal{E}$, using $O\rbra*{\frac{1}{\delta} \log^2\rbra*{\frac{1}{\delta}}}$ samples of $\rho$ such that $\Abs{\mathcal{E}-\mathcal{U}}_{\diamond} \leq \delta$, where $\mathcal{U} \colon \varrho \mapsto U \varrho U^\dag$ and $U$ is a $\rbra{4/\pi, 3, 0}$-block-encoding of $\rho$. 
\end{lemma}

In quantum algorithms, we usually need $U^\dag$, the inverse of $U$. As noted in \cite{WZ23}, we can also implement a quantum channel $\mathcal{E}^{\operatorname{inv}}$ that is $\delta$-close to $U^\dag$ in the diamond norm distance with the same complexity stated in \cref{lemma:sample-to-block-encoding}. 
Therefore, we can roughly regard the quantum channel $\mathcal{E}$ as if it were a unitary operator.

For our purpose, we show how to prepare block-encodings of quantum states without normalizing factors.

\begin{lemma} [Block-encodings of quantum states]
\label{lemma:block-encoding-of-quantum-states}
    For every $\delta \in \rbra{0, 1}$ and given access to independent samples of a quantum state $\rho$, we can implement a quantum channel $\mathcal{E}$, using $O\rbra*{\frac{1}{\delta}\log^2\rbra*{\frac{1}{\delta}}}$ samples of $\rho$ such that $\Abs{\mathcal{E} - \mathcal{U}}_{\diamond} \leq \delta$, where $\mathcal{U} \colon \varrho \mapsto U\varrho U^\dag$ and $U$ is a 
    $\rbra{2, 4, 0}$-block-encoding
    of $\rho$. 
    Moreover, we can also implement a quantum channel $\mathcal{E}^{\textup{inv}}$ with the same sample complexity such that $\Abs{\mathcal{E}^{\textup{inv}} - \mathcal{U}^{\textup{inv}}}_\diamond \leq \delta$, where $\mathcal{U}^{\textup{inv}} \colon \varrho \mapsto U^\dag \varrho U$.
\end{lemma}

\begin{proof}
    By \cref{lemma:sample-to-block-encoding}, we can implement a quantum channel $\mathcal{E}$ using $s = O\rbra*{\frac{1}{\delta}\log^2\rbra*{\frac 1 \delta}}$
    samples of $\rho$ such that 
    $\Abs{\mathcal{E} - \mathcal{U}_{\pi\rho/4}}_{\diamond} \leq \delta$, where $\mathcal{U}_{\pi\rho/4} \colon \varrho \mapsto U_{\pi\rho/4} \varrho U_{\pi\rho/4}^\dag$ and $U_{\pi\rho/4}$ is a $\rbra{4/\pi, 3, 0}$-block-encoding of $\rho$, i.e., a $\rbra{1, 3, 0}$-block-encoding of $\pi\rho/4$.
    We can also implement $\mathcal{E}^{\text{inv}}$ using $s$ samples, such that 
    $\Abs{\mathcal{E}^{\text{inv}} - \mathcal{U}_{\pi\rho/4}^\dag}_{\diamond} \leq \delta$, where $\mathcal{U}_{\pi\rho/4}^\dag$ denotes the quantum channel induced by $U_{\pi\rho/4}^\dag$.
    In addition, we can implement the controlled versions of $\mathcal{E}$ and $\mathcal{E}^{\text{inv}}$ with the same sample complexity by replacing each gate in their implementations with its controlled version. 

    By \cref{corollary:scaling-block-encoding}, we can construct a unitary opeator $U_{\rho/2}$ that is a $\rbra{1, 4, 0}$-block-encoding of $\rho/2$, i.e., a $\rbra{2, 4, 0}$-block-encoding of $\rho$, using $1$ query to $U_{\pi\rho/4}$.
    If we replace all the uses of (controlled-)$U_{\pi\rho/4}$ and (controlled-)$U_{\pi\rho/4}^\dag$ by (controlled-)$\mathcal{E}$ and (controlled-)$\mathcal{E}^{\text{inv}}$ in the implementation of $U_\rho$, then we will obtain an implementation of a quantum channel $\mathcal{E}_{\rho}$ such that
    $\Abs{\mathcal{E}_{\rho} - \mathcal{U}_\rho}_{\diamond} \leq O\rbra{\delta}$, where $\mathcal{U}_\rho \colon \varrho \mapsto U_\rho \varrho U_\rho^\dag$.
    Finally, rescaling $\delta$ by a constant factor will achieve our goal.

    To implement the inverse $\mathcal{E}^{\textup{inv}}$, we can directly apply the trick noted in \cite{WZ23}.
\end{proof}

\section{Quantum Sample-to-Query Lifting Theorem} \label{sec:lifting}

In this section, we prove our quantum sample-to-query lifting theorem, and then show its tightness. 

\subsection{Main theorem}

\begin{theorem} [Lifting]
\label{thm:main}
    For every promise problem $\mathcal{P}$ for quantum states, we have
    \[
    \mathsf{Q}_\diamond\rbra*{\mathcal{P}} = \Omega\rbra*{\frac{\sqrt{\mathsf{S}\rbra{\mathcal{P}}}}{\log\rbra{\mathsf{S}\rbra{\mathcal{P}}}}}.
    \]
\end{theorem}

For convenience, we first define the unitary property testing problem corresponding to a given quantum state testing problem, where the tested quantum state is given as block-encoded in the tested unitary operator. 

\begin{definition} \label{def:p-diamond}
    Suppose that $\mathcal{P} = \rbra{\mathcal{P}^{\textup{yes}}, \mathcal{P}^{\textup{no}}}$ is a promise problem for quantum state testing. 
    We define $\mathcal{P}^\diamond$ to be the unitary property testing problem corresponding to $\mathcal{P}$, which is a promise problem with
    \begin{itemize}
    \item {Yes} instance: a block-encoding $U$ of $\frac{1}{2}\rho$ such that $\rho \in \mathcal{P}^{\textup{yes}}$;
    \item {No} instance: a block-encoding $U$ of $\frac{1}{2}\rho$ such that $\rho \in \mathcal{P}^{\textup{no}}$.
\end{itemize}
\end{definition}

\begin{proof}[Proof of \cref{thm:main}]
    The proof is a reduction from $\mathcal{P}$ to $\mathcal{P}^\diamond$. 

    \paragraph{Reduction}
    Let $Q = \mathsf{Q}_\diamond\rbra{\mathcal{P}}$.
    That is, there is a quantum query algorithm $\mathcal{A}$ for the promise problem $\bb{\mathcal{P}}$ using queries to the tested unitary operator $U$ such that
    \[
    \mathcal{A}^U = G_Q U_Q \dots G_2 U_2 G_1 U_1 G_0,
    \]
    where each $U_j$ is either (controlled-)$U$ or (controlled-)$U^\dag$, and $G_j$ are quantum gates that do not depend on $U$. 
    By the definition of $\mathcal{A}$, for every block-encoding $U_\rho$ of $\rho$, 
    \begin{enumerate}
        \item $\Pr\sbra{\mathcal{A} \textup{ accepts } U_\rho} \geq 2/3$ if $\rho \in \mathcal{P}^{\textup{yes}}$,
        \item $\Pr\sbra{\mathcal{A} \textup{ accepts } U_\rho} \leq 1/3$ if $\rho \in \mathcal{P}^{\textup{no}}$.
    \end{enumerate}
    Here, $\mathcal{A}$ is said to accept $U_\rho$ if $\mathcal{A}^{U_{\rho}}$ outputs $0$, i.e., $\Pr\sbra{\mathcal{A} \textup{ accepts } U_\rho} = \Abs{\Pi \mathcal{A}^{U_\rho} \ket{0}}^2$ and $\Pi = \ket{0}\bra{0} \otimes I$ is the projector onto the first qubit.

    Then, we reduce $\mathcal{P}$ to $\bb{\mathcal{P}}$ by constructing a tester $\mathcal{T}$ for $\mathcal{P}$ that relates to the quantum query algorithm $\mathcal{A}$ as follows.

    \vspace{5pt}
    \noindent\fbox{
    \parbox{\textwidth-15.6pt}{
    \vspace{7pt}
    $\mathcal{T}$ --- \textbf{Tester for} $\mathcal{P}$ \textbf{based on the tester} $\mathcal{A}$ \textbf{for} $\bb{\mathcal{P}}$

    \begin{itemize}
        \item[] \textbf{Input}: Quantum state $\rho$.
        \item[] \textbf{Output}: A classical bit $b \in \cbra{0, 1}$. 
    \end{itemize}

    \begin{enumerate}
        \item Let $\delta = 1/9Q$.
        \item Prepare $\sigma = G_0 \ket{0} \bra{0} G_0^\dag$. 
        \item For $i = 1, 2, \dots, Q$ in this order,
        \begin{enumerate}
            \item[3.1.] Apply $\mathcal{E}_i$ to $\sigma$, where $\mathcal{E}_i$ is (controlled-)$\mathcal{E}_\rho$ (resp.\ $\mathcal{E}_{\rho}^{\textup{inv}}$) if $U_i$ is (controlled-)$U$ (resp.\ $U^\dag$), where $\mathcal{E}_\rho$ is the quantum channel defined in \cref{lemma:block-encoding-of-quantum-states} that is $\delta$-close to a block-encoding of $\frac{1}{2}\rho$ in the diamond norm distance.
            \item[3.2.] Apply $G_i$ to $\sigma$.
        \end{enumerate}
        \item Let $b$ be the measurement outcome of the first qubit of $\sigma$ in the computational basis. 
        \item Return $b$. 
    \end{enumerate}
    \vspace{7pt}
    }
    }
    \vspace{5pt}

    \paragraph{Complexity}
    The implementation of $\mathcal{T}$ is visualized in \cref{fig:samplized-block-encoded-sample-access}, where
    \begin{equation}
    \label{eq:s}
        s = O\rbra*{\frac{1}{\delta}\log^2\rbra*{\frac{1}{\delta}}} = O\rbra*{Q \log^2 \rbra{Q}}
    \end{equation}
    is the number of independent samples of $\rho$ used in the implementation of $\mathcal{E}_\rho$ in \cref{lemma:block-encoding-of-quantum-states}.
    It is clear that $\mathcal{T}$ uses $Qs$ independent samples of $\rho$. 
    On the other hand, if $\mathcal{T}$ solves $\mathcal{P}$, then the quantum sample complexity of $\mathcal{T}$ is lower bounded by that of $\mathcal{P}$, i.e., $Qs \geq \mathsf{S}\rbra{\mathcal{P}}$, which together with \cref{eq:s} gives
    \[
    Q = \Omega\rbra*{\frac{\sqrt{\mathsf{S}\rbra{\mathcal{P}}}}{\log\rbra{\mathsf{S}\rbra{\mathcal{P}}}}}.
    \]
\begin{figure} [!htp]
\centering
\begin{quantikz}
    & \wireoverride{n} & \lstick{$\rho^{\otimes s}$} \wireoverride{n} & \gate[2]{\mathcal{E}_1} & \wireoverride{n} & \lstick{$\rho^{\otimes s}$} \wireoverride{n} & \gate[2]{\mathcal{E}_2} & \wireoverride{n} & \wireoverride{n} & \lstick{$\rho^{\otimes s}$} \wireoverride{n} & \gate[2]{\mathcal{E}_Q} & \wireoverride{n} \\
    \lstick{$\ket 0$} & \gate{G_0}\qwbundle{} & \qw &      & \gate{G_1} & \qw &  & \gate{G_2} & \qw \midstick[2,brackets=none]{\(\cdots\)} & \qw &  & \gate{G_Q} & \qw \rstick{$\sigma$} \\
\end{quantikz}
\caption{Tester for $\mathcal{P}$ based on $\mathcal{A}$.}
\label{fig:samplized-block-encoded-sample-access}
\end{figure}
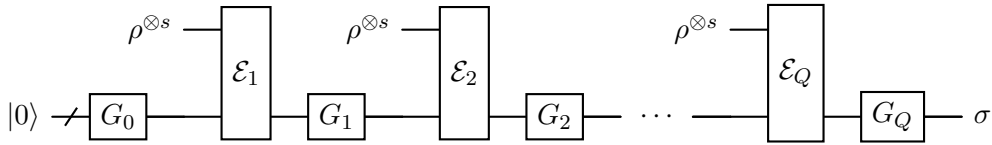

    \paragraph{Correctness}
    To complete the proof, it remains to show that $\mathcal{T}$ is a tester for $\mathcal{P}$. 
    Let $\rho$ be a quantum state.
    We only have to show that there is a constant $\varepsilon > 0$ such that
    \begin{enumerate}
        \item $\Pr\sbra{\mathcal{T}^\rho \textup{ outputs } 0} \geq 1/2 + \varepsilon$ if $\rho \in \mathcal{P}^{\textup{yes}}$,
        \item $\Pr\sbra{\mathcal{T}^\rho \textup{ outputs } 0} \leq 1/2 - \varepsilon$ if $\rho \in \mathcal{P}^{\textup{no}}$.
    \end{enumerate}
    One can achieve success probability $2/3$ by majority voting from a constant number of repetitions of $\mathcal{T}$.
    
    Let $U_\rho$ be the 
    $\rbra{2, 4, 0}$-block-encoding of $\rho$ defined in \cref{lemma:block-encoding-of-quantum-states}. 
    Note that in the implementation of $\mathcal{T}$, it holds that $\Abs{\mathcal{E}_i - \mathcal{U}_i}_\diamond \leq \delta$ for every $1 \leq i \leq Q$, where $\mathcal{U}_i \colon \varrho \mapsto U_i \varrho U_i^\dag$. 
    Then, we have
    \[
    \Abs*{\mathcal{T}^{\rho} - \mathcal{A}^{U_\rho}}_{\diamond} \leq Q\delta = \frac 1 9.
    \]
    Note that $\Pr\sbra{\mathcal{T}^{\rho} \textup{ outputs } 0} = \tr\rbra{\Pi \mathcal{T}^{\rho}\rbra{\ket{0}\bra{0}}}$ and $\Pr\sbra{\mathcal{A}^{U_\rho} \textup{ outputs } 0} = \Abs{\Pi \mathcal{A}^{U_\rho} \ket{0}}^2$. Then, we have
    \begin{align*}
    \abs*{ \Pr\sbra*{\mathcal{T}^{\rho} \textup{ outputs } 0} - \Pr\sbra*{\mathcal{A}^{U_{\rho}} \textup{ outputs } 0} }
    & = \abs*{ \tr\rbra*{ \Pi \mathcal{T}^{\rho} \rbra{\ket{0}\bra{0}} } - \tr\rbra*{ \Pi \mathcal{A}^{U_\rho} \ket{0}\bra{0} \rbra{\mathcal{A}^{U_\rho}}^\dag } } \\
    & \leq \Abs*{\mathcal{T}^{\rho} - \mathcal{A}^{U_\rho}}_\diamond \leq \frac 1 9.
    \end{align*}
    If $\rho \in \mathcal{P}^{\textup{yes}}$, then
    \[
    \Pr\sbra*{\mathcal{T}^{\rho} \textup{ outputs } 0} \geq \Pr\sbra*{\mathcal{A}^{U_{\rho}} \textup{ outputs } 0} - \frac 1 9 \geq \frac 2 3 - \frac 1 9 = \frac 5 9.
    \]
    If $\rho \in \mathcal{P}^{\textup{no}}$, then
    \[
    \Pr\sbra*{\mathcal{T}^{\rho} \textup{ outputs } 0} \leq \Pr\sbra*{\mathcal{A}^{U_{\rho}} \textup{ outputs } 0} + \frac 1 9 \leq \frac 1 3 + \frac 1 9 = \frac 4 9.
    \]
    Therefore, we conclude that $\mathcal{T}$ is a tester for $\mathcal{P}$.
\end{proof}

\subsection{Tightness}

\label{sec:tightness}

We show that \cref{thm:main} is tight by considering the promise problem $\textsc{Dis}_{\rho_+, \rho_-}$, where
\[
    \rho_{\pm} = \rbra*{\frac 1 2 \mp 8\varepsilon} \ket{0}\bra{0} + \rbra*{\frac 1 2 \pm 8\varepsilon} \ket{1}\bra{1}.
\]

\begin{theorem} [Lifting is tight]
\label{thm:tightness}
    $\mathsf{Q}_\diamond\rbra{\textsc{Dis}_{\rho_+, \rho_-}} = \widetilde \Theta\rbra{\sqrt{\mathsf{S}\rbra{\textsc{Dis}_{\rho_+, \rho_-}}}} = \widetilde \Theta\rbra{1/\varepsilon}$ for $0 < \varepsilon \leq 1/16$.
\end{theorem}

\begin{proof}
    We reduce $\bb{\textsc{Dis}_{\rho_+, \rho_-}}$ to quantum amplitude estimation.
    
    \paragraph{Reduction}
    Let $\mathcal{A}$ be the quantum query algorithm for amplitude estimation given in \cref{thm:ampest}.
    We construct a tester $\mathcal{T}$ for $\bb{\textsc{Dis}_{\rho_+, \rho_-}}$ using $\mathcal{A}$ as follows.

    \vspace{5pt}
    \noindent\fbox{
    \parbox{\textwidth-15.6pt}{
    \vspace{7pt}
    $\mathcal{T}$ --- \textbf{Tester for} $\bb{\textsc{Dis}_{\rho_+, \rho_-}}$ \textbf{using quantum amplitude estimator} $\mathcal{A}$

    \begin{itemize}
        \item[] \textbf{Input}: Quantum oracle $U$.
        \item[] \textbf{Output}: A classical bit $b \in \cbra{0, 1}$. 
    \end{itemize}

    \begin{enumerate}
        \item Let $\tilde p$ be an $\varepsilon$-estimate of $p$ returned by $\mathcal{A}^U$, where $U\ket{0} = \sqrt{p} \ket{0} + \sqrt{1-p} \ket{\perp}$.
        \item Return $0$ if $\tilde p < 1/16 + 16\varepsilon^2$, and $1$ otherwise.
    \end{enumerate}
    \vspace{7pt}
    }
    }
    \vspace{5pt}

    \paragraph{Complexity}
    It can be seen that the quantum query complexity of $\mathcal{T}$ is exactly that of $\mathcal{A}$, which is $O\rbra{1/\varepsilon}$. 
    Therefore, if $\mathcal{T}$ solves $\bb{\textsc{Dis}_{\rho_+, \rho_-}}$, then $\mathsf{Q}_\diamond\rbra{\textsc{Dis}_{\rho_+, \rho_-}} = O\rbra{1/\varepsilon}$.
    Note that $\gamma = 1 - \operatorname{F}\rbra{\rho_+, \rho_-} = 1 - \sqrt{1 - 256\varepsilon^2} \leq 256 {\varepsilon^2}$.
    By \cref{corollary:qsd-intro}, 
    \[
    \mathsf{Q}_\diamond\rbra{\textsc{Dis}_{\rho_+, \rho_-}} = \widetilde\Omega\rbra*{\sqrt{\mathsf{S}\rbra{\textsc{Dis}_{\rho_+, \rho_-}}}} = \widetilde \Omega\rbra{1/\sqrt{\gamma}} = \widetilde\Omega\rbra{1/\varepsilon}.
    \]
    The above together yield that $\mathsf{Q}_\diamond\rbra{\textsc{Dis}_{\rho_+, \rho_-}} = \widetilde \Theta\rbra{1/\varepsilon}$.

    \paragraph{Correctness} To complete the proof, it remains to show that $\mathcal{T}$ is a tester for $\bb{\textsc{Dis}_{\rho_+, \rho_-}}$. 
    To begin with, let $U_{\rho_+}$ and $U_{\rho_-}$ be block-encodings of $\rho_+$ and $\rho_-$, respectively. 
    We only have to show that 
    \begin{enumerate}
        \item $\Pr \sbra{ \mathcal{T}^{U_{\rho_+}} \textup{ outputs } 0 } \geq 2/3$,
        \item $\Pr \sbra{ \mathcal{T}^{U_{\rho_-}} \textup{ outputs } 0 } \leq 1/3$.
    \end{enumerate}
    Suppose that $U_{\rho_\pm}$ is a $\rbra{1, a, 0}$-block-encoding of $\frac{1}{2}\rho_\pm$. 
    Then, it can be seen that
    \[
    U_{\rho_\pm} \ket{0} \ket{0}^{\otimes a} = \rbra*{\frac 1 4 \mp 4\varepsilon} \ket{0} \ket{0}^{\otimes a} + \sqrt{1 - \rbra*{\frac 1 4 \mp 4\varepsilon}^2} \ket{\perp},
    \]
    where $\ket{\perp}$ is orthogonal to $\ket{0}\ket{0}^{\otimes a}$. 
    Let $\tilde p_{\pm}$ be the estimate returned by $\mathcal{A}^{U_{\rho_{\pm}}}$. 
    Then, 
    \[
    \Pr\sbra*{ \abs*{\tilde p_{\pm} - \rbra*{\frac 1 4 \mp 4\varepsilon}^2} \leq \varepsilon } \geq \frac 2 3,
    \]
    which implies that
    \[
    \Pr \sbra*{ \tilde p_+ \leq \frac{1}{16} - \varepsilon + 16\varepsilon^2 } \geq \frac 2 3, \qquad \Pr \sbra*{ \tilde p_- \geq \frac{1}{16} + \varepsilon + 16\varepsilon^2 } \geq \frac 2 3.
    \]
    Finally, we have
    \begin{align*}
    \Pr\sbra*{\mathcal{T}^{U_{\rho_+}} \textup{ outputs } 0} & = \Pr\sbra*{\tilde p_+ < \frac 1 {16} + 16\varepsilon^2} \geq \Pr \sbra*{ \tilde p_+ \leq \frac{1}{16} - \varepsilon + 16\varepsilon^2 } \geq \frac 2 3, \\
    \Pr\sbra*{\mathcal{T}^{U_{\rho_-}} \textup{ outputs } 0} & = \Pr\sbra*{\tilde p_- < \frac 1 {16} + 16\varepsilon^2} \leq 1 - \Pr \sbra*{ \tilde p_- \geq \frac{1}{16} + \varepsilon + 16\varepsilon^2 } \leq \frac 1 3.
    \end{align*}
    Therefore, we conclude that $\mathcal{T}$ is a tester for $\bb{\textsc{Dis}_{\rho_+, \rho_-}}$.
\end{proof}

The tester used in the proof of \cref{thm:tightness} is a reduction from block quantum state discrimination to quantum amplitude estimation.
If we apply the tester from another direction, we will obtain an information-theoretic quantum query lower bound for amplitude estimation (see \cref{sec:ampl-est}). 

\section{Quantum Query Lower Bounds}
\label{sec:query-lower-bound}

In this section, we give unified proofs for quantum query lower bounds related to a series of problems. 

\subsection{Quantum Gibbs sampling}

A quantum query algorithm $\mathcal{A}$ is said to be a quantum Gibbs sampler with precision $\varepsilon$ at inverse temperature $\beta$, if for every block-encoding $U$ of a Hamiltonian $H$, $\mathcal{A}^U$ prepares a mixed quantum state that is $\varepsilon$-close to the Gibbs state $e^{-\beta H}/\tr\rbra{e^{-\beta H}}$ in trace distance.
It was shown in \cite[Theorem 35]{GSLW19} that there is a quantum Gibbs sampler with constant precision at inverse temperature $\beta$ with quantum query complexity $O\rbra{\beta}$.

\begin{theorem} [Optimality of quantum Gibbs sampling]
\label{thm:gibbs}
    For $\beta \geq 4$ and $0 < \varepsilon \leq 1/5$, any quantum Gibbs sampler with precision $\varepsilon$ at inverse temperature $\beta$ has quantum query complexity $\widetilde\Omega\rbra{\beta}$. 
\end{theorem}

\begin{proof}
    We reduce $\bb{\textsc{Dis}_{\rho_+, \rho_-}}$ to quantum Gibbs sampling, where 
    \[
    \rho_{\pm} = \rbra*{\frac 1 2 \mp \frac{2}{\beta}} \ket{0}\bra{0} + \rbra*{\frac 1 2 \pm \frac{2}{\beta}} \ket{1}\bra{1}.
    \]
    
    \paragraph{Reduction}
    Suppose that there is a quantum Gibbs sampler $\mathcal{A}$ with precision $\varepsilon$ at inverse temperature $\beta$, which has quantum query complexity $Q$, such that for every block-encoding $U$ of an $n$-qubit Hamiltonian $H$,
    \[
    \Abs*{ \tr_{\textup{env}}\rbra*{ \mathcal{A}^{U} \rbra*{ \ket{0} \bra{0}^{\otimes n} \otimes \underbrace{\ket{0} \bra{0}^{\otimes n_{\textup{env}}}}_{\textup{env}} } \rbra*{\mathcal{A}^{U}}^\dag } - \frac{e^{-\beta H}}{\tr\rbra{e^{-\beta H}}} }_1 \leq \varepsilon.
    \]
    We reduce $\bb{\textsc{Dis}_{\rho_+, \rho_-}}$ to quantum Gibbs sampling by constructing a tester $\mathcal{T}$ for $\bb{\textsc{Dis}_{\rho_+, \rho_-}}$ using the quantum Gibbs sampler $\mathcal{A}$ as follows.
    Note that $\rho_+$ and $\rho_-$ are one-qubit mixed quantum states.
    
    \vspace{5pt}
    \noindent\fbox{
    \parbox{\textwidth-15.6pt}{
    \vspace{7pt}
    $\mathcal{T}$ --- \textbf{Tester for} $\bb{\textsc{Dis}_{\rho_+, \rho_-}}$ \textbf{using quantum Gibbs sampler} $\mathcal{A}$

    \begin{itemize}
        \item[] \textbf{Input}: Quantum oracle $U$.
        \item[] \textbf{Output}: A classical bit $b \in \cbra{0, 1}$. 
    \end{itemize}

    \begin{enumerate}
        \item Prepare $\ket{\psi} = \mathcal{A}^U \ket{0} \ket{0}^{\otimes n_{\textup{env}}}$. 
        \item Let $b$ be the measurement outcome of the first qubit of $\ket{\psi}$ in the computational basis. 
        \item Return $b$. 
    \end{enumerate}
    \vspace{7pt}
    }
    }
    \vspace{5pt}

    \paragraph{Complexity}
    Clearly, it can be seen that the quantum query complexity of $\mathcal{T}$ is exactly that of $\mathcal{A}$, which is $Q$. 
    On the other hand, if $\mathcal{T}$ solves $\bb{\textsc{Dis}_{\rho_+, \rho_-}}$, then the quantum query complexity of $\mathcal{T}$ is lower bounded by that of $\bb{\textsc{Dis}_{\rho_+, \rho_-}}$, i.e., $Q \geq \mathsf{Q}_\diamond\rbra{\textsc{Dis}_{\rho_+, \rho_-}}$.
    Note that 
    \[
    \gamma = 1 - \operatorname{F}\rbra{\rho_+, \rho_-} = 1 - \sqrt{1 - \frac{16}{\beta^2}} \leq \frac {16} {\beta^2}.
    \]
    By \cref{corollary:qsd-intro}, we have the quantum query lower bound $Q \geq \mathsf{Q}_\diamond\rbra{\textsc{Dis}_{\rho_+, \rho_-}} = \widetilde\Omega\rbra{1/\sqrt{\gamma}} = \widetilde\Omega\rbra{\beta}$ for quantum Gibbs sampling.

    \paragraph{Correctness}
    To complete the proof, it remains to show that $\mathcal{T}$ is a tester for $\bb{\textsc{Dis}_{\rho_+, \rho_-}}$. 
    To begin with, let $U_{\rho_+}$ and $U_{\rho_-}$ be block-encodings of $\frac{1}{2}\rho_+$ and $\frac{1}{2}\rho_-$, respectively. 
    We only have to show that 
    \begin{enumerate}
        \item $\Pr \sbra{ \mathcal{T}^{U_{\rho_+}} \textup{ outputs } 0 } \geq 2/3$,
        \item $\Pr \sbra{ \mathcal{T}^{U_{\rho_-}} \textup{ outputs } 0 } \leq 1/3$.
    \end{enumerate}
    Let $\Pi = \ket{0}\bra{0} \otimes I_{\textup{env}}$ be a projector onto the first qubit. Then, 
    \begin{align}
    \Pr\sbra*{\mathcal{T}^{U_{\rho_\pm}} \textup{ outputs } 0} 
    & = \Abs*{\Pi \mathcal{A}^{U_{\rho_\pm}} \ket{0} \ket{0}^{\otimes n_{\textup{env}}}}^2 \nonumber \\
    & = \tr\rbra*{ \ket{0}\bra{0} \cdot \tr_{\textup{env}} \rbra*{ \mathcal{A}^{U_{\rho_\pm}} \rbra*{ \ket{0} \bra{0} \otimes \underbrace{\ket{0} \bra{0}^{\otimes n_{\textup{env}}}}_{\textup{env}} } \rbra*{\mathcal{A}^{U_{\rho_\pm}}}^\dag } }. \label{eq:gibbs-prob}
    \end{align}
    By the definition of quantum Gibbs sampler, we have
    \begin{equation} \label{eq:gibbs-tr}
    \Abs*{ \tr_{\textup{env}} \rbra*{ \mathcal{A}^{U_{\rho_\pm}} \rbra*{ \ket{0} \bra{0} \otimes \underbrace{\ket{0} \bra{0}^{\otimes n_{\textup{env}}}}_{\textup{env}} } \rbra*{\mathcal{A}^{U_{\rho_\pm}}}^\dag } - \frac{e^{-\frac 1 2\beta \rho_\pm}}{\tr\rbra*{e^{-\frac 1 2\beta \rho_\pm}}} }_1 \leq \varepsilon.
    \end{equation}
    By \cref{eq:gibbs-prob} and \cref{eq:gibbs-tr}, we have
    \begin{equation} \label{eq:gibbs-prob-eps}
    \abs*{ \Pr\sbra*{\mathcal{T}^{U_{\rho_\pm}} \textup{ outputs } 0} - \tr\rbra*{\ket{0}\bra{0} \cdot \frac{e^{-\frac 1 2\beta \rho_\pm}}{\tr\rbra*{e^{-\frac 1 2\beta \rho_\pm}}}} } \leq \varepsilon.
    \end{equation}
    It can be verified that
    \[
    \frac{e^{-\frac 1 2\beta \rho_\pm}}{\tr\rbra*{e^{-\frac 1 2\beta \rho_\pm}}} = \frac{1}{e+e^{-1}} \rbra*{ e^{\pm 1} \ket{0}\bra{0} + e^{\mp 1} \ket{1}\bra{1} },
    \]
    and thus
    \[
    \tr\rbra*{\ket{0}\bra{0} \cdot \frac{e^{-\frac 1 2\beta \rho_\pm}}{\tr\rbra*{e^{-\frac 1 2\beta \rho_\pm}}}} = \frac{e^{\pm 1}}{e+e^{-1}}.
    \]
    By \cref{eq:gibbs-prob-eps}, we can consider $\rho_+$ and $\rho_-$ separately as follows:
    \begin{align*}
    \Pr\sbra*{\mathcal{T}^{U_{\rho_+}} \textup{ outputs } 0} & \geq \tr\rbra*{\ket{0}\bra{0} \cdot \frac{e^{-\frac 1 2\beta \rho_+}}{\tr\rbra*{e^{-\frac 1 2\beta \rho_+}}}} - \varepsilon = \frac{e}{e+e^{-1}} - \varepsilon \geq \frac 2 3, \\
    \Pr\sbra*{\mathcal{T}^{U_{\rho_-}} \textup{ outputs } 0} & \leq \tr\rbra*{\ket{0}\bra{0} \cdot \frac{e^{-\frac 1 2\beta \rho_-}}{\tr\rbra*{e^{-\frac 1 2\beta \rho_-}}}} + \varepsilon = \frac{e^{-1}}{e+e^{-1}} + \varepsilon \leq \frac 1 3.
    \end{align*}
    We conclude that $\mathcal{T}$ is a tester for $\bb{\textsc{Dis}_{\rho_+, \rho_-}}$.
\end{proof}

\subsubsection*{Quantum Gibbs sampler given the square root of the Hamiltonian}

A quantum query algorithm $\mathcal{A}$ is said to be a quantum Gibbs sampler given the square root of the Hamiltonian with precision $\varepsilon$ at inverse temperature $\beta$, if for every block-encoding $U$ of $\sqrt{H}$ where $H$ is a Hamiltonian, $\mathcal{A}^U$ prepares a mixed quantum state that is $\varepsilon$-close to the Gibbs state $e^{-\beta H}/\tr\rbra{e^{-\beta H}}$ in trace distance.

\begin{theorem} [Impossibility of improving quantum Gibbs sampling even with stronger oracles]
\label{thm:sqrt-beta-gibbs}
    For $\beta \geq 4$ and $0 < \varepsilon \leq 1/16$, any quantum Gibbs sampler given the square root of the Hamiltonian with precision $\varepsilon$ at inverse temperature $\beta$ has quantum query complexity $\widetilde\Omega\rbra{\beta}$. 
\end{theorem}

\begin{proof}
    Suppose that there is a quantum Gibbs sampler $\mathcal{A}$ with precision $\varepsilon$ at inverse temperature $\beta$, which has quantum query complexity $Q$, such that for every block-encoding $U$ of $\sqrt{H}$ where $H$ is an $n$-qubit Hamiltonian,
    \[
    \Abs*{ \tr_{\textup{env}}\rbra*{ \mathcal{A}^{U} \rbra*{ \ket{0} \bra{0}^{\otimes n} \otimes \underbrace{\ket{0} \bra{0}^{\otimes n_{\textup{env}}}}_{\textup{env}} } \rbra*{\mathcal{A}^{U}}^\dag } - \frac{e^{-\beta H}}{\tr\rbra{e^{-\beta H}}} }_1 \leq \varepsilon.
    \]
    
    The proof follows that of \cref{thm:gibbs}.
    The difference is that \cref{eq:gibbs-prob-eps} in the proof of \cref{thm:gibbs} becomes 
    \[
    \abs*{ \Pr\sbra*{\mathcal{T}^{U_{\rho_\pm}} \textup{ outputs } 0} - \tr\rbra*{\ket{0}\bra{0} \cdot \frac{e^{-\frac{1}{4}\beta \rho_\pm^2}}{\tr\rbra*{e^{-\frac{1}{4}\beta \rho_\pm^2}}}} } \leq \varepsilon.
    \]
    Then, a similar argument to the proof of \cref{thm:gibbs} still works here. This can be seen by noting that
    \[
    \frac{e^{-\frac 1 4\beta \rho_\pm^2}}{\tr\rbra*{e^{-\frac 1 4\beta \rho_\pm^2}}} = \frac{1}{e^{\frac 1 2}+e^{-\frac 1 2}} \rbra*{ e^{\pm \frac 1 2} \ket{0}\bra{0} + e^{\mp \frac 1 2} \ket{1}\bra{1} }
    \]
    under the same choice of $\rho_{\pm}$.
\end{proof}

\subsection{Entanglement entropy problem}

The entanglement entropy problem $\textsc{Entropy}_{N, a, b}$ with $0 < a < b \leq \ln\rbra{N}$ is a promise problem for unitary property testing, with
\begin{itemize}
    \item \textit{Yes} instance: $U = I - 2\ket{\psi}\bra{\psi}$ such that $S_2\rbra{\tr_B\rbra{\ket{\psi}\bra{\psi}}} \leq a$,
    \item \textit{No} instance: $U = I - 2\ket{\psi}\bra{\psi}$ such that $S_2\rbra{\tr_B\rbra{\ket{\psi}\bra{\psi}}} \geq b$,
\end{itemize}
where $S_2\rbra{\rho} = -\ln\rbra{\tr\rbra{\rho^2}}$ is the $2$-R\'enyi entropy, $\ket{\psi} \in \mathcal{H}_A \otimes \mathcal{H}_B$ and $\mathcal{H}_A = \mathcal{H}_B = \mathbb{C}^N$. 

\begin{theorem} [New bound for $\textsc{Entropy}_{N, a, b}$]
\label{thm:entangle-entro}
For $0 < \Delta = b-a \leq 1/4$, any quantum tester for $\textsc{Entropy}_{N, a, b}$ has query complexity $\widetilde \Omega\rbra{1/\sqrt{\Delta}}$.
\end{theorem}

\begin{proof}
    We reduce $\bb{\textsc{Dis}_{\rho, \sigma}}$ to $\textsc{Entropy}_{N, a, b}$, where
    \[
    \rho = \frac 1 2 \ket{0}\bra{0} + \frac 1 2 \ket{1}\bra{1}, \qquad \sigma = \rbra*{\frac 1 2 - \sqrt{\Delta}} \ket{0}\bra{0} + \rbra*{\frac 1 2 +\sqrt{\Delta}} \ket{1}\bra{1},
    \]
    with $\Delta=b-a$ and $b=\ln \rbra{2}$.

    \paragraph{Reduction}
    Suppose that there is a quantum tester $\mathcal{A}$ for $\textsc{Entropy}_{N, a, b}$ with quantum query complexity $Q$
    such that
    \begin{enumerate}
        \item $\Abs{\Pi \mathcal{A}^U \ket{0}}^2 \geq 3/4$ if $U$ is a \textit{yes} instance of $\textsc{Entropy}_{N, a, b}$,
        \item $\Abs{\Pi \mathcal{A}^U \ket{0}}^2 \leq 1/4$ if $U$ is a \textit{no} instance of $\textsc{Entropy}_{N, a, b}$,
    \end{enumerate}
    where $\Pi = \ket{0}\bra{0} \otimes I$ is the projector that measures the first qubit in the computational basis. 
    In the following, we construct a tester $\mathcal{T}$ for $\bb{\textsc{Dis}_{\rho, \sigma}}$ using $\mathcal{A}$.
    
    \vspace{5pt}
    \noindent\fbox{
    \parbox{\textwidth-15.6pt}{
    \vspace{7pt}
    $\mathcal{T}$ --- \textbf{Tester for} $\bb{\textsc{Dis}_{\rho, \sigma}}$ \textbf{using the tester} $\mathcal{A}$ \textbf{for} $\textsc{Entropy}_{N, a, b}$

    \begin{itemize}
        \item[] \textbf{Input}: Quantum oracle $U$ that is a $\rbra{1, m, 0}$-block-encoding of a quantum state $\frac{1}{2}\varrho$.
        \item[] \textbf{Output}: A classical bit $c \in \cbra{0, 1}$. 
    \end{itemize}

    \begin{enumerate}
        \item Let $\varepsilon = 1/240Q$.
        \item Let $U_{\ket{\varrho}\bra{\varrho}}$ be a $\rbra{2, O\rbra{n+m}, \varepsilon}$-block-encoding of $\ket{\varrho}\bra{\varrho}$ by \cref{lemma:block-encoding-of-purification}, where $\ket{\varrho}$ is a purification of $\varrho$.
        \item Let $U_{R_{\varrho}}$ be a $\rbra{1, O(n+m), 20\varepsilon}$-block-encoding of $R_{\varrho} = I - 2\ket{\varrho}\bra{\varrho}$ by \cref{lemma:lcu} and \cref{lmm:roaa}.
        \item Let $\mathcal{A}^{R_{\varrho} \to U_{R_{\varrho}}}$ be a $\rbra{1, n_{\textup{env}},20Q\varepsilon}$-block-encoding of $\mathcal{A}^{R_{\varrho}}$ by \cref{lemma:block-encoding-as-query}, where $n_{\textup{env}} = O\rbra{Q\rbra{n+m}}$. 
        \item Prepare $\ket{\psi} = \mathcal{A}^{R_{\varrho} \to U_{R_{\varrho}}} \ket{0} \ket{0}^{\otimes n_{\textup{env}}}$. 
        \item Measure the first and ancilla qubits of $\ket{\psi}$ in the computational basis. 
        Return $0$ if all the outcomes are $0$, and $1$ otherwise. 
    \end{enumerate}
    \vspace{7pt}
    }
    }
    \vspace{5pt}

    \paragraph{Complexity}
    Let us analyze the quantum query complexity of $\mathcal{T}$. 
    By \cref{lemma:block-encoding-of-purification} (with $\beta=2$, $\kappa = 4$ and $N = 2$), $U_{\ket{\varrho}\bra{\varrho}}$ can be implemented using $O\rbra{\log^{2}\rbra{1/\varepsilon}}$ queries to $U$.
    To implement $U_{R_{\varrho}}$,
    we first take $y=\rbra{1,-4}$ and $\beta=5$ in \cref{def:state-preparation-pair},
    then it is easy to see that $\rbra{P_L,P_R}$ with
    \[
    P_L\ket{0}=\frac{1}{\sqrt{5}}\ket{0}+\frac{2}{\sqrt{5}}\ket{1}, \text{ and } P_R\ket{0}=\frac{1}{\sqrt{5}}\ket{0}-\frac{2}{\sqrt{5}}\ket{1}
    \]
    is a $(5,1,0)$-state-preparation-pair and can be efficiently implemented without queries.
    By taking $A_1=I$ and $A_2=\frac{1}{2}\ket{\varrho}\bra{\varrho}$ in \cref{lemma:lcu}, we can then construct a $\rbra*{5,O(n+m),2.5\varepsilon}$-block-encoding of $R_{\varrho}$,
    which is equivalently a $\rbra*{1,O(n+m),\varepsilon/2}$-block-encoding of $R_{\varrho}/5$,
    using $O\rbra{\log^{2}(1/\varepsilon)}$ queries to $U$.
    As a result, combined with \cref{lmm:roaa}, 
    $U_{R_{\varrho}}$ can be implemented using $O\rbra{\log^{2}(1/\varepsilon)}$ queries to $U$.
    Since $\mathcal{A}^{R_{\varrho} \to U_{R_{\varrho}}}$ uses $Q$ queries to $U_{R_{\varrho}}$,
    the total quantum query complexity of $\mathcal{T}$ is $O\rbra{Q\log^{2}(1/\varepsilon)}=\widetilde{O}(Q)$.
    
    On the other hand, if $\mathcal{T}$ solves $\bb{\textsc{Dis}_{\rho, \sigma}}$, then $\widetilde O\rbra{Q} \geq \mathsf{Q}_\diamond\rbra{\textsc{Dis}_{\rho, \sigma}}$.
    Note that
    \[
    \gamma = 1 - \operatorname{F}\rbra{\rho, \sigma} = 1 - \frac{1}{\sqrt{2}}\rbra*{\sqrt{\frac{1}{2} -\sqrt{\Delta} } + \sqrt{\frac{1}{2} +\sqrt{\Delta}}} \leq 2\Delta.
    \]
    By \cref{corollary:qsd-intro}, $\mathsf{Q}_\diamond\rbra{\textsc{Dis}_{\rho, \sigma}} = \widetilde\Omega\rbra{1/\sqrt{\gamma}} = \widetilde\Omega\rbra{1/\sqrt{\Delta}}$.
    Therefore, $Q = \widetilde\Omega\rbra{1/\sqrt{\Delta}}$. 

    \paragraph{Correctness}
    To complete the proof, it remains to show that $\mathcal{T}$ is a tester for $\bb{\textsc{Dis}_{\rho, \sigma}}$. 
    To begin with, let $U_{\rho}$ and $U_{\sigma}$ be block-encodings of $\frac{1}{2}\rho$ and $\frac{1}{2}\sigma$, respectively. 
    We only have to show that 
    \begin{enumerate}
        \item $\Pr \sbra{ \mathcal{T}^{U_{\rho}} \textup{ outputs } 0 } \leq 1/3$,
        \item $\Pr \sbra{ \mathcal{T}^{U_{\sigma}} \textup{ outputs } 0 } \geq 2/3$.
    \end{enumerate}
    Note that $S_2\rbra{\rho} = \ln\rbra{2} \geq b$ and $S_2\rbra{\sigma} = -\ln\rbra{1/2+2\Delta} \leq \ln\rbra{2} - \Delta = a$.
    From this, we know that $R_\rho = I - 2\ket{\rho}\bra{\rho}$ is a \textit{no} instance of $\textsc{Entropy}_{N, a, b}$ and $R_\sigma = I - 2\ket{\sigma}\bra{\sigma}$ is a \textit{yes} instance of $\textsc{Entropy}_{N, a, b}$.
    Therefore, 
    \begin{equation} \label{eq:succ-prob}
    \Abs*{\Pi \mathcal{A}^{R_\rho} \ket{0}}^2 \leq \frac 1 4, \text{ and } \Abs*{\Pi \mathcal{A}^{R_\sigma} \ket{0}}^2 \geq \frac 3 4.
    \end{equation}
    For $\varrho \in \cbra{\rho, \sigma}$, we have
    \begin{align}
    \Pr \sbra*{\mathcal{T}^{U_{\varrho}}\textup{ outputs $0$}}
    & = \Abs*{ \rbra*{\Pi \otimes \ket{0}\bra{0}^{\otimes n_{\textup{env}}} } \mathcal{A}^{R_{\varrho} \to U_{R_{\varrho}}} \ket{0} \ket{0}^{\otimes n_{\textup{env}}}}^2 \nonumber \\
    & = \Abs*{ \Pi \rbra*{ \bra{0}^{\otimes n_{\textup{env}}} \mathcal{A}^{R_{\varrho} \to U_{R_{\varrho}}} \ket{0}^{\otimes n_{\textup{env}}} } \ket{0} }^2. \label{eq:prob-entr}
    \end{align}
    Note that $\mathcal{A}^{R_{\varrho} \to U_{R_{\varrho}}}$ is a $\rbra{1, n_{\textup{env}}, 20Q\varepsilon}$-block-encoding of $\mathcal{A}^{R_{\varrho}}$, i.e., 
    \begin{equation} \label{eq:err-block}
    \Abs*{\bra{0}^{\otimes n_{\textup{env}}} \mathcal{A}^{R_{\varrho} \to U_{R_{\varrho}}} \ket{0}^{\otimes n_{\textup{env}}} - \mathcal{A}^{R_{\varrho}}} \leq 20Q\varepsilon = \frac 1 {12}.
    \end{equation}
    By \cref{eq:prob-entr} and \cref{eq:err-block}, we have
    \[
    \abs*{ \Pr\sbra*{\mathcal{T}^{U_{\varrho}}\textup{ outputs } 0} - \Abs{\Pi \mathcal{A}^{R_{\varrho}} \ket{0}}^2 } \leq \frac 1 {12}.
    \]
    Then by \cref{eq:succ-prob}, we have
    \begin{align*}
        \Pr \sbra*{\mathcal{T}^{U_\rho}\textup{ outputs $0$}}& \leq \Abs*{\Pi \mathcal{A}^{R_\rho} \ket{0}}^2 + \frac 1 {12} \leq \frac 1 4 + \frac 1 {12} = \frac 1 3, \\
        \Pr \sbra*{\mathcal{T}^{U_\sigma}\textup{ outputs $0$}} & \geq \Abs*{\Pi \mathcal{A}^{R_\sigma} \ket{0}}^2 - \frac 1 {12} \geq \frac 3 4 - \frac 1 {12} = \frac 2 3.
    \end{align*}
    We conclude that $\mathcal{T}$ is a tester for $\bb{\textsc{Dis}_{\rho, \sigma}}$.
\end{proof}

\subsection{Hamiltonian simulation}

A quantum query algorithm $\mathcal{A}$ is for Hamiltonian simulation for time $t$ with precision $\varepsilon$ (in the diamond norm distance), if for every block-encoding $U$ of an $n$-qubit Hamiltonian $H$, $\mathcal{A}^U$ implements a quantum channel that is $\varepsilon$-close to the unitary operator $e^{-iHt}$ in the diamond norm distance.

\begin{theorem} [Optimality of Hamiltonian simulation]
\label{thm:ham-sim}
    For any $t \geq 2\pi$ and $0 \leq \varepsilon \leq 1/3$, any quantum query algorithm for Hamiltonian simulation for time $t$ with precision $\varepsilon$ (in the trace norm distance) has quantum query complexity $\widetilde\Omega\rbra{t}$.
\end{theorem}

\begin{remark}
    Due to \cref{eq:tr-vs-diamond}, every quantum algorithm for Hamiltonian simulation with respect to the diamond norm distance $\Abs{\cdot}_\diamond$ (which is usually considered in the literature, e.g., \cite{BCK15}) is a valid one with respect to the trace norm distance $\Abs{\cdot}_{\tr}$, but not vice versa. As a result, any lower bound for Hamiltonian simulation with respect to the trace norm distance is also a lower bound with respect to the diamond norm distance, but not vice versa.
\end{remark}

\begin{proof}
    We reduce $\bb{\textsc{Dis}_{\rho, \sigma}}$ to Hamiltonian simulation, where 
    \[
    \rho = \frac 1 2 \ket{0}\bra{0} + \frac 1 2 \ket{1}\bra{1}, \qquad \sigma = \rbra*{\frac 1 2 + \frac{\pi}{t}} \ket{0}\bra{0} + \rbra*{\frac 1 2 - \frac{\pi}{t}} \ket{1}\bra{1}.
    \]

    \paragraph{Reduction}
    Suppose that there is a quantum query algorithm $\mathcal{A}$ for Hamiltonian simulation for time $t$ with precision $\varepsilon$ (in the trace norm distance), which has quantum query complexity $Q$, such that for every block-encoding $U$ of an $n$-qubit Hamiltonian $H$, $\Abs{\mathcal{E}^{U} - e^{-iHt}}_{\tr} \leq \varepsilon$, where
    \[
    \mathcal{E}^U\rbra{\varrho} = \tr_{\textup{env}} \rbra*{ \mathcal{A}^U \rbra*{ \varrho \otimes \underbrace{\ket{0}\bra{0}^{\otimes n_{\textup{env}}}}_{\textup{env}} } \rbra*{\mathcal{A}^U}^\dag }.
    \]
    We reduce $\bb{\textsc{Dis}_{\rho, \sigma}}$ to Hamiltonian simulation by constructing a tester $\mathcal{T}$ for $\bb{\textsc{Dis}_{\rho, \sigma}}$ using the quantum Hamiltonian simulator $\mathcal{A}$ as follows.
    
    \vspace{5pt}
    \noindent\fbox{
    \parbox{\textwidth-15.6pt}{
    \vspace{7pt}
    $\mathcal{T}$ --- \textbf{Tester for} $\bb{\textsc{Dis}_{\rho, \sigma}}$ \textbf{using quantum Hamiltonian simulator} $\mathcal{A}$

    \begin{itemize}
        \item[] \textbf{Input}: Quantum oracle $U$.
        \item[] \textbf{Output}: A classical bit $b \in \cbra{0, 1}$. 
    \end{itemize}

    \begin{enumerate}
        \item Prepare $\ket{\psi} = \rbra{\mathit{Had} \otimes I_{\textup{env}}} \cdot \mathcal{A}^U \ket{+} \ket{0}^{\otimes n_{\textup{env}}}$, where $\mathit{Had}$ is the Hadamard gate. 
        \item Let $b$ be the measurement outcome of the first qubit of $\ket{\psi}$ in the computational basis. 
        \item Return $b$. 
    \end{enumerate}
    \vspace{7pt}
    }
    }
    \vspace{5pt}

    \paragraph{Complexity}

    Clearly, it can be seen that the quantum query complexity of $\mathcal{T}$ is exactly that of $\mathcal{A}$, which is $Q$. 
    On the other hand, if $\mathcal{T}$ solves $\bb{\textsc{Dis}_{\rho, \sigma}}$, then the quantum query complexity of $\mathcal{T}$ is lower bounded by that of $\bb{\textsc{Dis}_{\rho, \sigma}}$, i.e., $Q \geq \mathsf{Q}_\diamond\rbra{\textsc{Dis}_{\rho, \sigma}}$.
    Note that 
    \[
    \gamma = 1 - \operatorname{F}\rbra{\rho, \sigma} = 1 - \frac 1 2 \rbra*{\sqrt{1 + \frac{2\pi}{t}} + \sqrt{1 - \frac{2\pi}{t}}} \leq \frac {4\pi^2} {t^2}.
    \]
    By \cref{corollary:qsd-intro}, we have the quantum query lower bound $Q \geq \mathsf{Q}_\diamond\rbra{\textsc{Dis}_{\rho, \sigma}} = \widetilde \Omega\rbra{1/\sqrt{\gamma}} = \widetilde \Omega\rbra{t}$ for Hamiltonian simulation.

    \paragraph{Correctness}
    To complete the proof, it remains to show that $\mathcal{T}$ is a tester for $\bb{\textsc{Dis}_{\rho, \sigma}}$. 
    To begin with, let $U_{\rho}$ and $U_{\sigma}$ be block-encodings of $\frac{1}{2}\rho$ and $\frac{1}{2}\sigma$, respectively. 
    We only have to show that 
    \begin{enumerate}
        \item $\Pr \sbra{ \mathcal{T}^{U_{\rho}} \textup{ outputs } 0 } \geq 2/3$,
        \item $\Pr \sbra{ \mathcal{T}^{U_{\sigma}} \textup{ outputs } 0 } \leq 1/3$.
    \end{enumerate}
    Let $\Pi = \ket{0}\bra{0} \otimes I_{\textup{env}}$ be a projector onto the first qubit.
    For $\eta \in \cbra{\rho, \sigma}$, we have
    \begin{align*}
    \Pr\sbra*{\mathcal{T}^{U_{\eta}} \textup{ outputs } 0} 
    & = \Abs*{\Pi \rbra{\mathit{Had} \otimes I_{\textup{env}}} \cdot \mathcal{A}^{U_{\eta}} \ket{+} \ket{0}^{\otimes n_{\textup{env}}}}^2  \\
    & = \tr\rbra*{ \ket{+}\bra{+} \cdot \tr_{\textup{env}} \rbra*{ \mathcal{A}^{U_{\eta}} \rbra*{ \ket{+} \bra{+} \otimes \underbrace{\ket{0} \bra{0}^{\otimes n_{\textup{env}}}}_{\textup{env}} } \rbra*{\mathcal{A}^{U_{\eta}}}^\dag } }  \\
    & = \tr\rbra*{ \ket{+}\bra{+} \cdot \mathcal{E}^{U_{\eta}}\rbra*{\ket{+}\bra{+}} },
    \end{align*}
    where $\Abs{\mathcal{E}^{U_{\eta}} - e^{-i \frac{\eta}{2} t}}_{\tr} \leq \varepsilon$. 
    This implies that
    \begin{equation} \label{eq:hsim-prob-eps}
    \abs*{ \Pr\sbra*{\mathcal{T}^{U_{\eta}} \textup{ outputs } 0} - \tr\rbra*{ \ket{+}\bra{+} \cdot e^{-i \frac{\eta}{2} t} \ket{+}\bra{+} e^{i \frac{\eta}{2} t}} } \leq \varepsilon.
    \end{equation}
    It can be verified that
    \[
    e^{-i \frac{\rho}{2} t} \ket{+} = e^{-i\frac{t}{4}} \ket{+}, \qquad e^{-i \frac{\sigma}{2} t} \ket{+} = -i e^{-i\frac{t}{4}} \ket{-}.
    \]
    and thus
    \[
    \tr\rbra*{ \ket{+}\bra{+} \cdot e^{-i \frac{\rho}{2} t} \ket{+}\bra{+} e^{i \rho t}} = 1, \qquad \tr\rbra*{ \ket{+}\bra{+} \cdot e^{-i \frac{\sigma}{2} t} \ket{+}\bra{+} e^{i \sigma t}} = 0.
    \]
    By \cref{eq:hsim-prob-eps}, we can consider $\rho$ and $\sigma$ separately as follows:
    \begin{align*}
    \Pr\sbra*{\mathcal{T}^{U_{\rho}} \textup{ outputs } 0} & \geq \tr\rbra*{ \ket{+}\bra{+} \cdot e^{-i \frac{\rho}{2} t} \ket{+}\bra{+} e^{i \frac{\rho}{2} t}} - \varepsilon = 1 - \varepsilon \geq \frac 2 3, \\
    \Pr\sbra*{\mathcal{T}^{U_{\sigma}} \textup{ outputs } 0} & \leq \tr\rbra*{ \ket{+}\bra{+} \cdot e^{-i \frac{\sigma}{2} t} \ket{+}\bra{+} e^{i \frac{\sigma}{2} t}} + \varepsilon = 0 + \varepsilon \leq \frac 1 3.
    \end{align*}
    We conclude that $\mathcal{T}$ is a tester for $\bb{\textsc{Dis}_{\rho, \sigma}}$.
\end{proof}

\subsection{Phase estimation}

A quantum query algorithm $\mathcal{A}$ is for phase estimation with precision $\delta$, if for every unitary operator $U$ and its eigenvector $\ket{\psi}$ with $U \ket{\psi} = e^{i\lambda} \ket{\psi}$, $\mathcal{A}^{U}$ with input state $\ket{\psi}$ returns $\tilde \lambda$ such that $\abs{\tilde \lambda - \lambda \pmod{2\pi}} \leq \delta$ with probability at least $2/3$. 

\begin{theorem} [Optimality of phase estimation]
\label{thm:ph-est}
    For $0 < \delta \leq 1/8$, any quantum query algorithm for phase estimation with precision $\delta$ has quantum query complexity $\widetilde \Omega\rbra{1/\delta}$.
\end{theorem}

\begin{proof}
    We reduce $\bb{\textsc{Dis}_{\rho_+, \rho_-}}$ to phase estimation, where 
    \[
    \rho_{\pm} = \rbra*{\frac 1 2 \mp 4\delta} \ket{0}\bra{0} + \rbra*{\frac 1 2 \pm 4\delta} \ket{1}\bra{1}.
    \]

    \paragraph{Reduction}
    Suppose that there is a quantum query algorithm $\mathcal{A}$ for phase estimation with precision $\delta$, which has quantum query complexity $Q$.
    We reduce $\bb{\textsc{Dis}_{\rho_+, \rho_-}}$ to phase estimation by constructing a tester $\mathcal{T}$ for $\bb{\textsc{Dis}_{\rho_+, \rho_-}}$ using the quantum phase estimator $\mathcal{A}$ as follows.
    
    \vspace{5pt}
    \noindent\fbox{
    \parbox{\textwidth-15.6pt}{
    \vspace{7pt}
    $\mathcal{T}$ --- \textbf{Tester for} $\bb{\textsc{Dis}_{\rho_+, \rho_-}}$ \textbf{using quantum phase estimator} $\mathcal{A}$

    \begin{itemize}
        \item[] \textbf{Input}: Quantum oracle $U$ that is a $\rbra{1, a, 0}$-block-encoding of a one-qubit Hermitian $H$.
        \item[] \textbf{Output}: A classical bit $b \in \cbra{0, 1}$. 
    \end{itemize}

    \begin{enumerate}
        \item Let $\varepsilon = 1/9Q$. 
        \item Let $V$ be a $\rbra{1, a+2, \varepsilon}$-block-encoding of $e^{-iH}$ 
        by \cref{thm:hamiltonian-simulation}.
        \item Let $\mathcal{A}^{e^{-iH} \to V}$ be a $\rbra{1, Q\rbra{a+2}, Q\varepsilon}$-block-encoding of $\mathcal{A}^{e^{-iH}}$ 
        by \cref{lemma:block-encoding-as-query}.
        \item Let $\tilde \lambda$ be the estimate of the phase $\lambda$ of $\ket{0}$ returned by $\mathcal{A}^{e^{-iH} \to V}$.
        \item Return $0$ if $\tilde \lambda < 1/2$, and $1$ otherwise. 
    \end{enumerate}
    \vspace{7pt}
    }
    }
    \vspace{5pt}

    \paragraph{Complexity} 
    The unitary operator $V$ can be implemented by \cref{thm:hamiltonian-simulation} using $O\rbra*{\frac{\log\rbra{1/\varepsilon}}{\log\log\rbra{1/\varepsilon}}}$ queries to $U$. 
    Therefore, the quantum query complexity of $\mathcal{T}$ is 
    \[
    Q \cdot O\rbra*{\frac{\log\rbra{1/\varepsilon}}{\log\log\rbra{1/\varepsilon}}} = \widetilde O\rbra{Q}.
    \]
    On the other hand, if $\mathcal{T}$ solves $\bb{\textsc{Dis}_{\rho_+, \rho_-}}$, then the quantum query complexity of $\mathcal{T}$ is lower bounded by that of $\bb{\textsc{Dis}_{\rho_+, \rho_-}}$, i.e., $\widetilde O\rbra{Q} \geq \mathsf{Q}_\diamond\rbra{\textsc{Dis}_{\rho_+, \rho_-}}$.
    Note that $
    \gamma = 1 - \operatorname{F}\rbra{\rho_+, \rho_-} = 1 - \sqrt{1 - 64\delta^2} \leq 64 {\delta^2}$.
    By \cref{corollary:qsd-intro}, we have the quantum query lower bound $Q = \widetilde\Omega\rbra{\mathsf{Q}_\diamond\rbra{\textsc{Dis}_{\rho_+, \rho_-}}} = \widetilde \Omega\rbra{1/\sqrt{\gamma}} = \widetilde\Omega\rbra{1/\delta}$ for phase estimation.

    \paragraph{Correctness}
    To complete the proof, it remains to show that $\mathcal{T}$ is a tester for $\bb{\textsc{Dis}_{\rho_+, \rho_-}}$. 
    To begin with, let $U_{\rho_+}$ and $U_{\rho_-}$ be block-encodings of $\frac{1}{2}\rho_+$ and $\frac{1}{2}\rho_-$, respectively. 
    We only have to show that 
    \begin{enumerate}
        \item $\Pr \sbra{ \mathcal{T}^{U_{\rho_+}} \textup{ outputs } 0 } \geq 5/9$,
        \item $\Pr \sbra{ \mathcal{T}^{U_{\rho_-}} \textup{ outputs } 0 } \leq 4/9$.
    \end{enumerate}
    Note that
    \[
    e^{-i \frac{\rho_{\pm}}{2}} = e^{-i \rbra*{\frac 1 {4} \mp 2\delta}} \ket{0}\bra{0} + e^{-i \rbra*{\frac 1 {4} \pm 2\delta}} \ket{1}\bra{1}.
    \]
    The phase of $\ket{0}$ of $e^{-i \frac{\rho_{\pm}}{2}}$ is $\lambda_{\pm} = \frac 1 {4} \mp 2\delta$.
    Let $V_{\pm}$ be the unitary operator constructed at Step 2, using queries to $U_{\rho_{\pm}}$.
    The phase estimator $\mathcal{A}^{e^{-i\frac{\rho_{\pm}}{2}}\to V_{\pm}}$ will return an estimate $\tilde \lambda_{\pm}$ such that 
    \[
    \Pr\sbra*{ \abs*{ \tilde \lambda_{\pm} - \lambda_{\pm} } \leq \delta } \geq \frac 2 3 - Q\varepsilon = \frac{5}{9},
    \]
    which implies that
    \[
    \Pr \sbra*{ \tilde \lambda_+ \leq \frac{1}{4} - \delta } \geq \frac 5 9, \qquad \Pr \sbra*{ \tilde \lambda_- \geq \frac{1}{4} + \delta } \geq \frac 5 9.
    \]
    Finally, we have
    \begin{align*}
    \Pr\sbra*{\mathcal{T}^{U_{\rho_+}} \textup{ outputs } 0} & = \Pr\sbra*{\tilde \lambda_+ < \frac 1 {4}} \geq \Pr \sbra*{ \tilde \lambda_+ \leq \frac{1}{4} - \delta } \geq \frac 5 9, \\
    \Pr\sbra*{\mathcal{T}^{U_{\rho_-}} \textup{ outputs } 0} & = \Pr\sbra*{\tilde \lambda_- < \frac 1 {4}} \leq 1 - \Pr \sbra*{ \tilde \lambda_- \geq \frac 1 {4} + \delta } \leq \frac 4 9.
    \end{align*}
    We conclude that $\mathcal{T}$ is a tester for $\bb{\textsc{Dis}_{\rho_+, \rho_-}}$.
\end{proof}

\subsection{Amplitude estimation}
\label{sec:ampl-est}

A quantum query algorithm $\mathcal{A}$ is for amplitude estimation with precision $\varepsilon$, if for every unitary operator 
\[
U \ket{0} = \sqrt{p} \ket{0}\ket{\phi_0} + \sqrt{1-p} \ket{1}\ket{\phi_1},
\]
where $\ket{\phi_0}$ and $\ket{\phi_1}$ are some normalized pure states, 
$\mathcal{A}^U$ returns $\tilde p$ such that $\abs{\tilde p - p} \leq \varepsilon$ with probability at least $2/3$. 

\begin{theorem} [Optimality of amplitude estimation]
\label{thm:ampl-est}
    For $0 < \varepsilon \leq 1/16$, any quantum query algorithm for amplitude estimation with precision $\varepsilon$ has quantum query complexity $\widetilde \Omega\rbra{1/\varepsilon}$.
\end{theorem}

\begin{proof}
    We reduce $\bb{\textsc{Dis}_{\rho_+, \rho_-}}$ to amplitude estimation, where
    \[
    \rho_{\pm} = \rbra*{\frac 1 2 \mp 8\varepsilon} \ket{0}\bra{0} + \rbra*{\frac 1 2 \pm 8\varepsilon} \ket{1}\bra{1}.
    \]
    The reduction follows the idea of the proof of \cref{thm:tightness}. 
    Here, We just outline how to establish a quantum query lower bound for amplitude estimation. 

    \paragraph{Reduction}
    Suppose that there is a quantum query algorithm $\mathcal{A}$ for amplitude estimation with precision $\varepsilon$, which has quantum query complexity $Q$.
    We reduce $\bb{\textsc{Dis}_{\rho_+, \rho_-}}$ to amplitude estimation by the tester $\mathcal{T}$ for $\bb{\textsc{Dis}_{\rho_+, \rho_-}}$ using the quantum amplitude estimator $\mathcal{A}$ given in the proof of \cref{thm:tightness}.
    
    \paragraph{Complexity}
    Clearly, it can be seen that the quantum query complexity of $\mathcal{T}$ is exactly that of $\mathcal{A}$, which is $Q$. 
    On the other hand, if $\mathcal{T}$ solves $\bb{\textsc{Dis}_{\rho_+, \rho_-}}$, then the quantum query complexity of $\mathcal{T}$ is lower bounded by that of $\bb{\textsc{Dis}_{\rho_+, \rho_-}}$, i.e., $Q \geq \mathsf{Q}_\diamond\rbra{\textsc{Dis}_{\rho_+, \rho_-}}$.
    By the lower bound given in \cref{thm:tightness}, we have the quantum query lower bound $Q \geq \mathsf{Q}_\diamond\rbra{\textsc{Dis}_{\rho_+, \rho_-}} = \widetilde\Omega\rbra{1/\varepsilon}$ for amplitude estimation.

    \paragraph{Correctness}

    The correctness has already been shown in \cref{thm:tightness}.
\end{proof}

\subsection{Matrix spectrum testing}

We can also derive a series of quantum query lower bounds for testing properties of the spectrum of a matrix block-encoded in a unitary operator, by reducing from testing properties of the spectrum of quantum states \cite{OW21}. 

\begin{corollary} [Rank Testing]
\label{corollary:rank-testing}
    Testing whether a matrix $A \in \mathbb{C}^{N \times N}$ block-encoded in a unitary operator $U$ has rank $r$ or $\varepsilon$-far (in trace norm) requires $\widetilde\Omega\rbra{\sqrt{r/\varepsilon}}$ queries to $U$. 
\end{corollary}

\begin{proof}
    We reduce from the promise problem $\textsc{Rank}_{r,\varepsilon}$ for quantum state testing, with
    \begin{itemize}
        \item \textit{Yes} instance: a quantum state $\rho$ of rank $\leq r$, 
        \item \textit{No} instance: a quantum state $\rho$ such that $\Abs{\rho - \sigma}_1 > \varepsilon$ for any quantum state $\sigma$ of rank $r$.
    \end{itemize}
    It was shown in \cite[Theorem 1.9]{OW21} that $\mathsf{S}\rbra{\textsc{Rank}_{r,\varepsilon}} = \Omega\rbra{r/\varepsilon}$.
    By \cref{thm:lifting}, we immediately have $\mathsf{Q}_\diamond\rbra{\textsc{Rank}_{r,\varepsilon}} = \widetilde\Omega\rbra{\sqrt{r/\varepsilon}}$.
    Let $\textsc{RankTesting}_{r, \varepsilon}$ be the problem of Rank Testing with parameters $r$ and $\varepsilon$.
    To complete the proof, we note that any tester for $\textsc{RankTesting}_{r, \varepsilon}$ can be used to solve $\textsc{Rank}_{r,2\varepsilon}^\diamond$.
\end{proof}

\begin{corollary} [Mixedness Testing]
\label{corollary:mixed-testing}
    Testing whether an Hermitian matrix $A \in \mathbb{C}^{N \times N}$ block-encoded in a unitary operator $U$ has spectrum $\rbra{1/N, 1/N, \dots, 1/N}$ or $\varepsilon$-far (in trace norm) requires $\widetilde\Omega\rbra{\sqrt{N}/\varepsilon}$ queries to $U$. 
\end{corollary}

\begin{proof}
    Let $N \geq 8$ be an even integer and $\varepsilon \in (0, 1/4)$.
    We reduce from the promise problem $\textsc{Mix}_{N,\varepsilon}$ for quantum state testing, with
    \begin{itemize}
        \item \textit{Yes} instance: an $N$-dimensional quantum state $\rho$ with the uniform spectrum $\rbra{1/N, \dots, 1/N}$, 
        \item \textit{No} instance: an $N$-dimensional quantum state $\rho$ with $\frac{N}{2}$ eigenvalues being $\frac{1+2\varepsilon}{N}$ and the rest $\frac{N}{2}$ eigenvalues being $\frac{1-2\varepsilon}{N}$.
    \end{itemize}
    It was shown in \cite[Theorem 1.8]{OW21} that $\mathsf{S}\rbra{\textsc{Mix}_{N,\varepsilon}} = \Omega\rbra{N/\varepsilon^2}$.
    By \cref{thm:lifting}, we immediately have $\mathsf{Q}_\diamond\rbra{\textsc{Mix}_{N,\varepsilon}} = \widetilde\Omega\rbra{\sqrt{N}/\varepsilon}$.
    This means that, given unitary operator $U$ that is a block-encoding of $\rho/2$, it requires $\widetilde \Omega\rbra{\sqrt{N}/\varepsilon}$ queries to $U$ to check whether $\rho$ is a \textit{yes} instance or a \textit{no} instance.
    Suppose there is a tester $\mathcal{T}$ for $\textsc{MixednessTesting}_{N,\varepsilon}$, the problem of Mixedness Testing with parameters $N$ and $\varepsilon$, with query complexity $Q$.
    Note that a unitary $\rbra{1, a, 0}$-block-encoding $U$ of $\frac{1}{2}\rho$ is actually a $\rbra{4, a, 0}$-block-encoding of $2\rho$.
    Using \cref{corollary:up-scaling-block-encoding} with $\alpha \coloneqq 4$, $\beta \coloneqq 2$, $A \coloneqq 2\rho$ (note that $\Abs{\rho} \leq \frac{1+2\varepsilon}{N} \leq 1/4$ holds in both cases), $\varepsilon \coloneqq 0$, and $\varepsilon' \coloneqq 1/9Q < 1/2$, 
    we can implement a unitary operator $\widetilde U$ using $O\rbra{\log\rbra{1/\varepsilon'}}$ queries to $U$ such that
    \[
    \Abs*{\widetilde U - U'} \leq \frac{2\sqrt{2}\varepsilon'}{\sqrt{16-\rbra*{2+\frac{1}{2}}^2}} \leq \varepsilon'.
    \]
    where $U'$ is a unitary operator that is a $\rbra{2, a+3, 0}$-block-encoding of $2\rho$, i.e., a $\rbra{1, a+3, 0}$-block-encoding of $\rho$.
    Therefore, $\Abs{\mathcal{T}^{\widetilde U} - \mathcal{T}^{U'}} \leq Q\varepsilon'$, which implies that
    \[
    \abs*{ \Pr\sbra{\mathcal{T}^{\widetilde U} \text{ outputs } b} - \Pr\sbra{\mathcal{T}^{U'} \text{ outputs } b} } \leq Q\varepsilon' = \frac{1}{9}
    \]
    for $b \in \cbra{0, 1}$. 
    On the other hand, 
    \begin{itemize}
        \item $\mathcal{T}^{U'}$ outputs $1$ with probability at least $2/3$ if $\rho$ is a \textit{yes} instance, and
        \item $\mathcal{T}^{U'}$ outputs $0$ with probability at least $2/3$ if $\rho$ is a \textit{no} instance.
    \end{itemize}
    This gives that
    \begin{itemize}
        \item $\mathcal{T}^{\widetilde U}$ outputs $1$ with probability at least $5/9$ if $\rho$ is a \textit{yes} instance, and
        \item $\mathcal{T}^{\widetilde U}$ outputs $0$ with probability at least $5/9$ if $\rho$ is a \textit{no} instance.
    \end{itemize}
    One can amplify the success probability of $\mathcal{T}^{\widetilde U}$ from $5/9$ to $2/3$ with a constant number of repetitions. 
    Note that $\mathcal{T}^{\widetilde U}$ can be implemented by $O\rbra{Q\log\rbra{1/\varepsilon'}}$ queries to $U$; this implementation of $\mathcal{T}^{\widetilde U}$ using $U$ is actually a tester for $\textsc{Mix}_{N,\varepsilon}^{\diamond}$. 
    Therefore, $O\rbra{Q\log\rbra{1/\varepsilon'}} \geq \mathsf{Q}_\diamond\rbra{\textsc{Mix}_{N,\varepsilon}} = \widetilde\Omega\rbra{\sqrt{N}/\varepsilon}$, which gives $Q = \widetilde \Omega\rbra{\sqrt{N}/\varepsilon}$ and completes the proof.
\end{proof}

\begin{corollary} [Uniformity Testing]
\label{corollary:uniformity-distinguishing}
    Testing whether an Hermitian matrix $A \in \mathbb{C}^{N \times N}$ block-encoded in a unitary operator $U$ has spectrum uniform on $r$ or $r+\Delta$ eigenvalues requires $\Omega^*\rbra{r/\sqrt{\Delta}}$ queries to $U$. 
\end{corollary}

\begin{proof}
    Let $r \geq 4$.
    We reduce from the promise problem $\textsc{Unif}_{r,\Delta}$ for quantum state testing, with
    \begin{itemize}
        \item \textit{Yes} instance: a quantum state $\rho$ with a uniform spectrum on $r$ eigenvalues, 
        \item \textit{No} instance: a quantum state $\rho$ with a uniform spectrum on $r+\Delta$ eigenvalues.
    \end{itemize}
    It was shown in \cite[Theorem 1.10]{OW21} that $\mathsf{S}\rbra{\textsc{Unif}_{r,\Delta}} = \Omega^*\rbra{r^2/\Delta}$.
    By \cref{thm:lifting}, we immediately have $\mathsf{Q}_\diamond\rbra{\textsc{Unif}_{r,\Delta}} = \Omega^*\rbra{r/\sqrt{\Delta}}$.
    To complete the proof, we need to approximately implement a unitary operator that is a block-encoding of $\rho$, using queries to unitary block-encoding of $\frac{1}{2}\rho$, in the same way (note that $\Abs{\rho} \leq 1/4$ holds in both cases) as in the proof of \cref{corollary:mixed-testing}. 
\end{proof}

\section*{Acknowledgments}
The work of Qisheng Wang was supported in part by the Engineering and Physical Sciences Research Council under Grant EP/X026167/1 and in part by the MEXT Quantum Leap Flagship Program (MEXT Q-LEAP) under Grant JPMXS0120319794. 
The work of Zhicheng Zhang was supported by the Sydney Quantum Academy, NSW, Australia.

The authors would like to thank Yupan Liu for pointing out the similarity of our results to the investigation of quantum query-to-communication lifting in \cite{ABDK21}.
They also thank Andr\'{a}s Gily\'{e}n for pointing out the query lower bound for quantum Gibbs sampling shown in \cite{CKBG23}.
They also thank Andr\'{a}s Gily\'{e}n, Yassine Hamoudi, Minbo Gao, and the anonymous reviewers for pointing out a mistake in the up-scaling of block-encoded operators in an earlier version of this paper.
The authors also thank the anonymous reviewers for their constructive suggestions and for pointing out a mistake in \cref{corollary:up-scaling-block-encoding} in an earlier version of this paper. 
Qisheng Wang would also like to thank Fran\c{c}ois Le Gall for helpful discussions.
The authors also thank Jinge Bao for pointing out that an earlier arXiv version of this paper contained incorrect theorem numbering.

\addcontentsline{toc}{section}{References}

\bibliographystyle{unsrturl}
\bibliography{main}

\appendix
\section{Up-Scaling of Block-Encoded Operators} \label{sec:up-scaling}

To prove \cref{corollary:up-scaling-block-encoding}, we recall singular value decomposition and transformation. 
Let $A \in \mathbb{C}^{\tilde n \times n}$ be a complex-valued matrix with singular value transformation $A = \sum_{j=1}^{n_{\min}} \varsigma_j \ket{\tilde \psi_j} \bra{\psi_j}$, where $\cbra{\ket{\tilde \psi_j}}$ and $\cbra{\ket{\psi_j}}$ are orthonormal bases, $n_{\min} \coloneqq \min\cbra{\tilde n, n}$ and $\varsigma_j \geq 0$. 
Let $f \colon \mathbb{R} \to \mathbb{C}$ be an even or odd function. 
The singular value transformation of $A$ corresponding to $f$ is defined by $f^{\SV}\rbra{A} \coloneqq \sum_{j=1}^{n_{\min}} f\rbra{\varsigma_j} \ket{\tilde \psi_j} \bra{\psi_j}$ if $f$ is odd and $f^{\SV}\rbra{A} \coloneqq \sum_{j=1}^{n_{\min}} f\rbra{\varsigma_j} \ket{\psi_j} \bra{\psi_j}$ if $f$ is even.
In particular, if $A$ is Hermitian, then $f^{\SV}\rbra{A} = f\rbra{A}$. 
We need the quantum singular value transformation for odd/even polynomials, given as follows. 

\begin{theorem} [Quantum singular value transformation for odd/even polynomials, {\cite[Theorem 2]{GSLW19}}] \label{thm:qsvt-odd-even}
    Suppose that $U$ is a $\rbra{1, a, 0}$-block-encoding of a matrix $A$ with $\Abs{A} \leq 1$. 
    Let $P \colon \sbra{-1, 1} \to \sbra{-1, 1}$ be an odd/even polynomial of degree $d$.
    Then, there is a quantum circuit $\widetilde U$ that is a $\rbra{1, a+1, 0}$-block-encoding of $P^{\SV}\rbra{A}$, using $O\rbra{d}$ queries to $U$.
\end{theorem}

We also need polynomial approximations of the rectangle function for our purpose. 

\begin{lemma} [Polynomial approximations of the rectangle function, {\cite[Corollary 16]{GSLW19}}] \label{lemma:rect}
    For every $\delta',\varepsilon' \in \rbra{0, \frac{1}{2}}$ and $t \in (\delta', 1-\delta')$, there is an efficiently computable even polynomial $P \in \mathbb{R}\sbra{x}$ of degree $O\rbra{\frac{1}{\delta'}\log\rbra{\frac{1}{\varepsilon'}}}$ such that
    \begin{itemize}
        \item For all $x \in \sbra{-1, 1}$, $\abs{P\rbra{x}} \leq 1$,
        \item For all $x \in \sbra{-t+\delta', t-\delta'}$, $P\rbra{x} \in \sbra{1-\varepsilon', 1}$,
        \item For all $x \in \sbra{-1, -t-\delta'} \cup \sbra{t+\delta', 1}$, $P\rbra{x} \in \sbra{0, \varepsilon'}$.
    \end{itemize}
\end{lemma}

To bound the error, we need the following lemma for the robustness of singular value transformation. 

\begin{lemma} [Robustness of singular value transformation, {\cite[Lemma 22 in the full version]{GSLW19}}] \label{lemma:robust-qsvt}
    Let $P \in \mathbb{C}\sbra{x}$ be an odd/even polynomial of degree $d$ such that
    \begin{enumerate}
        \item $P$ has parity $d \bmod 2$,
        \item For all $x \in \sbra{-1, 1}$, $\abs{P\rbra{x}} \leq 1$,
        \item For all $x \in (-\infty, -1] \cup [1, +\infty)$, $\abs{P\rbra{x}} \geq 1$,
        \item If $d$ is even, then for all $x \in \mathbb{R}$, $\abs{P\rbra{ix}} \geq 1$.
    \end{enumerate}
    For every $A, A' \in \mathbb{C}^{\tilde n \times n}$ with $\Abs{A} \leq 1$ and $\Abs{A'} \leq 1$, we have
    \[
    \Abs*{P^{\SV}\rbra{A} - P^{\SV}\rbra{A'}} \leq 4 d \sqrt{\Abs*{A - A'}}.
    \]
\end{lemma}

We also need the extensions of real-valued polynomials that satisfy the four conditions required in \cref{lemma:robust-qsvt}.

\begin{lemma} [Real quantum signal processing, {\cite[Corollary 10 in the full version]{GSLW19}}] \label{lemma:real-qsp}
    Let $R \in \mathbb{R}\sbra{x}$ be an odd/even polynomial of degree $d$ such that 
    \begin{itemize}
        \item $R$ has parity $d \bmod 2$,
        \item For all $x \in \sbra{-1, 1}$, $\abs{P\rbra{x}} \leq 1$.
    \end{itemize}
    Then, there exists a polynomial $S \in \mathbb{C}\sbra{x}$ of degree $d$ that satisfies the four conditions required in \cref{lemma:robust-qsvt} such that $S\rbra{x} = R\rbra{x}$ for all $x \in \mathbb{R}$.
\end{lemma}

Since we are dealing with block-encodings, we also need the following lemma for the robustness.

\begin{lemma} [Robustness of block-encodings, adapted from {\cite[Lemma 23 in the full version]{GSLW19}}, see also {\cite[Corollary 12]{GP22}}] \label{lemma:robust-block-encoding}
    Let $U$ be a $d \times d$ unitary operator that is a block-encoding of an $n \times n$ matrix $A$, where $d \geq 8n$.
    Then, for any $n \times n$ matrix $\tilde A$ satisfying
    \[
    \Abs*{A - \tilde A} + \Abs*{\frac{A + \tilde A}{2}}^2 \leq 1,
    \]
    there is a $d \times d$ unitary operator $\tilde U$ that is a block-encoding of $\tilde A$ satisfying
    \[
    \Abs*{U - \tilde U} \leq \sqrt{\frac{2}{1 - \Abs*{\dfrac{A+\tilde A}{2}}^2}} \Abs*{A - \tilde A}.
    \]
\end{lemma}

Now we are ready to prove \cref{corollary:up-scaling-block-encoding}.

\begin{proof} [Proof of \cref{corollary:up-scaling-block-encoding}]
    Let $P\rbra{x}$ be the polynomial given in \cref{lemma:rect} with $\delta' \coloneqq \frac{\beta-1}{2\alpha} \in \rbra{0, \frac{1}{2}}$, $\varepsilon' \coloneqq \varepsilon' \in \rbra{0, \frac{1}{2}}$ and $t = \frac{\beta+1}{2\alpha} \in \rbra{\delta', 1-\delta'}$ of degree $d = O\rbra{\frac{1}{\delta'}\log\rbra{\frac{1}{\varepsilon'}}} = O\rbra{\frac{\alpha}{\beta-1}\log\rbra{\frac{1}{\varepsilon'}}}$ such that
    \begin{itemize}
        \item $\abs{P\rbra{x}} \leq 1$ for $x \in \sbra{-1, 1}$, 
        \item $P\rbra{x} \in \sbra{1-\varepsilon', 1}$ for $x \in \sbra{-\frac{1}{\alpha},\frac{1}{\alpha}}$,
        \item $P\rbra{x} \in \sbra{0, \varepsilon'}$ for $x \in \sbra{-1, -\frac{\beta}{\alpha}} \cup \sbra{\frac{\beta}{\alpha}, 1}$.
    \end{itemize}
    Let $R\rbra{x} \coloneqq \frac{\alpha}{\beta}xP\rbra{x}$.
    It can be verified that 
    \begin{itemize}
        \item $\abs{R\rbra{x}} \leq 1$ for $x \in \sbra{-1, 1}$,
        \item $\abs{R\rbra{x} - \frac{\alpha}{\beta}x} \leq \frac{\varepsilon'}{\beta}$ for $x \in \sbra{-\frac{1}{\alpha}, \frac{1}{\alpha}}$.
    \end{itemize}
    Recall that the unitary operator $U$ is an $\rbra{\alpha, a, \varepsilon}$-block-encoding of $A$. Then, there is a linear operator $A'$ satisfying $\Abs{A'} \leq 1$ and $\Abs{\alpha A' - A} \leq \varepsilon$ such that $U$ is a $\rbra{1, a, 0}$-block-encoding of $A'$.
    By \cref{thm:qsvt-odd-even}, there is a quantum circuit $\widetilde U$ that is a $\rbra{1, a+1, 0}$-block-encoding of $R^{\SV}\rbra{A'}$, using $O\rbra{d} = O\rbra{\frac{\alpha}{\beta-1}\log\rbra{\frac{1}{\varepsilon'}}}$ queries to $U$. 
    To bound the error, note that
    \[
        \Abs{\beta R^{\SV}\rbra{A'} - A}
        \leq \beta \rbra*{\Abs*{R^{\SV}\rbra{A'} - R^{\SV}\rbra*{\frac{A}{\alpha}}} + \Abs*{R^{\SV}\rbra*{\frac{A}{\alpha}} - \frac{A}{\beta}}}.
    \]
    By \cref{lemma:real-qsp}, there exists a polynomial $S \in \mathbb{C}\sbra{x}$ of degree $d$ that satisfies the four conditions required in \cref{lemma:robust-qsvt} such that $S\rbra{x} = R\rbra{x}$ for all $x \in \mathbb{R}$. 
    Then, $S^{\SV}\rbra{A'} = R^{\SV}\rbra{A'}$ and $S^{\SV}\rbra{\frac{A}{\alpha}} = R^{\SV}\rbra{\frac{A}{\alpha}}$. 
    By \cref{lemma:robust-qsvt}, we further have
    \[
    \Abs*{R^{\SV}\rbra{A'} - R^{\SV}\rbra*{\frac{A}{\alpha}}} = \Abs*{S^{\SV}\rbra{A'} - S^{\SV}\rbra*{\frac{A}{\alpha}}} \leq 4d\sqrt{\Abs*{A' - \frac{A}{\alpha}}} \leq 4d\sqrt{\frac{\varepsilon}{\alpha}}.
    \]
    On the other hand, since $\Abs{A} \leq 1$, we have
    \[
        \Abs*{R^{\SV}\rbra*{\frac{A}{\alpha}} - \frac{A}{\beta}} \leq \max_{x \in \sbra{0, \frac{1}{\alpha}}} \abs*{R\rbra{x} - \frac{\alpha}{\beta} x} \leq \frac{\varepsilon'}{\beta}.
    \]
    Combining the above, we have
    \[
    \Abs{\beta R^{\SV}\rbra{A'} - A} \leq \beta \rbra*{4d\sqrt{\frac{\varepsilon}{\alpha}} + \frac{\varepsilon'}{\beta}} = \Theta\rbra*{\sqrt{\alpha\varepsilon} \log\rbra*{\frac{1}{\varepsilon'}}} + \varepsilon'.
    \]
    Therefore, $\widetilde U$ is a $\rbra{\beta, a+1, \Theta\rbra{\sqrt{\alpha\varepsilon}\log\rbra{\frac{1}{\varepsilon'}}}+\varepsilon'}$-block-encoding of $A$.

    At last, we consider the special case that $\varepsilon = 0$ and $\varepsilon' \leq 2\sqrt{2\beta^2+2\beta}-2\beta-2$. 
    Note that $A$ is an $N \times N$ matrix, where $N = 2^n$. Then, $\widetilde U$ is a $2^{a+1} N \times 2^{a+1}N$ unitary operator. 
    Let $I_2$ be the identity operator acting on $2$ qubits, which is a $4 \times 4$ matrix. 
    Then, $\widetilde U \otimes I_2$ is a $2^{a+3}N \times 2^{a+3}N$ unitary operator that is a $\rbra{\beta, a+3, \varepsilon'}$-block-encoding of $A$. 
    Assume that $\widetilde U \otimes I_2$ is a $\rbra{1, a+3, 0}$-block-encoding of an $N \times N$ matrix $B$, then $\Abs{\beta B - A} \leq \varepsilon'$, i.e.,
    \[
    \Abs*{B - \frac{A}{\beta}} \leq \frac{\varepsilon'}{\beta},
    \]
    which implies that
    \[
    \Abs*{B} \leq \Abs*{\frac{A}{\beta}} + \frac{\varepsilon'}{\beta} \leq \frac{1+\varepsilon'}{\beta}.
    \]
    Furthermore, we have
    \[
    \Abs*{\frac{B + A/\beta}{2}}^2
    \leq \frac{1}{4} \rbra*{ \Abs*{B} + \Abs*{\frac{A}{\beta}} }^2 
    \leq \frac{1}{4} \rbra*{ \frac{1+\varepsilon'}{\beta} + \frac{1}{\beta} }^2 = \rbra*{\frac{2+\varepsilon'}{2\beta}}^2,
    \]
    and thus 
    \[
    \Abs*{B - \frac{A}{\beta}} + \Abs*{\frac{B + A/\beta}{2}}^2 \leq \frac{\varepsilon'}{\beta} + \rbra*{\frac{2+\varepsilon'}{2\beta}}^2 \leq 1.
    \]
    Using \cref{lemma:robust-block-encoding} with $d \coloneqq 2^{a+3}N$, $n \coloneqq N$, $U \coloneqq \widetilde U \otimes I_2$, $A \coloneqq B$, and $\tilde A \coloneqq A/\beta$, there exists a $2^{a+3}N \times 2^{a+3}N$ unitary operator $U'$ that is a $\rbra{1, a+3, 0}$-block-encoding of $A/\beta$, i.e., a $\rbra{\beta, a+3, 0}$-block-encoding of $A$, satisfying
    \begin{align*}
        \Abs*{\widetilde U \otimes I_2 - U'}
        & \leq \sqrt{\frac{2}{1-\Abs*{\dfrac{B + A/\beta}{2}}^2}} \Abs*{B - \frac{A}{\beta}} \\
        & \leq \sqrt{\frac{2}{1-\rbra*{\dfrac{2+\varepsilon'}{2\beta}}^2}} \cdot \frac{\varepsilon'}{\beta} \\
        & = \frac{2\sqrt{2}\varepsilon'}{\sqrt{4\beta^2 - \rbra*{2+\varepsilon'}^2}}.
    \end{align*}
\end{proof}

\iffalse
\begin{lemma} [Polynomial approximation via local Taylor series, {\cite[Corollary 66 in the full version]{GSLW19}}] \label{lemma:poly-approx}
    Let $x_0 \in \sbra{-1, 1}$, $r \in (0, 2]$, $\delta \in (0, r]$ and let $f \colon \sbra{-x_0-r-\delta, x_0+r+\delta} \to \mathbb{C}$ be such that $f\rbra{x_0+x} = \sum_{j=0}^{\infty} a_j x^j$ for all $x \in \sbra{-r-\delta,r+\delta}$. 
    Suppose that $B > 0$ satisfies $\sum_{j=0}^{\infty} \rbra{r+\delta}^j \abs{a_j} \leq B$. 
    Let $\varepsilon \in (0, \frac{1}{2B}]$, then there is an efficiently computable polynomial $P \in \mathbb{C}\sbra{x}$ of degree $O\rbra{\frac{1}{\delta}\log\rbra{\frac{B}{\varepsilon}}}$ such that
    \begin{align*}
        \Abs{f\rbra{x} - P\rbra{x}}_{\sbra{x_0-r,x_0+r}} & \leq \varepsilon, \\
        \Abs{P\rbra{x}}_{\sbra{-1, 1}} & \leq B + \varepsilon, \\
        \Abs{P\rbra{x}}_{\sbra{-1, 1} \setminus \sbra{x_0-r-\delta/2,x_0+r+\delta/2}} & \leq \varepsilon,
    \end{align*}
    where $\Abs{g\rbra{x}}_{I} \coloneqq \sup_{x \in I} \abs{g\rbra{x}}$ for a set $I$ of real numbers.
\end{lemma}

Let $x_0 = 0$, $r > 0$, and $\delta > 0$ with $r + \delta = 1/\alpha$. 
Consider the function $f \colon \sbra{-x_0-r-\delta, x_0+r+\delta} \to \mathbb{C}$ defined by $f\rbra{x} = \alpha x$.
We simply take $B = 1$ and let $\varepsilon \in (0, 1/2]$. 
Then, by \cref{lemma:poly-approx}, there is a polynomial $P \in \mathbb{C}\sbra{x}$.
\fi

\section{Quantum Query Complexity with General Parameter} \label{sec:alpha-query}

For completeness, we discuss the relationships between $\alpha$-quantum query complexities $\mathsf{Q}_\diamond^\alpha\rbra{\mathcal{P}}$ for different values of constant $\alpha > 1$. 

\begin{theorem}
    For every constant $1 < \alpha < \beta$ and promise problem $\mathcal{P}$ for quantum states, we have $\mathsf{Q}_\diamond^{\alpha}\rbra{\mathcal{P}} = \widetilde \Theta\rbra{\mathsf{Q}_\diamond^{\beta}\rbra{\mathcal{P}}}$.
\end{theorem}
\begin{proof}
    This is implied by \cref{lemma:alpha-leq-beta} and \cref{lemma:beta-leq-alpha}.
\end{proof}

\begin{lemma} \label{lemma:alpha-leq-beta}
    For every constant $1 < \alpha < \beta$ and promise problem $\mathcal{P}$ for quantum states, we have  $\mathsf{Q}_\diamond^{\alpha}\rbra{\mathcal{P}} \leq \mathsf{Q}_\diamond^{\beta}\rbra{\mathcal{P}}$.
\end{lemma}
\begin{proof}
    Let $U_\alpha$ be a unitary operator that is an $\rbra{\alpha, a, 0}$-block-encoding of $\rho$, where $\rho$ is a mixed quantum state. 
    By \cref{corollary:scaling-block-encoding}, we can implement a unitary operator $U_\beta$ that is a $\rbra{\beta, a+1, 0}$-block-encoding of $\rho$, using $1$ query to $U_\alpha$.
    Therefore, any tester using $Q$ queries to $U_\beta$ can be converted to a tester using $Q$ queries to $U_\alpha$ with the same behavior with respect to the mixed state $\rho$ to be tested. 
    This implies that $\mathsf{Q}_\diamond^{\alpha}\rbra{\mathcal{P}} \leq \mathsf{Q}_\diamond^{\beta}\rbra{\mathcal{P}}$.
\end{proof}

\begin{lemma} \label{lemma:beta-leq-alpha}
    For every constant $1 < \alpha < \beta$ and promise problem $\mathcal{P}$ for quantum states, we have $\mathsf{Q}_\diamond^{\alpha}\rbra{\mathcal{P}} \leq O\rbra{\mathsf{Q}_\diamond^{\beta}\rbra{\mathcal{P}} \log\rbra{\mathsf{Q}_\diamond^{\beta}\rbra{\mathcal{P}}}}$.
\end{lemma}
\begin{proof}
    Let $U_\beta$ be a unitary operator that is a $\rbra{\beta, a, 0}$-block-encoding of $\rho$, where $\rho$ is a mixed quantum state. 
    By \cref{corollary:up-scaling-block-encoding}, we can implement a unitary operator $U_\alpha$ that is an $\rbra{\alpha, a+1, \varepsilon}$-block-encoding of $\rho$, using $O\rbra{\frac{\beta}{\alpha-1}\log\rbra{\frac{1}{\varepsilon}}}$ query to $U_\beta$, where $0 < \varepsilon < \min\cbra{\frac{\alpha}{\beta}, \frac{1}{2}}$ is to be determined.
    Therefore, any tester $\mathcal{T}_\alpha$ using $Q$ queries to $U_\alpha$ can be simulated by a tester $\mathcal{T}_\beta$ using $Q \cdot O\rbra{\frac{\beta}{\alpha-1}\log\rbra{\frac{1}{\varepsilon}}}$ queries to $U_\beta$ to precision $\Theta\rbra{Q\varepsilon}$. 
    It is sufficient to choose $\varepsilon = \Theta\rbra{1/Q}$ to make $\mathcal{T}_{\alpha}$ and $\mathcal{T}_{\beta}$ have similar enough behavior. 
    This implies that $\mathsf{Q}_\diamond^{\beta}\rbra{\mathcal{P}} \leq \mathsf{Q}_\diamond^{\alpha}\rbra{\mathcal{P}} \cdot O\rbra{\frac{\beta}{\alpha-1}\log\rbra{\mathsf{Q}_\diamond^{\alpha}\rbra{\mathcal{P}}}}$.
    When $\alpha$ and $\beta$ are constant, this gives $\mathsf{Q}_\diamond^{\beta}\rbra{\mathcal{P}} \leq O\rbra{\mathsf{Q}_\diamond^{\alpha}\rbra{\mathcal{P}}\log\rbra{\mathsf{Q}_\diamond^{\alpha}\rbra{\mathcal{P}}}}$.
\end{proof}

\section{Reducing Search to Quantum Gibbs Sampling}
\label{sec:gibbs-dimension}

\begin{theorem}
    For $\beta = \Omega\rbra{\log\rbra{N}}$, every quantum query algorithm for preparing the Gibbs state of an $N$-dimensional Hamiltonian $H$ at inverse temperature $\beta$, given an oracle that is a block-encoding of $H$, has query complexity $\Omega\rbra{\sqrt{N}}$. 
\end{theorem}

\begin{proof}
    We reduce the search problem to quantum Gibbs sampling. 
    Suppose that $\mathcal{O}$ is a quantum oracle for the search problem over $N = 2^n$ items $x_0, x_1, \dots, x_{N-1}$, defined by $\mathcal{O} \colon \ket{i}\ket{0} \mapsto \ket{i}\ket{x_i}$ for $0 \leq i < N$. 
    Assume that there is only one index $i^*$ with $x_{i^*} = 1$ and the rest items are $0$'s. 
    It is known in \cite{BBBV97,BBHT98,Zal99} that finding the index $i^*$ requires $\Omega\rbra{\sqrt{N}}$ queries to $\mathcal{O}$. 
    
    On the other hand, we note that $\mathcal{O}$ is a block-encoding of the Hamiltonian 
    \[
    H = I - \diag\rbra{x_0, x_1, \dots, x_{N-1}}. 
    \]
    To see this, let $A$ and $B$ denote the first and second registers that $\mathcal{O}$ acts on, respectively.
    Then, $\mathcal{O}$ is a $\rbra{1,1,0}$-block-encoding of $\bra{0}_B \mathcal{O} \ket{0}_B$. 
    For every $0 \leq i < N$, $\bra{0}_B \mathcal{O} \ket{0}_B \ket{i}_A = \ket{i}_A$ if $x_i = 0$, and $0$ (the zero vector) if $x_i = 1$, which implies that $\bra{0}_B \mathcal{O} \ket{0}_B$ is a diagonal matrix with the $i$-th diagonal element being $1 - x_i$, i.e., $\bra{0}_B \mathcal{O} \ket{0}_B=H$. 
    
    Let $\beta = \ceil{2 \ln\rbra{N}}$.
    The measure outcome of the Gibbs state of $H$ in the computational basis will be $i^*$ with probability
    \[
    \frac{e^{-\beta \rbra{1-x_{i^*}}}}{\sum_{i=0}^{N-1} e^{-\beta \rbra{1-x_i}}} = \frac{1}{1+\rbra{N-1}e^{-\beta}} \geq \frac{1}{1+\rbra{N-1}/N^2} \geq \frac 4 5.
    \]
    This means that we can find the solution $i^*$ with high probability by one application of any quantum Gibbs sampler.
    Therefore, any quantum Gibbs sampler at low temperature such that $\beta = \Omega\rbra{\log\rbra{N}}$ has query complexity $\Omega\rbra{\sqrt{N}}$.
\end{proof}

\end{document}